\documentclass{amsart}
\usepackage{mathrsfs}
\usepackage{graphicx}
\usepackage{amsmath}
\usepackage{amscd}
\usepackage{cite}
\usepackage{hyperref}
\usepackage{epsfig}
\usepackage{subfigure}
\usepackage{caption}
\usepackage{tikz}
  \usepackage{float}
  \usepackage{amssymb}
\usepackage{amsfonts}
\usepackage[all,cmtip]{xy}
\usepackage{appendix}

\newtheorem{Example}{Example}[section]
\newtheorem{Definition}{Definition}[section]
\newtheorem{Theorem}{Theorem}[section]
\newtheorem{Theorem/Definition}{Theorem/Definition}[section]

\newtheorem{Lemma}{Lemma}[section]
\newtheorem{Corollary}{Corollary}[section]

\newcommand{\pd}{\partial}
\newcommand{\bC}{{\mathbb C}}
 
\newcommand{\bN}{{\mathbb N}}
\newcommand{\bP}{{\mathbb P}}
\newcommand{\bQ}{{\mathbb Q}}
\newcommand{\bR}{{\mathbb R}}
\newcommand{\bZ}{{\mathbb Z}}

\newcommand{\cD}{{\mathcal D}}
\newcommand{\cE}{{\mathcal E}}
\newcommand{\cF}{{\mathcal F}}
\newcommand{\cG}{{\mathcal G}}

\newcommand{\cM}{{\mathcal M}}

\newcommand{\half}{\frac{1}{2}}
\newcommand{\cV}{{\mathcal V}}

\newcommand{\dbar}{\bar{\partial}}
\newcommand{\Mbar}{\overline{\cM}}

\newcommand{\wF}{{\widehat F}}
\newcommand{\wcF}{{\widehat{\mathcal{F}}}}

\newcommand{\be}{\begin{equation}}
\newcommand{\ee}{\end{equation}}
\newcommand{\bea}{\begin{eqnarray}}
\newcommand{\ben}{\begin{eqnarray*}}
\newcommand{\een}{\end{eqnarray*}}
\newcommand{\eea}{\end{eqnarray}}

\DeclareMathOperator{\Aut}{Aut}
\DeclareMathOperator{\ext}{ext}
\DeclareMathOperator{\val}{val}

\usepackage{color}

\definecolor{yellow}{rgb}{1,1,0}
\definecolor{orange}{rgb}{1,.7,0}
\definecolor{red}{rgb}{1,0,0} \definecolor{green}{rgb}{0,1,1}
\definecolor{white}{rgb}{1,1,1}

\definecolor{A}{rgb}{.75,1,.75}

\theoremstyle{remark}

\begin{document}

\newtheorem{myDef}{Definition}
\newtheorem{thm}{Theorem}
\newtheorem{eqn}{equation}

\title[Holomorphic Anomaly Equations and Quantum Spectral Curves]
{A Unified Approach to Holomorphic Anomaly Equations and Quantum Spectral Curves}

\author{Zhiyuan Wang}
\address{Department of Mathematical Sciences\\
Tsinghua University\\Beijing, 100084, China}
\email{zhiyuan-14@mails.tsinghua.edu.cn}

\author{Jian Zhou}
\address{Department of Mathematical Sciences\\
Tsinghua University\\Beijing, 100084, China}
\email{jzhou@math.tsinghua.edu.cn}

\begin{abstract}
We present a unified approach to holomorphic anomaly equations
and some well-known quantum spectral curves.
We develop a formalism of abstract quantum field theory based on
the diagrammatics of the Deligne-Mumford moduli spaces $\overline\cM_{g,n}$
and derive a quadratic recursion relation
for the abstract free energies in terms of the edge-cutting operators.
This abstract quantum field theory can be realized by various choices
of a sequence of holomorphic functions or formal power series and suitable propagators,
and the realized quantum field theory can be represented by formal Gaussian integrals.
Various applications are given.
\end{abstract}

\maketitle

\tableofcontents

\section{Introduction}

There are various types of quadratic recursions which have played an
important role in Gromov-Witten theory.
The Virasoro constraints are a sequence of recursion relations \cite{dvv, ehx, get} for the free energy
defined on the big phase space.
It involves differential operators of infinitely many variables.
There are also quadratic recursion relations for the free energy restricted to the small phase space
that involves only finitely many variables.
For example,
the BCOV holomorphic anomaly equation \cite{bcov2}
and the Eynard-Orantin topological recursion relations \cite{eo}.
The former was developed by
Bershadsky, Cecotti, Ooguri and Vafa
to compute Gromov-Witten invariants of the quintic Calabi-Yau threefold \cite{bcov1, bcov2}.
The latter was first discovered by Eynard and Orantin in the setting of matrix models
and later formulated in a form which is conjectured and proved to hold in general
for $n$-point functions.
For an approach to the holomorphic anaomaly equation from the point of view of Eynard-Orantin
topological relation,
see Eynard-Marin\~o-Orantin \cite{emo}.
In the case of Witten-Kontsevich tau-function,
it has been shown by the second author that the Virasoro constraints are equivalent to
the EO topological recursion relations \cite{zhou3}.
It is  an interesting problem to establish a relationship between the Virasoro constraints
and the EO topological relations.
The basic ingredient for EO topological recursion is an algebraic curve
with some extra data, called the spectral curve.
A suitable quantization of the  spectral curve leads to a Schr\"odinger equation
satisfied by a partition function constructed from the EO topological recursion,
 called the quantum spectral curve (cf. Gukov-Su{\l}kowski \cite{gs}).
It is a natural problem to understand the relationship between the holomorphic
anomaly equation and the quantum spectral curves.
In this work we will present a formalism that gives a unified construction
of the holomorphic anomaly equation and the wave functions for the quantum spectral curves.

To formulate the BCOV holomorphic anomaly equation,
it is crucial to first extend the free energy $F_g(t)$ to a non-holomorphic free energy $\widehat{F}_g(t, \bar{t})$
whose geometric meaning is not mathematically clear in the literature.
It is supposed to be modular-invariant,
and involve contributions from  the boundary strata of the moduli spaces.
Originally the holomorphic anomaly is an equation for
$\pd_{\bar{t}^i}\widehat{F}_g$ which involves $\wF_h$ for $h < g$ of the following form:
\ben
\dbar_{\bar{t}_i} \wF_g = \frac{1}{2}
\bar{C}^{(0)jk}_{\bar{i}} \biggl(D_jD_k \wF_{g-1}
+ \sum^{g-1}_{r=1}
D_j\wF_rD_k\wF_{g-r}\biggr).
\een
It is shown in \cite{bcov2} that this system of equations can be recursively solved  to get
\be
\wF_g(t, \bar{t}) = \Gamma^{(g)}(\Delta^{ij}, \Delta^i, \Delta, C^{(r<g)}_{i_1, \dots, i_n})
+ f^{(g)}(t),
\ee
where $\Gamma^{(g)}$ is a polynomial equation in
$\Delta^{ij}, \Delta^i, \Delta$ (propagators) and
$C^{(r<g)}_{i_1, \dots, i_n}$ (lower genus vertices).
Furthermore,
the authors of that work present
some explicit expressions of $\wF_g$ for $g=2$ and $3$
and the corresponding Feynman graphs and Feynman rules (cf. \cite[(6.7), Figure 17]{bcov2}
and  \cite[(6.8), Figure 18]{bcov2}).
Inspired by Witten \cite{wit1},
Aganagic, Bouchard and Klemm \cite[(2.16)]{abk} obtain
the following expression for $\wF_g$:
\be
\wF_g(t, \bar{t}) = F_g(t) + \Gamma_g \biggl(-\big((\tau - \bar{\tau})^{-1} \big)^{IJ},
\pd_{I_1} \cdots \pd_{I_n}F_{r<g}(t)\biggr),
\ee
where $\tau = (\tau_{ij}) = \big(\frac{\pd^2F_0}{\pd t_i\pd t_j})$.
This was interpreted for matrix models using Eynard-Orantin topological recursion
by Eynard, Mari\~no and Orantin \cite{emo},
and furthermore,
they also reformulate the partition function
as a formal Gaussian integral \cite[(4.27)]{emo} and
present the Feynman graphs and Feynman rules for the terms that contribute to $\Gamma_g$.
These results have been generalized to other models by Grimm, Klemm, Mari\~no and Weiss \cite{gkmw}.
These authors reformulate the holomorphic anomaly equation as a
quadratic recursion relation for the derivative of $\wF_g$ with respective to the propagators
$\Delta^{IJ}$ (cf. \cite[(7.50)]{gkmw}:
\be \label{eqn:HA}
\frac{\pd \wF_g}{\pd \Delta^{IJ}} =
\frac{1}{2}
D_I\pd_J\wF_{g-1} + \frac{1}{2}
\sum^{g-1}_{r=1}
\pd_I\wF_r\pd_J\wF_{g-r}.
\ee
Here the propagators $\Delta^{IJ}$ has the form
\be \label{eqn:Delta}
\Delta^{IJ} = -\frac{1}{2\sqrt{-1}} ((\tau-\bar{\tau})^{-1})^{IJ} +\cE^{IJ},
\ee
where $\cE^{IJ}$ is a holomorphic function.
They also derive the formal Gaussian integral representation of the partition function
and the Feynman expansions for $\wF_g$.

To summarize,
one can see the common features of the above works on holomorphic anomaly equations.
One starts with a sequence of functions $F_g(t)$ ($g\geq 0$) for finitely many variables $t_1, \dots, t_n$,
and seek for another sequence $\wF_g(t,\bar{t})$ such that
$\lim_{\bar{t}\to \infty} \wF_g(t,\bar{t}) = F_g(t)$,
and such that the sequence $\wF_g(t,\bar{t})$ satisfies  quadratic recursion equations
of the form \eqref{eqn:HA}.
Furthermore,
the partition function
$\widehat{Z}(t, \bar{t})= \exp \sum_{g\geq 0} \lambda^{2g-2}\wF_g(t, \bar{t})$
can be represented as a formal Gaussian integral,
and there are Feynman rules that express $\wF_g$ as a polynomial in the propagators
and the derivatives of $F_r$ for $r \leq g$.
As pointed out in \cite{bcov2},
$\widehat{F}(t, \bar{t})$ has contributions from the degenerate Riemann surfaces,
from lower boundary strata of the Deligne-Mumford moduli spaces.
It turns out that the Feynman graphs that appear in the Feynman expansion of
$\wF_g$ in the above works are just the stable graphs that index the stratification of
the Deligne-Mumford moduli spaces  $\Mbar_g$.

In the present work,
we will formulate a general construction of some kind of {\em abstract quantum field theory}
that has the above works on holomorphic anomaly equations
as special realizations.
The key ingredient in our construction is
the diagrammatics for the stratification of the Deligne-Mumford moduli spaces $\overline\cM_{g,n}$,
using the language of the dual graphs.
There are some natural operators acting on these graphs,
including an edge-cutting operator $K$.
We define an abstract free energy
\ben
\widehat{F}_g=\sum_{\Gamma\in\cG_{g,0}^c}\frac{1}{|\Aut(\Gamma)|}\Gamma,
\een
and derive a quadratic relation
\ben
K\wF_g=\frac{1}{2}(D\partial\wF_{g-1}+\sum_{r=1}^{g-1}\partial\wF_{r}\partial\wF_{g-r})
\een
for $g \geq 2$ using the operators $\partial$ and $\cD$ which correspond to adding external edges on the graphs.
For details,
see Section \ref{sec2}.
For generalizations that involves more labelling on the graphs,
see Section \ref{sec3}.

Another ingredient in our formalism is what we call the {\em realizations of the abstract quantum field theory},
we need a sequence of functions $F_g(t)$, $t=(t_1, \dots, t_n)$, $g \geq 0$,
and a symmetric nondegenerate matrix $\kappa$ that depends on $t$ and $\bar{t}$.
They will be used to give the Feynman rules that associate a weight $\omega_\Gamma$ to each $\Gamma$.
This kind of realization of the above abstract quantum field theory can also be
represented as a formal Gaussian integral.
We will refer to this as the {\em formal Gaussian integral representation
of the abstract quantum field theory}.
These will be presented in Section \ref{sec4}.

Using this formalism, given a holomorphic function $F(t)=\sum \lambda^{2g-2}F_g(t)$,
we will obtain a family of free energies $\wF_g(\kappa,t)$.
Here $\wF_g(\kappa,t)$ is a polynomial in $\kappa$ of degree $3g-3$,
whose coefficients are some differential polynomials in $\{F_k\}_{0\leq k\leq g}$ for $g\geq 2$.
In particular, $\wF_g(0,t)=F_g(t)$ for $g \geq 2$.

Let us now mention some applications of our formalism.
First of all,
for the choice $\kappa=(\bar{\tau}-\tau)^{-1}$,
we recover the holomorphic anomaly for matrix models \cite{emo,ey},
or more generally, when the propagator $\kappa^{IJ}$ takes the form of $\Delta^{IJ}$
in \eqref{eqn:Delta},
we recover the construction of holomorphic anomaly equation in Gromov-Witten theory
as discussed in \cite{gkmw}.
For details,
see Section \ref{sec5}.
Secondly,
the propagators $\kappa$ can be also chosen to be holomorphic.
We will use this formalism for holomorphic propagators to reinterpret some results
of the second author \cite{zhou2} in the theory of topological 1D gravity.

We also combine the results on graph sum representation of Eynard-Orantin
recursion \cite{ko, ey2, doss} with the  construction of the wave functions $Z(z)$ of
quantum spectral curves \cite{gs}:
\be
\widehat A Z(z) =0,
\ee
where $Z(z)$
\be
Z(z):=\exp\biggl(\sum_{n=0}^{\infty}\hbar^{n-1}S_n(z)\biggr),
\ee
and $S_n(z)$ are obtained from the spectral curve and Eynard-Orantin topological recursion \cite{gs}
(see also \eqref{eqn:GS}).
We reinterpret $S_n$ as  Feynman sum over labelled stable graphs,
hence relating them to a realization of  abstract QFT as developed in this work (cf. \eqref{eq-reln-S-F}).
We present some well-known examples of quantum spectral curves \cite{gs}
reexamined in our formalism.

In a subsequent work \cite{wz2} we have applied our formalism to the problem of computing the orbifold
Euler characteristics of $\Mbar_{g,n}$ based on the formulas for $\chi_{orb}(\cM_{g,n})$
of Harer-Zagier \cite{hz} and Penner \cite{pen}.
We will deal with the case of propagators of the form \eqref{eqn:Delta}
in another subsequent work that generalizes the formalism of this work \cite{wz3}.

The rest of this paper is arranged as follows.
In Section \ref{sec2} we formulate the diagrammatics and
derive a quadratic recursive relation for the free energies using the edge-cutting operators.
In Section \ref{sec3} we generalize the diagrammatics to the case of labelled graphs.
We give a field theory realization of the diagrammatics in Section \ref{sec4},
where the free energies can be obtained from formal Gaussian integrals and Feynman rules.
The holomorphic anomaly equation is derived in Section \ref{sec5} in our formalism
using non-holomorphic propagators.
In Section \ref{sec6} we write the wave function of the quantum spectral curve
of Gukov-Su{\l}kowski as a summation over stable graphs,
and understand it as a particular realization of the abstract QFT.
In Section \ref{sec7} we apply our formalism to topological 1D gravity
and derive quadratic recursion relations for this theory.
As applications,
we relate the recursion relations to quantum spectral curves
in some important examples.
Other applications include recursion relations in enumeration problems
of stable graphs and general graphs.

\section{Diagrammatics of Moduli Space of Curves and Abstract QFT}\label{sec2}

In this section we recall the diagrammatics related to the stratification of
the Deligne-Mumford moduli spaces $\overline{\cM}_{g,n}$.
We introduce a notion of abstract quantum field theory based on stable graphs.
We define the edge-cutting operators and some other operators
acting on the graphs, and use them to derive a recursive relation of the abstract free energy.

\subsection{Moduli space of stable curves and Feynman diagrams}

We first recall the stratification of the Deligne-Mumford moduli space $\overline\cM_{g,n}$ of stable curves \cite{dm, kn}.

The Deligne-Mumford moduli space $\overline\cM_{g,n}$ is a smooth orbifold of complex dimension $3g-3+n$. There are two types of basic maps over these spaces, namely the forgetful map
\ben
\pi_{n+1}:\overline\cM_{g,n+1}\to\overline\cM_{g,n},
\een
which simply forgets the $(n+1)$-th marked point and contract the unstable component; and the gluing maps
\ben
\xi_1:\overline\cM_{g_1,n_1+1}\times\overline\cM_{g_2,n_2+1}\to\overline\cM_{g_1+g_2,n_1+n_2},
\een
\ben
\xi_2:\overline\cM_{g-1,n+2}\to\overline\cM_{g,n},
\een
which glues two marked points together to produce a new nodal point.
The union of images of $\xi_1$ and $\xi_2$ is the boundary $\overline\cM_{g,n}\backslash\cM_{g,n}$;
moreover, it is of codimension one in $\overline\cM_{g,n}$.
In this way,
we get a stratification of the compactified moduli space $\overline\cM_{g,n}$ by decomposing
it into the union of some products of the moduli spaces $\cM_{h,m}$ of smooth stable curves.

An efficient way of describing this stratification is to use the language of graphs.
Recall that for a stable curve $(C;x_1,\cdots,x_n)$ of genus g with n marked points,
its dual graph is defined as follows. Let $\tilde{C}$ be its normalization,
then we associate a vertex to each of the connected components of $\tilde{C}$,
and write down the genus of this component at this vertex.
For each nodal point of $C$, we draw an internal edge connecting the corresponding vertices;
and for each marked point, we draw an external edge attaching to the corresponding vertex,
and write the markings $i\in\{1,2,\cdots,n\}$ besides the external edges.

In what follows we will forget all the markings of the external edges,
which corresponds to the moduli space $\overline\cM_{g,n}/S_n$ parametrizing
equivalent classes of stable curves with $n$ marked points while we ignore the ordering of these marked points.
In this way, we obtain a graph of genus $g$ with $n$ external edges,
where the genus is defined to be the sum of the number of loops in the graphs
and all the numbers associated to the vertices.

\begin{Example}\label{eg-graph}
All possible dual graphs of curves in $\overline{\cM}_{1,1}$, and $\overline{\cM}_{2,0}$
are listed in Figure \ref{m11} and Figure \ref{m20} respectively.

\begin{figure}[H]
\begin{tikzpicture}
\draw (1,0) circle [radius=0.2];
\draw (1.2,0)--(1.5,0);
\node [align=center,align=center] at (1,0) {$1$};
\end{tikzpicture}
\quad
\begin{tikzpicture}
\draw (1,0) circle [radius=0.2];
\draw (1.2,0)--(1.5,0);
\draw (0.84,0.1) .. controls (0.5,0.2) and (0.5,-0.2) ..  (0.84,-0.1);
\node [align=center,align=center] at (1,0) {$0$};
\end{tikzpicture}
\caption{Dual graphs for $\overline{\cM}_{1,1}$}
\label{m11}
\end{figure}
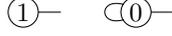

\begin{figure}[H]
\begin{tikzpicture}
\draw (1,0) circle [radius=0.2];
\node [align=center,align=center] at (1,0) {$2$};
\end{tikzpicture}
\quad
\begin{tikzpicture}
\draw (1,0) circle [radius=0.2];
\draw (0.84,0.1) .. controls (0.5,0.2) and (0.5,-0.2) ..  (0.84,-0.1);
\node [align=center,align=center] at (1,0) {$1$};
\end{tikzpicture}
\quad
\begin{tikzpicture}
\draw (1,0) circle [radius=0.2];
\draw (1.2,0)--(1.4,0);
\draw (1.6,0) circle [radius=0.2];
\node [align=center,align=center] at (1,0) {$1$};
\node [align=center,align=center] at (1.6,0) {$1$};
\end{tikzpicture}
\quad
\begin{tikzpicture}
\draw (1,0) circle [radius=0.2];
\draw (0.84,0.1) .. controls (0.5,0.2) and (0.5,-0.2) ..  (0.84,-0.1);
\draw (1.16,0.1) .. controls (1.5,0.2) and (1.5,-0.2) ..  (1.16,-0.1);
\node [align=center,align=center] at (1,0) {$0$};
\end{tikzpicture}
\quad
\begin{tikzpicture}
\draw (1,0) circle [radius=0.2];
\draw (0.4,0) circle [radius=0.2];
\draw (0.6,0)--(0.8,0);
\draw (1.16,0.1) .. controls (1.5,0.2) and (1.5,-0.2) ..  (1.16,-0.1);
\node [align=center,align=center] at (1,0) {$0$};
\node [align=center,align=center] at (0.4,0) {$1$};
\end{tikzpicture}
\quad
\begin{tikzpicture}
\draw (1,0) circle [radius=0.2];
\draw (0.4,0) circle [radius=0.2];
\draw (0.6,0)--(0.8,0);
\draw (1.16,0.1) .. controls (1.5,0.2) and (1.5,-0.2) ..  (1.16,-0.1);
\draw (0.24,0.1) .. controls (-0.1,0.2) and (-0.1,-0.2) ..  (0.24,-0.1);
\node [align=center,align=center] at (1,0) {$0$};
\node [align=center,align=center] at (0.4,0) {$0$};
\end{tikzpicture}
\quad
\begin{tikzpicture}
\draw (1,0) circle [radius=0.2];
\draw (1.2,0)--(1.4,0);
\draw (1.16,0.1)--(1.44,0.1);
\draw (1.16,-0.1)--(1.44,-0.1);
\draw (1.6,0) circle [radius=0.2];
\node [align=center,align=center] at (1,0) {$0$};
\node [align=center,align=center] at (1.6,0) {$0$};
\end{tikzpicture}
\caption{Dual graphs for $\overline{\cM}_{2,0}$}
\label{m20}
\end{figure}
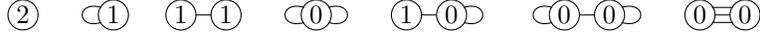
\end{Example}

A graph $\Gamma$ is called stable if all its vertices are stable, that is, the valence of each vertex of genus $0$ is at least three, and the valence of each vertex of genus $1$ is at least one. It is clear that the stability of curves is equivalent to the stability of its dual graph.
The stable graphs will be the Feynman graphs used in this paper.

Let $\cG_{g,n}$ be the set of all stable graphs of genus $g$ with $n$ external edges (not necessarily connected),
and $\cG_{g,n}^c$ be the subset of connected stable graphs.

Now let $\Gamma$ be a stable graph of genus $g$ with $n$ external edges,
and define $\cM_\Gamma$ to be the subset of $\overline\cM_{g,n}/S_n$ consisting of all equivalent classes of curves whose dual graph is $\Gamma$, then
\be\label{eq-stra}
\overline\cM_{g,n}/S_n=\bigsqcup_{\Gamma\in\cG_{g,n}^c}\cM_\Gamma
\ee
gives the stratification of $\overline\cM_{g,n}/S_n$.

\subsection{Operators on vector space associated to the set of stable graphs}

\label{sec:Operators-1}

Denote by $\cV^c$ the vector space over $\bQ$ generated by all {\em connected}
stable graphs,
i.e.,
\be
\cV^c = \bigoplus_{\substack{\Gamma \in \cG^c_{g,n}\\ 2g-2+n > 0}} \bQ\Gamma;
\ee
similarly,
denote by $\cV$ the vector space over $\bQ$ generated by all  stable graphs,
not necessarily connected,
i.e.,
\be
\cV = \bigoplus_{\substack{\Gamma \in \cG_{g,n}\\ 2g-2+n > 0}} \bQ\Gamma.
\ee

In this subsection we define some operators on $\cV$.
These operators may be understood as the inverse procedures of the gluing maps and forgetful maps.

 Define an operator $K$ acting on $\cV$:
\be
K:\cV \to \cV, \;\;\;\Gamma \in \cG_{g,n}\mapsto
K(\Gamma)= \sum_{\Gamma'} \Gamma',
\ee
where the summation is over all graphs $\Gamma'$ obtained by
cutting an internal edge of $\Gamma$.
We will refer to this operator as the {\em edge cutting operator}.
It is easy to see that the operator $K$ does not affect the stability of the graph.
It decreases the number of internal edges by one, and increases the number of external edges by two.
In the picture of stable curves, this operator corresponds to breaking up a node and regard
it as two new marked points.
So the operator $K$ can be regarded as the inverse operator of the gluing maps $\xi_1$ and $\xi_2$.

\begin{Example}
\begin{flalign*}
\begin{split}
&\begin{tikzpicture}
\node [align=center,align=center] at (0.4,0) {$K$};
\draw (1,0) circle [radius=0.2];
\node [align=center,align=center] at (1,0) {$2$};
\node [align=center,align=center] at (1.6,0) {$=0$};
\end{tikzpicture},
\\
&\begin{tikzpicture}
\node [align=center,align=center] at (-0.4,0) {$K$};
\draw (1,0) circle [radius=0.2];
\draw (0.4,0) circle [radius=0.2];
\draw (0.6,0)--(0.8,0);
\draw (1.16,0.1) .. controls (1.5,0.2) and (1.5,-0.2) ..  (1.16,-0.1);
\draw (0.24,0.1) .. controls (-0.1,0.2) and (-0.1,-0.2) ..  (0.24,-0.1);
\node [align=center,align=center] at (1,0) {$0$};
\node [align=center,align=center] at (0.4,0) {$0$};
\node [align=center,align=center] at (1.8,0) {$=2$};
\draw (3.3,0) circle [radius=0.2];
\draw (2.7,0) circle [radius=0.2];
\draw (2.9,0)--(3.1,0);
\draw (3.46,0.1)--(3.7,0.15);
\draw (3.46,-0.1)--(3.7,-0.15);
\draw (2.54,0.1) .. controls (2.2,0.2) and (2.2,-0.2) ..  (2.54,-0.1);
\node [align=center,align=center] at (3.3,0) {$0$};
\node [align=center,align=center] at (2.7,0) {$0$};
\node [align=center,align=center] at (4,0) {$+$};
\draw (5.7,0) circle [radius=0.2];
\draw (4.7,0) circle [radius=0.2];
\draw (4.9,0)--(5.1,0);
\draw (5.3,0)--(5.5,0);
\draw (5.86,0.1) .. controls (6.2,0.2) and (6.2,-0.2) ..  (5.86,-0.1);
\draw (4.54,0.1) .. controls (4.2,0.2) and (4.2,-0.2) ..  (4.54,-0.1);
\node [align=center,align=center] at (5.7,0) {$0$};
\node [align=center,align=center] at (4.7,0) {$0$};
\end{tikzpicture},
\\
&\begin{tikzpicture}
\node [align=center,align=center] at (0.4,0) {$K$};
\draw (1,0) circle [radius=0.2];
\draw (1.2,0)--(1.4,0);
\draw (1.16,0.1)--(1.44,0.1);
\draw (1.16,-0.1)--(1.44,-0.1);
\draw (1.6,0) circle [radius=0.2];
\node [align=center,align=center] at (1,0) {$0$};
\node [align=center,align=center] at (1.6,0) {$0$};
\node [align=center,align=center] at (2.2,0) {$=3$};
\draw (3.1,0) circle [radius=0.2];
\draw (3.28,0.07)--(3.52,0.07);
\draw (3.28,-0.07)--(3.52,-0.07);
\draw (2.6,0)--(2.9,0);
\draw (3.9,0)--(4.2,0);
\draw (3.7,0) circle [radius=0.2];
\node [align=center,align=center] at (3.1,0) {$0$};
\node [align=center,align=center] at (3.7,0) {$0$};
\end{tikzpicture}.
\end{split}
\end{flalign*}

\end{Example}

Let us define another operator $\pd$.
This operator has two parts,
one is to attach an external edge to a vertex and sum over all vertices,
and the other is to break up an internal edge and insert a stable vertex of genus $0$ with three edges.

\begin{Example}
\begin{flalign*}
\begin{split}
&\begin{tikzpicture}
\node [align=center,align=center] at (0.4,0) {$\pd$};
\draw (1,0) circle [radius=0.2];
\draw (1.2,0)--(1.5,0);
\node [align=center,align=center] at (1,0) {$1$};
\node [align=center,align=center] at (1.9,0) {$=$};
\draw (2.8,0) circle [radius=0.2];
\draw (3,0)--(3.3,0);
\draw (2.3,0)--(2.6,0);
\node [align=center,align=center] at (2.8,0) {$1$};
\end{tikzpicture},
\\
&\begin{tikzpicture}
\node [align=center,align=center] at (-0.4,0) {$\pd$};
\draw (1,0) circle [radius=0.2];
\draw (0.4,0) circle [radius=0.2];
\draw (0.6,0)--(0.8,0);
\draw (1.16,0.1) .. controls (1.5,0.2) and (1.5,-0.2) ..  (1.16,-0.1);
\draw (0.24,0.1) .. controls (-0.1,0.2) and (-0.1,-0.2) ..  (0.24,-0.1);
\node [align=center,align=center] at (1,0) {$0$};
\node [align=center,align=center] at (0.4,0) {$0$};
\node [align=center,align=center] at (1.8,0) {$=2$};
\draw (3.2,0) circle [radius=0.2];
\draw (2.6,0) circle [radius=0.2];
\draw (2.8,0)--(3,0);
\draw (2.6,0.2)--(2.6,0.4);
\draw (3.36,0.1) .. controls (3.7,0.2) and (3.7,-0.2) ..  (3.36,-0.1);
\draw (2.44,0.1) .. controls (2.1,0.2) and (2.1,-0.2) ..  (2.44,-0.1);
\node [align=center,align=center] at (2.6,0) {$0$};
\node [align=center,align=center] at (3.2,0) {$0$};
\node [align=center,align=center] at (4,0) {$+2$};
\draw (5.4,0) circle [radius=0.2];
\draw (4.8,0) circle [radius=0.2];
\draw (5,0)--(5.2,0);
\draw (4.64,0.1) .. controls (4.3,0.2) and (4.3,-0.2) ..  (4.64,-0.1);
\draw (5.58,0.07)--(5.82,0.07);
\draw (5.58,-0.07)--(5.82,-0.07);
\draw (6.2,-0)--(6.5,0);
\draw (6,0) circle [radius=0.2];
\node [align=center,align=center] at (5.4,0) {$0$};
\node [align=center,align=center] at (4.8,0) {$0$};
\node [align=center,align=center] at (6,0) {$0$};
\node [align=center,align=center] at (7,0) {$+$};
\draw (8.4,0) circle [radius=0.2];
\draw (7.8,0) circle [radius=0.2];
\draw (9,0) circle [radius=0.2];
\draw (8,0)--(8.2,0);
\draw (8.6,0)--(8.8,0);
\draw (8.4,0.2)--(8.4,0.4);
\draw (9.16,0.1) .. controls (9.5,0.2) and (9.5,-0.2) ..  (9.16,-0.1);
\draw (7.64,0.1) .. controls (7.3,0.2) and (7.3,-0.2) ..  (7.64,-0.1);
\node [align=center,align=center] at (8.4,0) {$0$};
\node [align=center,align=center] at (7.8,0) {$0$};
\node [align=center,align=center] at (9,0) {$0$};
\end{tikzpicture},
\\
&\begin{tikzpicture}
\node [align=center,align=center] at (0.4,0) {$\pd$};
\draw (1,0) circle [radius=0.2];
\draw (1.2,0)--(1.4,0);
\draw (1.16,0.1)--(1.44,0.1);
\draw (1.16,-0.1)--(1.44,-0.1);
\draw (1.6,0) circle [radius=0.2];
\node [align=center,align=center] at (1,0) {$0$};
\node [align=center,align=center] at (1.6,0) {$1$};
\node [align=center,align=center] at (2.1,0) {$=$};
\draw (3,0) circle [radius=0.2];
\draw (3.2,0)--(3.4,0);
\draw (2.5,0)--(2.8,0);
\draw (3.16,0.1)--(3.44,0.1);
\draw (3.16,-0.1)--(3.44,-0.1);
\draw (3.6,0) circle [radius=0.2];
\node [align=center,align=center] at (3,0) {$0$};
\node [align=center,align=center] at (3.6,0) {$1$};
\node [align=center,align=center] at (4.1,0) {$+$};
\draw (4.7,0) circle [radius=0.2];
\draw (4.9,0)--(5.1,0);
\draw (5.5,0)--(5.8,0);
\draw (4.86,0.1)--(5.14,0.1);
\draw (4.86,-0.1)--(5.14,-0.1);
\draw (5.3,0) circle [radius=0.2];
\node [align=center,align=center] at (4.7,0) {$0$};
\node [align=center,align=center] at (5.3,0) {$1$};
\node [align=center,align=center] at (6.2,0) {$+3$};
\draw (6.8,-0.2) circle [radius=0.2];
\draw (6.98,-0.13)--(7.42,-0.13);
\draw (6.98,-0.27)--(7.42,-0.27);
\draw (7.6,-0.2) circle [radius=0.2];
\draw (7.2,0.25) circle [radius=0.2];
\draw (6.94,-0.06)--(7.06,0.11);
\draw (7.46,-0.06)--(7.34,0.11);
\node [align=center,align=center] at (6.8,-0.2) {$0$};
\node [align=center,align=center] at (7.6,-0.2) {$1$};
\node [align=center,align=center] at (7.2,0.25) {$0$};
\draw (7.2,0.45)--(7.2,0.65);
\end{tikzpicture}.
\end{split}
\end{flalign*}
\end{Example}

In the dual picture,
the operator $\pd$ is just adding one marked point,
or breaking up a node of the curve and gluing a $3$-pointed sphere.

Now we define a third operator $\gamma$,
which acts on $\Gamma$ by simply attaching a stable vertex of genus $0$
with three edges to an external edge of $\Gamma$,
and summing over all external edges.

\begin{Example}
\begin{flalign*}
\begin{split}
&\begin{tikzpicture}
\node [align=center,align=center] at (0.5,0) {$\gamma$};
\draw (1,0) circle [radius=0.2];
\draw (1.2,0)--(1.5,0);
\node [align=center,align=center] at (1,0) {$1$};
\node [align=center,align=center] at (1.9,0) {$=$};
\draw (2.4,0) circle [radius=0.2];
\draw (3,0) circle [radius=0.2];
\draw (2.6,0)--(2.8,0);
\draw (3.16,0.1)--(3.4,0.15);
\draw (3.16,-0.1)--(3.4,-0.15);
\node [align=center,align=center] at (2.4,0) {$1$};
\node [align=center,align=center] at (3,0) {$0$};
\end{tikzpicture},
\\
&\begin{tikzpicture}
\node [align=center,align=center] at (0.2,0) {$\gamma$};
\draw (0.9,0) circle [radius=0.2];
\draw (1.06,0.1)--(1.3,0.15);
\draw (1.06,-0.1)--(1.3,-0.15);
\draw (0.74,0.1) .. controls (0.4,0.2) and (0.4,-0.2) ..  (0.74,-0.1);
\node [align=center,align=center] at (0.9,0) {$0$};
\node [align=center,align=center] at (1.8,0) {$=2$};
\draw (2.7,0) circle [radius=0.2];
\draw (3.3,0) circle [radius=0.2];
\draw (2.9,0)--(3.1,0);
\draw (2.7,0.2)--(2.7,0.4);
\draw (3.47,0.1)--(3.7,0.15);
\draw (3.47,-0.1)--(3.7,-0.15);
\draw (2.54,0.1) .. controls (2.2,0.2) and (2.2,-0.2) ..  (2.54,-0.1);
\node [align=center,align=center] at (2.7,0) {$0$};
\node [align=center,align=center] at (3.3,0) {$0$};
\end{tikzpicture}.
\end{split}
\end{flalign*}
\end{Example}

If the graph $\Gamma$ has no external edge, then we may simply have $\gamma(\Gamma)=0$. This operator is to glue a $3$-pointed sphere to the curve along marked points to get a new node.

Clearly the operators $\pd$ and $\gamma$ also preserves the stability,
and they both increase the number of external edges by one.
Intuitively, $\pd+\gamma$ is the inverse of  the forgetful map.
Write $\cD=\pd+\gamma$. Then $\cD$ preserves the subspaces $\cV^c$.

\subsection{Abstract free energy}
\label{sec:Abstract-1}

Let us now define an {\em abstract quantum field theory}
based on the diagrammatics of stable graphs discussed above.

\begin{Definition}
For $g\geq 2$, we define the {\em abstract free energy} of genus $g$ to be
\be
\widehat{\cF}_g=\sum_{\Gamma\in\cG_{g,0}^c}\frac{1}{|\Aut(\Gamma)|}\Gamma.
\ee
In general, for $2g-2+n>0$, define the abstract $n$-point function of genus $g$ to be
\be
\widehat{\cF}_{g,n}=\sum_{\Gamma\in\cG_{g,n}^c}\frac{1}{|\Aut(\Gamma)|}\Gamma,
\ee
then the free energy $\widehat{\cF}_g$ is just $\widehat{\cF}_{g,0}$.
\end{Definition}

\begin{Example}\label{eg-free energy}
Here we list some explicit expressions of $\widehat{\cF}_{g,n}$ for small $(g,n)$.

\begin{flalign*}
\begin{tikzpicture}
\node [align=center,align=center] at (0.3,0) {$\widehat{\cF}_{0,3}=\frac{1}{6}$};
\draw (1.6,0) circle [radius=0.2];
\draw (1.1,0)--(1.4,0);
\draw (1.76,0.1)--(2,0.15);
\draw (1.76,-0.1)--(2,-0.15);
\node [align=center,align=center] at (1.6,0) {$0$};
\end{tikzpicture},&&
\end{flalign*}

\begin{flalign*}
\begin{tikzpicture}
\node [align=center,align=center] at (0.3,0) {$\widehat{\cF}_{0,4}=\frac{1}{24}$};
\draw (1.6,0) circle [radius=0.2];
\draw (1.1,0.15)--(1.44,0.1);
\draw (1.1,-0.15)--(1.44,-0.1);
\draw (1.76,0.1)--(2.1,0.15);
\draw (1.76,-0.1)--(2.1,-0.15);
\node [align=center,align=center] at (1.6,0) {$0$};
\node [align=center,align=center] at (2.6,0) {$+\frac{1}{8}$};
\draw (1.6+1.9,0) circle [radius=0.2];
\draw (2.2+1.9,0) circle [radius=0.2];
\draw (1.1+1.9,0.15)--(1.44+1.9,0.1);
\draw (1.1+1.9,-0.15)--(1.44+1.9,-0.1);
\draw (1.1+1.9,0.15)--(1.44+1.9,0.1);
\draw (2.36+1.9,0.1)--(2.7+1.9,0.15);
\draw (2.36+1.9,-0.1)--(2.7+1.9,-0.15);
\draw (1.8+1.9,0)--(2+1.9,0);
\node [align=center,align=center] at (1.6+1.9,0) {$0$};
\node [align=center,align=center] at (2.2+1.9,0) {$0$};
\end{tikzpicture},&&
\end{flalign*}

\begin{flalign*}
\begin{tikzpicture}
\node [align=center,align=center] at (0.3-0.1,0) {$\widehat{\cF}_{0,5}=\frac{1}{120}$};
\draw (1.6,0) circle [radius=0.2];
\draw (1.1,0.15)--(1.44,0.1);
\draw (1.1,-0.15)--(1.44,-0.1);
\draw (1.76,0.1)--(2.1,0.15);
\draw (1.76,-0.1)--(2.1,-0.15);
\draw (1.6,0.2)--(1.6,0.4);
\node [align=center,align=center] at (1.6,0) {$0$};
\node [align=center,align=center] at (2.6,0) {$+\frac{1}{8}$};
\draw (1.6+1.9,0) circle [radius=0.2];
\draw (2.2+1.9,0) circle [radius=0.2];
\draw (2.8+1.9,0) circle [radius=0.2];
\draw (1.1+1.9,0.15)--(1.44+1.9,0.1);
\draw (1.1+1.9,-0.15)--(1.44+1.9,-0.1);
\draw (1.1+1.9,0.15)--(1.44+1.9,0.1);
\draw (2.96+1.9,0.1)--(3.3+1.9,0.15);
\draw (2.96+1.9,-0.1)--(3.3+1.9,-0.15);
\draw (1.8+1.9,0)--(2+1.9,0);
\draw (2.4+1.9,0)--(2.6+1.9,0);
\draw (2.2+1.9,0.2)--(2.2+1.9,0.4);
\node [align=center,align=center] at (1.6+1.9,0) {$0$};
\node [align=center,align=center] at (2.2+1.9,0) {$0$};
\node [align=center,align=center] at (2.8+1.9,0) {$0$};
\node [align=center,align=center] at (5.8,0) {$+\frac{1}{12}$};
\draw (1.6+5.1,0) circle [radius=0.2];
\draw (2.2+5.1,0) circle [radius=0.2];
\draw (1.1+5.1,0.15)--(1.44+5.1,0.1);
\draw (1.1+5.1,-0.15)--(1.44+5.1,-0.1);
\draw (1.1+5.1,0.15)--(1.44+5.1,0.1);
\draw (2.36+5.1,0.1)--(2.7+5.1,0.15);
\draw (2.36+5.1,-0.1)--(2.7+5.1,-0.15);
\draw (1.8+5.1,0)--(2+5.1,0);
\draw (2.4+5.1,0)--(2.7+5.1,0);
\node [align=center,align=center] at (1.6+5.1,0) {$0$};
\node [align=center,align=center] at (2.2+5.1,0) {$0$};
\end{tikzpicture},&&
\end{flalign*}

\begin{flalign*}
\begin{tikzpicture}
\node [align=center,align=center] at (0.3-0.2,0) {$\widehat{\cF}_{1,1}=$};
\draw (1,0) circle [radius=0.2];
\draw (1.2,0)--(1.5,0);
\node [align=center,align=center] at (1,0) {$1$};
\node [align=center,align=center] at (2,0) {$+\frac{1}{2}$};
\draw (1+1.8,0) circle [radius=0.2];
\draw (1.2+1.8,0)--(1.5+1.8,0);
\draw (0.84+1.8,0.1) .. controls (0.5+1.8,0.2) and (0.5+1.8,-0.2) ..  (0.84+1.8,-0.1);
\node [align=center,align=center] at (1+1.8,0) {$0$};
\end{tikzpicture},&&
\end{flalign*}

\begin{flalign*}
\begin{tikzpicture}
\node [align=center,align=center] at (0.3-0.6,0) {$\widehat{\cF}_{1,2}=\frac{1}{2}$};
\draw (1,0) circle [radius=0.2];
\draw (1.2,0)--(1.5,0);
\draw (0.5,0)--(0.8,0);
\node [align=center,align=center] at (1,0) {$1$};
\node [align=center,align=center] at (2,0) {$+\frac{1}{4}$};
\draw (1+1.8,0) circle [radius=0.2];
\draw (1.17+1.8,0.1)--(1.4+1.8,0.15);
\draw (1.17+1.8,-0.1)--(1.4+1.8,-0.15);
\draw (0.84+1.8,0.1) .. controls (0.5+1.8,0.2) and (0.5+1.8,-0.2) ..  (0.84+1.8,-0.1);
\node [align=center,align=center] at (1+1.8,0) {$0$};
\node [align=center,align=center] at (3.7,0) {$+\frac{1}{2}$};
\draw (1+3.9,0) circle [radius=0.2];
\draw (0.4+3.9,0) circle [radius=0.2];
\draw (1.17+3.9,0.1)--(1.4+3.9,0.15);
\draw (1.17+3.9,-0.1)--(1.4+3.9,-0.15);
\draw (0.6+3.9,0)--(0.8+3.9,0);
\node [align=center,align=center] at (1+3.9,0) {$0$};
\node [align=center,align=center] at (0.4+3.9,0) {$1$};
\node [align=center,align=center] at (5.8,0) {$+\frac{1}{4}$};
\draw (1+6.2,0) circle [radius=0.2];
\draw (0.4+6.2,0) circle [radius=0.2];
\draw (1.17+6.2,0.1)--(1.4+6.2,0.15);
\draw (1.17+6.2,-0.1)--(1.4+6.2,-0.15);
\draw (0.6+6.2,0)--(0.8+6.2,0);
\node [align=center,align=center] at (1+6.3-0.1,0) {$0$};
\node [align=center,align=center] at (0.4+6.3-0.1,0) {$0$};
\draw (0.24+6.3-0.1,0.1) .. controls (-0.1+6.3-0.1,0.2) and (-0.1+6.3-0.1,-0.2) ..  (0.24+6.3-0.1,-0.1);
\node [align=center,align=center] at (8.2-0.1,0) {$+\frac{1}{4}$};
\draw (1+8.1-0.1,0) circle [radius=0.2];
\draw (0.5+8.1-0.1,0)--(0.8+8.1-0.1,0);
\draw (1.18+8.1-0.1,0.07)--(1.42+8.1-0.1,0.07);
\draw (1.18+8.1-0.1,-0.07)--(1.42+8.1-0.1,-0.07);
\draw (1.8+8.1-0.1,-0)--(2.1+8.1-0.1,0);
\draw (1.6+8.1-0.1,0) circle [radius=0.2];
\node [align=center,align=center] at (1+8.1-0.1,0) {$0$};
\node [align=center,align=center] at (1.6+8.1-0.1,0) {$0$};
\end{tikzpicture},&&
\end{flalign*}

\begin{flalign*}
\begin{split}
&\begin{tikzpicture}
\node [align=center,align=center] at (0.3,0) {$\widehat{\cF}_{1,3}=\frac{1}{6}$};
\draw (1.6,0) circle [radius=0.2];
\draw (1.1,0)--(1.4,0);
\draw (1.76,0.1)--(2,0.15);
\draw (1.76,-0.1)--(2,-0.15);
\node [align=center,align=center] at (1.6,0) {$1$};
\node [align=center,align=center] at (2.6,0) {$+\frac{1}{12}$};
\draw (1+2.5,0) circle [radius=0.2];
\draw (1.17+2.5,0.1)--(1.4+2.5,0.15);
\draw (1.17+2.5,-0.1)--(1.4+2.5,-0.15);
\draw (1.2+2.5,0)--(1.5+2.5,0);
\draw (0.84+2.5,0.1) .. controls (0.5+2.5,0.2) and (0.5+2.5,-0.2) ..  (0.84+2.5,-0.1);
\node [align=center,align=center] at (1+2.5,0) {$0$};
\node [align=center,align=center] at (4.5,0) {$+\frac{1}{2}$};
\draw (1+5,0) circle [radius=0.2];
\draw (0.4+5,0) circle [radius=0.2];
\draw (1.17+5,0.1)--(1.4+5,0.15);
\draw (1.17+5,-0.1)--(1.4+5,-0.15);
\draw (-0.1+5,0)--(0.2+5,0);
\draw (0.6+5,0)--(0.8+5,0);
\node [align=center,align=center] at (1+5,0) {$0$};
\node [align=center,align=center] at (0.4+5,0) {$1$};
\node [align=center,align=center] at (6.8,0) {$+\frac{1}{6}$};
\draw (1+7,0) circle [radius=0.2];
\draw (0.4+7,0) circle [radius=0.2];
\draw (1.17+7,0.1)--(1.4+7,0.15);
\draw (1.17+7,-0.1)--(1.4+7,-0.15);
\draw (0.6+7,0)--(0.8+7,0);
\draw (1.2+7,0)--(1.5+7,0);
\node [align=center,align=center] at (1+7,0) {$0$};
\node [align=center,align=center] at (0.4+7,0) {$1$};
\node [align=center,align=center] at (9,0) {$+\frac{1}{4}$};
\draw (1+9.4,0) circle [radius=0.2];
\draw (0.4+9.4,0) circle [radius=0.2];
\draw (1.17+9.4,0.1)--(1.4+9.4,0.15);
\draw (1.17+9.4,-0.1)--(1.4+9.4,-0.15);
\draw (0.6+9.4,0)--(0.8+9.4,0);
\draw (0.4+9.4,0.2)--(0.4+9.4,0.4);
\node [align=center,align=center] at (1+9.4,0) {$0$};
\node [align=center,align=center] at (0.4+9.4,0) {$0$};
\draw (0.24+9.4,0.1) .. controls (-0.1+9.4,0.2) and (-0.1+9.4,-0.2) ..  (0.24+9.4,-0.1);
\end{tikzpicture}
\\
&\qquad\quad
\begin{tikzpicture}
\node [align=center,align=center] at (0.1,0) {$+\frac{1}{4}$};
\draw (1,0) circle [radius=0.2];
\draw (0.5,0)--(0.8,0);
\draw (1.18,0.07)--(1.42,0.07);
\draw (1.18,-0.07)--(1.42,-0.07);
\draw (1.76,0.1)--(2,0.15);
\draw (1.76,-0.1)--(2,-0.15);
\draw (1.6,0) circle [radius=0.2];
\node [align=center,align=center] at (1,0) {$0$};
\node [align=center,align=center] at (1.6,0) {$0$};
\node [align=center,align=center] at (2.5,0) {$+\frac{1}{12}$};
\draw (1+3,0) circle [radius=0.2];
\draw (0.4+3,0) circle [radius=0.2];
\draw (1.17+3,0.1)--(1.4+3,0.15);
\draw (1.17+3,-0.1)--(1.4+3,-0.15);
\draw (0.6+3,0)--(0.8+3,0);
\draw (1.2+3,0)--(1.5+3,0);
\node [align=center,align=center] at (1+3,0) {$0$};
\node [align=center,align=center] at (0.4+3,0) {$0$};
\draw (0.24+3,0.1) .. controls (-0.1+3,0.2) and (-0.1+3,-0.2) ..  (0.24+3,-0.1);
\node [align=center,align=center] at (5,0) {$+\frac{1}{2}$};
\draw (1+5.8,0) circle [radius=0.2];
\draw (0.4+5.8,0) circle [radius=0.2];
\draw (-0.2+5.8,0) circle [radius=0.2];
\draw (1.17+5.8,0.1)--(1.4+5.8,0.15);
\draw (1.17+5.8,-0.1)--(1.4+5.8,-0.15);
\draw (0+5.8,0)--(0.2+5.8,0);
\draw (0.6+5.8,0)--(0.8+5.8,0);
\draw (0.4+5.8,0.2)--(0.4+5.8,0.4);
\node [align=center,align=center] at (1+5.8,0) {$0$};
\node [align=center,align=center] at (0.4+5.8,0) {$0$};
\node [align=center,align=center] at (-0.2+5.8,0) {$1$};
\node [align=center,align=center] at (7.7,0) {$+\frac{1}{4}$};
\draw (1+8.8,0) circle [radius=0.2];
\draw (0.4+8.8,0) circle [radius=0.2];
\draw (-0.2+8.8,0) circle [radius=0.2];
\draw (1.17+8.8,0.1)--(1.4+8.8,0.15);
\draw (1.17+8.8,-0.1)--(1.4+8.8,-0.15);
\draw (0.6+8.8,0)--(0.8+8.8,0);
\draw (-0.7+8.8,0)--(-0.4+8.8,0);
\draw (-0.02+8.8,0.07)--(0.22+8.8,0.07);
\draw (-0.02+8.8,-0.07)--(0.22+8.8,-0.07);
\node [align=center,align=center] at (1+8.8,0) {$0$};
\node [align=center,align=center] at (0.4+8.8,0) {$0$};
\node [align=center,align=center] at (-0.2+8.8,0) {$0$};
\end{tikzpicture}
\\
&\qquad\quad
\begin{tikzpicture}
\node [align=center,align=center] at (-1,0) {$+\frac{1}{4}$};
\draw (1,0) circle [radius=0.2];
\draw (0.4,0) circle [radius=0.2];
\draw (-0.2,0) circle [radius=0.2];
\draw (1.17,0.1)--(1.4,0.15);
\draw (1.17,-0.1)--(1.4,-0.15);
\draw (0,0)--(0.2,0);
\draw (0.6,0)--(0.8,0);
\draw (0.4,0.2)--(0.4,0.4);
\node [align=center,align=center] at (1,0) {$0$};
\node [align=center,align=center] at (0.4,0) {$0$};
\node [align=center,align=center] at (-0.2,0) {$0$};
\draw (-0.36,0.1) .. controls (-0.7,0.2) and (-0.7,-0.2) ..  (-0.36,-0.1);
\node [align=center,align=center] at (1.9,0) {$+\frac{1}{6}$};
\draw (1+1.8,0-0.1) circle [radius=0.2];
\draw (1.2+1.8,0-0.1)--(1.6+1.8,0-0.1);
\draw (0.5+1.8,0-0.1)--(0.8+1.8,0-0.1);
\draw (2+1.8,0-0.1)--(2.3+1.8,0-0.1);
\draw (1.8+1.8,0-0.1) circle [radius=0.2];
\draw (1.4+1.8,0.45-0.1) circle [radius=0.2];
\draw (1.14+1.8,0.14-0.1)--(1.26+1.8,0.31-0.1);
\draw (1.66+1.8,0.14-0.1)--(1.54+1.8,0.31-0.1);
\node [align=center,align=center] at (1+1.8,0-0.1) {$0$};
\node [align=center,align=center] at (1.8+1.8,0-0.1) {$0$};
\node [align=center,align=center] at (1.4+1.8,0.45-0.1) {$0$};
\draw (1.4+1.8,0.65-0.1)--(1.4+1.8,0.85-0.1);
\end{tikzpicture},
\end{split}&&
\end{flalign*}

\begin{flalign*}
\begin{tikzpicture}
\node [align=center,align=center] at (0.1+0.4,0) {$\widehat{\cF}_{2,0}=$};
\draw (1+0.3,0) circle [radius=0.2];
\node [align=center,align=center] at (1+0.3,0) {$2$};
\node [align=center,align=center] at (1.6+0.2,0) {$+\frac{1}{2}$};
\draw (1+1.4+0.2,0) circle [radius=0.2];
\draw (0.84+1.4+0.2,0.1) .. controls (0.5+1.4+0.2,0.2) and (0.5+1.4+0.2,-0.2) ..  (0.84+1.4+0.2,-0.1);
\node [align=center,align=center] at (1+1.4+0.2,0) {$1$};
\node [align=center,align=center] at (3+0.2,0) {$+\frac{1}{2}$};
\draw (1+2.6+0.2,0) circle [radius=0.2];
\draw (1.2+2.6+0.2,0)--(1.4+2.6+0.2,0);
\draw (1.6+2.6+0.2,0) circle [radius=0.2];
\node [align=center,align=center] at (1+2.6+0.2,0) {$1$};
\node [align=center,align=center] at (1.6+2.6+0.2,0) {$1$};
\node [align=center,align=center] at (5,0) {$+\frac{1}{8}$};
\draw (1+4.8,0) circle [radius=0.2];
\draw (0.84+4.8,0.1) .. controls (0.5+4.8,0.2) and (0.5+4.8,-0.2) ..  (0.84+4.8,-0.1);
\draw (1.16+4.8,0.1) .. controls (1.5+4.8,0.2) and (1.5+4.8,-0.2) ..  (1.16+4.8,-0.1);
\node [align=center,align=center] at (1+4.8,0) {$0$};
\node [align=center,align=center] at (6.6,0) {$+\frac{1}{2}$};
\draw (1+6.8,0) circle [radius=0.2];
\draw (0.4+6.8,0) circle [radius=0.2];
\draw (0.6+6.8,0)--(0.8+6.8,0);
\draw (1.16+6.8,0.1) .. controls (1.5+6.8,0.2) and (1.5+6.8,-0.2) ..  (1.16+6.8,-0.1);
\node [align=center,align=center] at (1+6.8,0) {$0$};
\node [align=center,align=center] at (0.4+6.8,0) {$1$};
\node [align=center,align=center] at (8.6,0) {$+\frac{1}{8}$};
\draw (1+9,0) circle [radius=0.2];
\draw (0.4+9,0) circle [radius=0.2];
\draw (0.6+9,0)--(0.8+9,0);
\draw (1.16+9,0.1) .. controls (1.5+9,0.2) and (1.5+9,-0.2) ..  (1.16+9,-0.1);
\draw (0.24+9,0.1) .. controls (-0.1+9,0.2) and (-0.1+9,-0.2) ..  (0.24+9,-0.1);
\node [align=center,align=center] at (1+9,0) {$0$};
\node [align=center,align=center] at (0.4+9,0) {$0$};
\node [align=center,align=center] at (10.6+0.2,0) {$+\frac{1}{12}$};
\draw (1+10.2+0.2,0) circle [radius=0.2];
\draw (1.2+10.2+0.2,0)--(1.4+10.2+0.2,0);
\draw (1.16+10.2+0.2,0.1)--(1.44+10.2+0.2,0.1);
\draw (1.16+10.2+0.2,-0.1)--(1.44+10.2+0.2,-0.1);
\draw (1.6+10.2+0.2,0) circle [radius=0.2];
\node [align=center,align=center] at (1+10.2+0.2,0) {$0$};
\node [align=center,align=center] at (1.6+10.2+0.2,0) {$0$};
\end{tikzpicture},&&
\end{flalign*}

\begin{flalign*}
\begin{split}
&\begin{tikzpicture}
\node [align=center,align=center] at (0.2,0) {$\widehat{\cF}_{2,1}=$};
\draw (1,0) circle [radius=0.2];
\draw (1.2,0)--(1.5,0);
\node [align=center,align=center] at (1,0) {$2$};
\node [align=center,align=center] at (2,0) {$+\frac{1}{2}$};
\draw (1+1.8,0) circle [radius=0.2];
\draw (1.2+1.8,0)--(1.5+1.8,0);
\draw (0.84+1.8,0.1) .. controls (0.5+1.8,0.2) and (0.5+1.8,-0.2) ..  (0.84+1.8,-0.1);
\node [align=center,align=center] at (1+1.8,0) {$1$};
\node [align=center,align=center] at (3.7,0) {$+$};
\draw (1+3.2,0) circle [radius=0.2];
\draw (1.2+3.2,0)--(1.4+3.2,0);
\draw (1.8+3.2,0)--(2.1+3.2,0);
\draw (1.6+3.2,0) circle [radius=0.2];
\node [align=center,align=center] at (1+3.2,0) {$1$};
\node [align=center,align=center] at (1.6+3.2,0) {$1$};
\node [align=center,align=center] at (5.8,0) {$+\frac{1}{8}$};
\draw (1+5.6,0) circle [radius=0.2];
\draw (1+5.6,0.2)--(1+5.6,0.4);
\draw (0.84+5.6,0.1) .. controls (0.5+5.6,0.2) and (0.5+5.6,-0.2) ..  (0.84+5.6,-0.1);
\draw (1.16+5.6,0.1) .. controls (1.5+5.6,0.2) and (1.5+5.6,-0.2) ..  (1.16+5.6,-0.1);
\node [align=center,align=center] at (1+5.6,0) {$0$};
\node [align=center,align=center] at (7.4,0) {$+\frac{1}{2}$};
\draw (1+7.9,0) circle [radius=0.2];
\draw (0.4+7.9,0) circle [radius=0.2];
\draw (0.6+7.9,0)--(0.8+7.9,0);
\draw (-0.1+7.9,0)--(0.2+7.9,0);
\draw (1.16+7.9,0.1) .. controls (1.5+7.9,0.2) and (1.5+7.9,-0.2) ..  (1.16+7.9,-0.1);
\node [align=center,align=center] at (1+7.9,0) {$0$};
\node [align=center,align=center] at (0.4+7.9,0) {$1$};
\node [align=center,align=center] at (9.8,0) {$+\frac{1}{2}$};
\draw (1+9.4,0) circle [radius=0.2];
\draw (1.18+9.4,0.07)--(1.42+9.4,0.07);
\draw (1.18+9.4,-0.07)--(1.42+9.4,-0.07);
\draw (1.8+9.4,-0)--(2.1+9.4,0);
\draw (1.6+9.4,0) circle [radius=0.2];
\node [align=center,align=center] at (1+9.4,0) {$1$};
\node [align=center,align=center] at (1.6+9.4,0) {$0$};
\end{tikzpicture}
\\
&\qquad\quad
\begin{tikzpicture}
\node [align=center,align=center] at (-0.2-0.6,0) {$+\frac{1}{2}$};
\draw (1-0.6,0) circle [radius=0.2];
\draw (0.4-0.6,0) circle [radius=0.2];
\draw (0.6-0.6,0)--(0.8-0.6,0);
\draw (1-0.6,0.2)--(1-0.6,0.4);
\draw (1.16-0.6,0.1) .. controls (1.5-0.6,0.2) and (1.5-0.6,-0.2) ..  (1.16-0.6,-0.1);
\node [align=center,align=center] at (1-0.6,0) {$0$};
\node [align=center,align=center] at (0.4-0.6,0) {$1$};
\node [align=center,align=center] at (1.8-0.6,0) {$+\frac{1}{2}$};
\draw (1+2,0) circle [radius=0.2];
\draw (0.4+2,0) circle [radius=0.2];
\draw (-0.2+2,0) circle [radius=0.2];
\draw (0+2,0)--(0.2+2,0);
\draw (0.6+2,0)--(0.8+2,0);
\draw (0.4+2,0.2)--(0.4+2,0.4);
\node [align=center,align=center] at (1+2,0) {$1$};
\node [align=center,align=center] at (0.4+2,0) {$0$};
\node [align=center,align=center] at (-0.2+2,0) {$1$};
\node [align=center,align=center] at (3.6,0) {$+\frac{1}{4}$};
\draw (1+4,0) circle [radius=0.2];
\draw (0.4+4,0) circle [radius=0.2];
\draw (0.6+4,0)--(0.8+4,0);
\draw (1+4,0.2)--(1+4,0.4);
\draw (1.16+4,0.1) .. controls (1.5+4,0.2) and (1.5+4,-0.2) ..  (1.16+4,-0.1);
\draw (0.24+4,0.1) .. controls (-0.1+4,0.2) and (-0.1+4,-0.2) ..  (0.24+4,-0.1);
\node [align=center,align=center] at (1+4,0) {$0$};
\node [align=center,align=center] at (0.4+4,0) {$0$};
\node [align=center,align=center] at (5.8,0) {$+\frac{1}{4}$};
\draw (1+6.3,0) circle [radius=0.2];
\draw (0.4+6.3,0) circle [radius=0.2];
\draw (0.58+6.3,0.07)--(0.82+6.3,0.07);
\draw (0.58+6.3,-0.07)--(0.82+6.3,-0.07);
\draw (-0.1+6.3,0)--(0.2+6.3,0);
\draw (1.16+6.3,0.1) .. controls (1.5+6.3,0.2) and (1.5+6.3,-0.2) ..  (1.16+6.3,-0.1);
\node [align=center,align=center] at (1+6.3,0) {$0$};
\node [align=center,align=center] at (0.4+6.3,0) {$0$};
\end{tikzpicture}
\\
&\qquad\quad
\begin{tikzpicture}
\node [align=center,align=center] at (0.4,0) {$+\frac{1}{6}$};
\draw (1,0) circle [radius=0.2];
\draw (1.2,0)--(1.4,0);
\draw (1.8,0)--(2.1,0);
\draw (1.16,0.1)--(1.44,0.1);
\draw (1.16,-0.1)--(1.44,-0.1);
\draw (1.6,0) circle [radius=0.2];
\node [align=center,align=center] at (1,0) {$0$};
\node [align=center,align=center] at (1.6,0) {$0$};
\node [align=center,align=center] at (2.5,0) {$+\frac{1}{2}$};
\draw (1+3.5,0) circle [radius=0.2];
\draw (0.4+3.5,0) circle [radius=0.2];
\draw (-0.2+3.5,0) circle [radius=0.2];
\draw (0+3.5,0)--(0.2+3.5,0);
\draw (0.6+3.5,0)--(0.8+3.5,0);
\draw (0.4+3.5,0.2)--(0.4+3.5,0.4);
\node [align=center,align=center] at (1+3.5,0) {$1$};
\node [align=center,align=center] at (0.4+3.5,0) {$0$};
\node [align=center,align=center] at (-0.2+3.5,0) {$0$};
\draw (-0.36+3.5,0.1) .. controls (-0.7+3.5,0.2) and (-0.7+3.5,-0.2) ..  (-0.36+3.5,-0.1);
\node [align=center,align=center] at (5.1,0) {$+\frac{1}{2}$};
\draw (1+6.2,0) circle [radius=0.2];
\draw (0.4+6.2,0) circle [radius=0.2];
\draw (-0.2+6.2,0) circle [radius=0.2];
\draw (0.6+6.2,0)--(0.8+6.2,0);
\draw (-0.7+6.2,0)--(-0.4+6.2,0);
\draw (-0.02+6.2,0.07)--(0.22+6.2,0.07);
\draw (-0.02+6.2,-0.07)--(0.22+6.2,-0.07);
\node [align=center,align=center] at (1+6.2,0) {$1$};
\node [align=center,align=center] at (0.4+6.2,0) {$0$};
\node [align=center,align=center] at (-0.2+6.2,0) {$0$};
\node [align=center,align=center] at (7.8,0) {$+\frac{1}{4}$};
\draw (1+8.9,0) circle [radius=0.2];
\draw (0.4+8.9,0) circle [radius=0.2];
\draw (-0.2+8.9,0) circle [radius=0.2];
\draw (0.6+8.9,0)--(0.8+8.9,0);
\draw (-0.7+8.9,0)--(-0.4+8.9,0);
\draw (-0.02+8.9,0.07)--(0.22+8.9,0.07);
\draw (-0.02+8.9,-0.07)--(0.22+8.9,-0.07);
\node [align=center,align=center] at (1+8.9,0) {$0$};
\node [align=center,align=center] at (0.4+8.9,0) {$0$};
\node [align=center,align=center] at (-0.2+8.9,0) {$0$};
\draw (1.16+8.9,0.1) .. controls (1.5+8.9,0.2) and (1.5+8.9,-0.2) ..  (1.16+8.9,-0.1);
\end{tikzpicture}
\\
&\qquad\quad
\begin{tikzpicture}
\node [align=center,align=center] at (-0.4,0) {$+\frac{1}{8}$};
\draw (1,0) circle [radius=0.2];
\draw (0.4,0) circle [radius=0.2];
\draw (1.6,0) circle [radius=0.2];
\draw (0.6,0)--(0.8,0);
\draw (1.2,0)--(1.4,0);
\draw (1,0.2)--(1,0.4);
\draw (1.76,0.1) .. controls (2.1,0.2) and (2.1,-0.2) ..  (1.76,-0.1);
\draw (0.24,0.1) .. controls (-0.1,0.2) and (-0.1,-0.2) ..  (0.24,-0.1);
\node [align=center,align=center] at (1,0) {$0$};
\node [align=center,align=center] at (0.4,0) {$0$};
\node [align=center,align=center] at (1.6,0) {$0$};
\node [align=center,align=center] at (2.4,0) {$+\frac{1}{4}$};
\draw (1+2,0-0.1) circle [radius=0.2];
\draw (1.18+2,0.07-0.1)--(1.62+2,0.07-0.1);
\draw (1.18+2,-0.07-0.1)--(1.62+2,-0.07-0.1);
\draw (1.8+2,0-0.1) circle [radius=0.2];
\draw (1.4+2,0.45-0.1) circle [radius=0.2];
\draw (1.14+2,0.14-0.1)--(1.26+2,0.31-0.1);
\draw (1.66+2,0.14-0.1)--(1.54+2,0.31-0.1);
\node [align=center,align=center] at (1+2,0-0.1) {$0$};
\node [align=center,align=center] at (1.8+2,0-0.1) {$0$};
\node [align=center,align=center] at (1.4+2,0.45-0.1) {$0$};
\draw (1.4+2,0.65-0.1)--(1.4+2,0.85-0.1);
\end{tikzpicture}.
\end{split}&&
\end{flalign*}
More examples will be given in Appendix \ref{app1}.
\end{Example}

\subsection{A recursion relation for the abstract free energy}

\label{sec:Recursion-1}

\begin{Example}
Using the expressions in the above example, one can directly check the following identities:
\ben
&&K \widehat{\cF}_{0,4}=\frac{9}{2}(\widehat{\cF}_{0,3})^2,\\
&&K \widehat{\cF}_{0,5}=12\widehat{\cF}_{0,3}\widehat{\cF}_{0,4},\\
&&K \widehat{\cF}_{1,1}=3\widehat{\cF}_{0,3},\\
&&K \widehat{\cF}_{1,2}=6\widehat{\cF}_{0,4}+3\widehat{\cF}_{1,1}\widehat{\cF}_{0,3},\\
&&K \widehat{\cF}_{1,3}=10\widehat{\cF}_{0,5}+4\widehat{\cF}_{1,1}\widehat{\cF}_{0,4}+6\widehat{\cF}_{1,2}\widehat{\cF}_{0,3},\\
&&K \widehat{\cF}_2=\widehat{\cF}_{1,2}+\frac{1}{2}(\widehat{\cF}_{1,1})^2,\\
&&K \widehat{\cF}_{2,1}=3\widehat{\cF}_{1,3}+2\widehat{\cF}_{1,1}\widehat{\cF}_{1,2}.
\een
\end{Example}

In general, we have the following relation.
\begin{Theorem}\label{thm1}
For $2g-2+n>0$, we have
\be\label{eq-thm1}
K\widehat{\cF}_{g,n}=\binom{n+2}{2}\widehat{\cF}_{g-1,n+2}+\frac{1}{2}\sum_{\substack{g_1+g_2=g,n_1+n_2=n+2,\\n_1\geq 1,n_2\geq 1}}(n_1\widehat{\cF}_{g_1,n_1})(n_2\widehat{\cF}_{g_2,n_2}),
\ee
where the sum on the right hand side is taken over all stable cases.
\end{Theorem}

\begin{proof}
First note that cutting off an internal edge of a stable graph in $\cG_{g,n}^c$ gives us two types of new stable graphs, one type is still connected, and must be of genus $g-1$ with $n+2$ external edges; the other type of graphs have two connected components, and have total genus $g$ with $n+2$ external edges. Thus we write
\be\label{eq0}
K\widehat{\cF}_{g,n}=I_{g,n}+J_{g,n},
\ee
where $I_{g,n}$ and $J_{g,n}$ are sums over connected graphs and disconnected graphs respectively.

We already know that graphs appearing in $I_{g,n}$ all belong to $\cG_{g-1,n+2}^c$;
conversely, every graph in $\cG_{g-1,n+2}^c$ may be transformed to
a graph appearing in $\widehat{\cF}_{g,n}$ once we glue two external edges together.
This tells us that the graphs appearing in $I_{g,n}$ are the exactly the same graphs
in $\widehat{\cF}_{g-1,n+2}$.
Now we check the coefficients of every graph in the two sums coincide.
For a stable graph $\Gamma\in\cG_{g-1,n+2}^c$,
there are $\binom{n+2}{2}$ ways to choose two external edges to be glued together.
Denote by $\widetilde{\Gamma}$  one of the equivalence classes of graphs
obtained from $\Gamma$ by gluing two external edges.
Then we need to show that
\be
K'\biggl(\sum_{\widetilde{\Gamma}\in \cG_\Gamma} \frac{1}{|\Aut(\widetilde{\Gamma})|}
\cdot \widetilde{\Gamma} \biggr)
= \binom{n+2}{2} \cdot \frac{1}{|\Aut(\Gamma)|}\Gamma,
\ee
where $\cG_\Gamma$ is the set of equivalence classes of stable graphs $\widetilde{\Gamma}$ such that
one can obtain $\Gamma$ by cutting an edge of $\widetilde{\Gamma}$,
and $K'(\widetilde{\Gamma}) = N(\widetilde{\Gamma})\cdot \Gamma$,
where $N(\widetilde{\Gamma})$ is the number of ways to cut an edge of $\widetilde{\Gamma}$
to get $\Gamma$.
It suffices to show that
\be
\sum_{\widetilde{\Gamma}\in \cG_\Gamma} N(\widetilde{\Gamma})\cdot
\frac{|\Aut(\Gamma)|}{|\Aut(\widetilde{\Gamma})|}
= \binom{n+2}{2}.
\ee
Consider the set of the choices of picking a pair of external edges of $\Gamma$.
This set has $\binom{n+2}{2}$ elements.
We partition this set into disjoint union of $H_{\widetilde{\Gamma}}$
consisting of those choices for which one can get $\widetilde{\Gamma}$
by gluing the pair of external edges,
where $\widetilde{\Gamma}$ runs over the set $\cG_\Gamma$ of all possibilities.
Therefore,
\be
\sum_{\widetilde{\Gamma}\in \cG_\Gamma} |H_{\widetilde{\Gamma}}| = \binom{n+2}{2}.
\ee
Now we show that
\be
|H_{\widetilde{\Gamma}}| = N(\widetilde{\Gamma})\cdot
\frac{|\Aut(\Gamma)|}{|\Aut(\widetilde{\Gamma})|},
\ee
or equivalently,
\be
\frac{|H_{\widetilde{\Gamma}}|}{|\Aut(\Gamma)|} =
\frac{N(\widetilde{\Gamma})}{|\Aut(\widetilde{\Gamma})|}.
\ee

Let $\Gamma''$ be the graph  obtained from $\widetilde{\Gamma}$ by
changing the internal edge that we cut to get $\Gamma$ into a dotted edge,
and let $\Gamma'$ be the graph with two dotted external edges
obtained from $\Gamma''$ by cutting
the dotted edge.
For examples, see Figure \ref{eg-fig1} and Figure \ref{eg-fig2}.
\begin{figure}[H]
\begin{tikzpicture}
\draw (0,0) circle [radius=0.2];
\draw (-0.6,0) circle [radius=0.2];
\draw (0.6,0) circle [radius=0.2];
\draw (0,0.6) circle [radius=0.2];
\draw (0,-0.6) circle [radius=0.2];
\node [align=center,align=center] at (0.6,0.35) {$v$};
\draw (0.2,0)--(0.4,0);
\draw (-0.2,0)--(-0.4,0);
\draw (0,0.2)--(0,0.4);
\draw (0,-0.2)--(0,-0.4);
\node [align=center,align=center] at (0,0) {$0$};
\node [align=center,align=center] at (0.6,0) {$1$};
\node [align=center,align=center] at (-0.6,0) {$1$};
\node [align=center,align=center] at (0,0.6) {$1$};
\node [align=center,align=center] at (0,-0.6) {$1$};
\node [align=center,align=center] at (0,-1.4) {$\Gamma$};
\draw (0.8,0)--(1.1,0);
\draw (0.76,0.1)--(1.1,0.2);
\draw (0.76,-0.1)--(1.1,-0.2);
\draw (-0.8,0)--(-1.1,0);
\draw (-0.76,0.1)--(-1.1,0.2);
\draw (-0.76,-0.1)--(-1.1,-0.2);
\draw (0,0.8)--(0,1.1);
\draw (0.1,0.76)--(0.2,1.1);
\draw (-0.1,0.76)--(-0.2,1.1);
\draw (0.2,-0.6)--(0.5,-0.6);
\draw (-0.1,-0.76) .. controls (-0.2,-1.1) and (0.2,-1.1) ..  (0.1,-0.76);
\end{tikzpicture}
\quad
\begin{tikzpicture}
\draw (0,0) circle [radius=0.2];
\draw (-0.6,0) circle [radius=0.2];
\draw (0.6,0) circle [radius=0.2];
\draw (0,0.6) circle [radius=0.2];
\draw (0,-0.6) circle [radius=0.2];
\node [align=center,align=center] at (0.6,0.35) {$v$};
\draw (0.2,0)--(0.4,0);
\draw (-0.2,0)--(-0.4,0);
\draw (0,0.2)--(0,0.4);
\draw (0,-0.2)--(0,-0.4);
\node [align=center,align=center] at (0,0) {$0$};
\node [align=center,align=center] at (0.6,0) {$1$};
\node [align=center,align=center] at (-0.6,0) {$1$};
\node [align=center,align=center] at (0,0.6) {$1$};
\node [align=center,align=center] at (0,-0.6) {$1$};
\node [align=center,align=center] at (0,-1.4) {$\Gamma'$};
\draw [dotted,ultra thick](0.8,0)--(1.1,0);
\draw (0.76,0.1)--(1.1,0.2);
\draw [dotted,ultra thick](0.76,-0.1)--(1.1,-0.2);
\draw (-0.8,0)--(-1.1,0);
\draw (-0.76,0.1)--(-1.1,0.2);
\draw (-0.76,-0.1)--(-1.1,-0.2);
\draw (0,0.8)--(0,1.1);
\draw (0.1,0.76)--(0.2,1.1);
\draw (-0.1,0.76)--(-0.2,1.1);
\draw (0.2,-0.6)--(0.5,-0.6);
\draw (-0.1,-0.76) .. controls (-0.2,-1.1) and (0.2,-1.1) ..  (0.1,-0.76);
\end{tikzpicture}
\quad
\begin{tikzpicture}
\draw (0,0) circle [radius=0.2];
\draw (-0.6,0) circle [radius=0.2];
\draw (0.6,0) circle [radius=0.2];
\draw (0,0.6) circle [radius=0.2];
\draw (0,-0.6) circle [radius=0.2];
\node [align=center,align=center] at (0.8,0.35) {$v$};
\draw (0.2,0)--(0.4,0);
\draw (-0.2,0)--(-0.4,0);
\draw (0,0.2)--(0,0.4);
\draw (0,-0.2)--(0,-0.4);
\node [align=center,align=center] at (0,0) {$0$};
\node [align=center,align=center] at (0.6,0) {$1$};
\node [align=center,align=center] at (-0.6,0) {$1$};
\node [align=center,align=center] at (0,0.6) {$1$};
\node [align=center,align=center] at (0,-0.6) {$1$};
\node [align=center,align=center] at (0,-1.4) {$\Gamma''$};
\draw (0.6,0.2)--(0.6,0.5);
\draw (-0.8,0)--(-1.1,0);
\draw (-0.76,0.1)--(-1.1,0.2);
\draw (-0.76,-0.1)--(-1.1,-0.2);
\draw (0,0.8)--(0,1.1);
\draw (0.1,0.76)--(0.2,1.1);
\draw (-0.1,0.76)--(-0.2,1.1);
\draw (0.2,-0.6)--(0.5,-0.6);
\draw [dotted,ultra thick](0.76,0.1) .. controls (1.1,0.2) and (1.1,-0.2) ..  (0.76,-0.1);
\draw (-0.1,-0.76) .. controls (-0.2,-1.1) and (0.2,-1.1) ..  (0.1,-0.76);
\end{tikzpicture}
\quad
\begin{tikzpicture}
\draw (0,0) circle [radius=0.2];
\draw (-0.6,0) circle [radius=0.2];
\draw (0.6,0) circle [radius=0.2];
\draw (0,0.6) circle [radius=0.2];
\draw (0,-0.6) circle [radius=0.2];
\node [align=center,align=center] at (0.8,0.35) {$v$};
\draw (0.2,0)--(0.4,0);
\draw (-0.2,0)--(-0.4,0);
\draw (0,0.2)--(0,0.4);
\draw (0,-0.2)--(0,-0.4);
\node [align=center,align=center] at (0,0) {$0$};
\node [align=center,align=center] at (0.6,0) {$1$};
\node [align=center,align=center] at (-0.6,0) {$1$};
\node [align=center,align=center] at (0,0.6) {$1$};
\node [align=center,align=center] at (0,-0.6) {$1$};
\node [align=center,align=center] at (0,-1.4) {$\widetilde\Gamma$};
\draw (0.6,0.2)--(0.6,0.5);
\draw (-0.8,0)--(-1.1,0);
\draw (-0.76,0.1)--(-1.1,0.2);
\draw (-0.76,-0.1)--(-1.1,-0.2);
\draw (0,0.8)--(0,1.1);
\draw (0.1,0.76)--(0.2,1.1);
\draw (-0.1,0.76)--(-0.2,1.1);
\draw (0.2,-0.6)--(0.5,-0.6);
\draw (0.76,0.1) .. controls (1.1,0.2) and (1.1,-0.2) ..  (0.76,-0.1);
\draw (-0.1,-0.76) .. controls (-0.2,-1.1) and (0.2,-1.1) ..  (0.1,-0.76);
\end{tikzpicture}
\caption{An example where the dotted edge is a loop.}
\label{eg-fig1}
\end{figure}
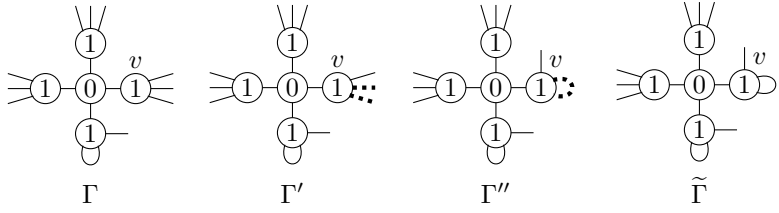

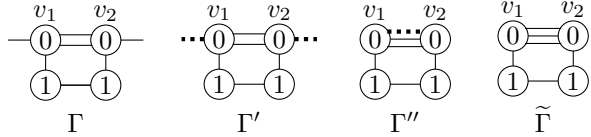
\begin{figure}[H]
\begin{tikzpicture}
\draw (0,0) circle [radius=0.2];
\draw (0.8,0) circle [radius=0.2];
\draw (0,0.6) circle [radius=0.2];
\draw (0.8,0.6) circle [radius=0.2];
\node [align=center,align=center] at (0,0.95) {$v_1$};
\node [align=center,align=center] at (0.8,0.95) {$v_2$};
\draw (0.2,0)--(0.6,0);
\draw (0,0.2)--(0,0.4);
\draw (0.8,0.2)--(0.8,0.4);
\draw (0.18,0.67)--(0.62,0.67);
\draw (0.18,0.53)--(0.62,0.53);
\draw (-0.2,0.6)--(-0.5,0.6);
\draw (1,0.6)--(1.3,0.6);
\draw (0.2,0)--(0.6,0);
\node [align=center,align=center] at (0,0) {$1$};
\node [align=center,align=center] at (0.8,0) {$1$};
\node [align=center,align=center] at (0,0.6) {$0$};
\node [align=center,align=center] at (0.8,0.6) {$0$};
\node [align=center,align=center] at (0.4,-0.5) {$\Gamma$};
\end{tikzpicture}
\quad
\begin{tikzpicture}
\draw (0,0) circle [radius=0.2];
\draw (0.8,0) circle [radius=0.2];
\draw (0,0.6) circle [radius=0.2];
\draw (0.8,0.6) circle [radius=0.2];
\node [align=center,align=center] at (0,0.95) {$v_1$};
\node [align=center,align=center] at (0.8,0.95) {$v_2$};
\draw (0.2,0)--(0.6,0);
\draw (0,0.2)--(0,0.4);
\draw (0.8,0.2)--(0.8,0.4);
\draw (0.18,0.67)--(0.62,0.67);
\draw (0.18,0.53)--(0.62,0.53);
\draw [dotted, ultra thick](-0.2,0.6)--(-0.5,0.6);
\draw [dotted, ultra thick](1,0.6)--(1.3,0.6);
\node [align=center,align=center] at (0,0) {$1$};
\node [align=center,align=center] at (0.8,0) {$1$};
\node [align=center,align=center] at (0,0.6) {$0$};
\node [align=center,align=center] at (0.8,0.6) {$0$};
\node [align=center,align=center] at (0.4,-0.5) {$\Gamma'$};
\end{tikzpicture}
\quad
\begin{tikzpicture}
\draw (0,0) circle [radius=0.2];
\draw (0.8,0) circle [radius=0.2];
\draw (0,0.6) circle [radius=0.2];
\draw (0.8,0.6) circle [radius=0.2];
\node [align=center,align=center] at (0,0.95) {$v_1$};
\node [align=center,align=center] at (0.8,0.95) {$v_2$};
\draw (0.2,0)--(0.6,0);
\draw (0,0.2)--(0,0.4);
\draw (0.8,0.2)--(0.8,0.4);
\draw (0.2,0.6)--(0.6,0.6);
\draw [dotted, ultra thick](0.16,0.7)--(0.64,0.7);
\draw (0.16,0.5)--(0.64,0.5);
\node [align=center,align=center] at (0,0) {$1$};
\node [align=center,align=center] at (0.8,0) {$1$};
\node [align=center,align=center] at (0,0.6) {$0$};
\node [align=center,align=center] at (0.8,0.6) {$0$};
\node [align=center,align=center] at (0.4,-0.5) {$\Gamma''$};
\end{tikzpicture}
\quad
\begin{tikzpicture}
\draw (0,0) circle [radius=0.2];
\draw (0.8,0) circle [radius=0.2];
\draw (0,0.6) circle [radius=0.2];
\draw (0.8,0.6) circle [radius=0.2];
\node [align=center,align=center] at (0,0.95) {$v_1$};
\node [align=center,align=center] at (0.8,0.95) {$v_2$};
\draw (0.2,0)--(0.6,0);
\draw (0,0.2)--(0,0.4);
\draw (0.8,0.2)--(0.8,0.4);
\draw (0.2,0.6)--(0.6,0.6);
\draw (0.16,0.7)--(0.64,0.7);
\draw (0.16,0.5)--(0.64,0.5);
\node [align=center,align=center] at (0,0) {$1$};
\node [align=center,align=center] at (0.8,0) {$1$};
\node [align=center,align=center] at (0,0.6) {$0$};
\node [align=center,align=center] at (0.8,0.6) {$0$};
\node [align=center,align=center] at (0.4,-0.5) {$\widetilde\Gamma$};
\end{tikzpicture}
\caption{An example where the dotted edge is not a loop.}
\label{eg-fig2}
\end{figure}

When we talk about automorphisms of $\Gamma'$ and $\Gamma''$,
the dotted edges can not be mapped to solid edges. Now we claim:
\be\label{eq-claim}
\frac{|H_{\widetilde{\Gamma}}|}{|\Aut(\Gamma)|} =
\frac{1}{|\Aut(\Gamma')|}=
\frac{1}{|\Aut(\Gamma'')|}=
\frac{N(\widetilde{\Gamma})}{|\Aut(\widetilde{\Gamma})|}.
\ee

The second equality is trivial. We only need to consider the first and third equalities.
Recall that $|H_{\widetilde{\Gamma}}|$ is the number of ways to
choose two external edges of $\Gamma$
to be glued together to get $\widetilde\Gamma$,
thus $|H_{\widetilde{\Gamma}}|$ equals to the number of all possible ways
to change a pair of external edges of $\Gamma$ to dotted edges to obtain $\Gamma'$.
Similarly, $N(\widetilde{\Gamma})$ equals to the number of
all possible ways to change an internal edge of $\widetilde\Gamma$
to a dotted edge to obtain $\Gamma''$.

In the case where the dotted edge is a loop in $\Gamma''$,
denote by $v$ the vertex where it is incident to.
Then let $\Aut_v(\Gamma)$ be the group of automorphisms of $\Gamma$ that fix the vertex $v$
(e.g., for Figure \ref{eg-fig1}, $\Aut_v(\Gamma)\cong(S_3)^3\times(S_2)^2$).
Let $O_v$ be the orbit of $v$ in $\Gamma$ under $\Aut(\Gamma)$, then
\be\label{orbit-v}
\frac{|\Aut(\Gamma)|}{|\Aut_v(\Gamma)|}=|O_v|.
\ee
It is clear that any automorphism of $\Gamma'$ fixes the vertex $v$ since there is
no other vertex with dotted half edges,
and so $\Aut(\Gamma')$ is a subgroup of $\Aut_v(\Gamma)$.
To compute the index of this subgroup,
it suffices to consider the actions of these automorphisms groups
on the external half edges incident at $v$.
Assume there are $l$ external edges attaching to $v$ in $\Gamma'$ including two dotted edges,
then we have
\be\label{choose-edge}
\frac{|\Aut_v(\Gamma)|}{|\Aut(\Gamma')|}=\frac{l!}{(l-2)!}\cdot\frac{1}{2},
\ee
here the factor $\frac{1}{2}$ comes from the symmetry between two dotted edges.
Now \eqref{orbit-v} and \eqref{choose-edge} give us
\be
\frac{|\Aut(\Gamma)|}{|\Aut(\Gamma')|}=|O_v|\cdot\binom{l}{2}.
\ee
Noticing that $|O_v|$ equals to the number of ways to choose the vertex $v$,
and $\binom{l}{2}$ is the number of ways to choose two external edges on $v$,
thus
\ben
|O_v|\cdot\binom{l}{2}=|H_{\widetilde{\Gamma}}|
\een
and this proves the first equality in \eqref{eq-claim}.

Using the same argument, we now consider $\Aut_v(\widetilde{\Gamma})$, i.e.,
the group of automorphisms of $\widetilde\Gamma$ that fix $v$
(e.g., for Figure \ref{eg-fig1} $\Aut_v(\widetilde{\Gamma})\cong(S_3)^2\times(S_2)^3$).
Let $\widetilde O_v$ be the orbit of $v$ in $\widetilde\Gamma$ under $\Aut(\widetilde\Gamma)$,
then
\be\label{orbit-v2}
\frac{|\Aut(\widetilde\Gamma)|}{|\Aut_v(\widetilde\Gamma)|}=|\widetilde O_v|.
\ee
Assume there are $k$ loops attaching to $v$ in $\Gamma''$ including a dotted one, then
\be\label{choose-edge2}
\frac{|\Aut_v(\widetilde\Gamma)|}{|\Aut(\Gamma'')|}=\frac{k!\cdot 2^{k}}{(k-1)!\cdot 2^{k-1}}\cdot\frac{1}{2}=k.
\ee
Thus \eqref{orbit-v2} and \eqref{choose-edge2} give us
\ben
\frac{|\Aut(\widetilde\Gamma)|}{|\Aut(\Gamma'')|}=k\cdot |\widetilde O_v|,
\een
and the right hand side equals to $N(\widetilde{\Gamma})$.
Therefore the third equality in \eqref{eq-claim} also holds.

The case where the dotted edge in $\Gamma''$ is not a loop can be proved similarly.
Let $v=(v_1,v_2)$ be the ordered pair of end points of the dotted edge.
As above,
we consider $\Aut_v(\Gamma)$
(here the order of $(v_1,v_2)$ should be preserved,
e.g., for Figure \ref{eg-fig2}, $\Aut_v(\Gamma)\cong S_2)$,
and  the orbit $O_v$ of the ordered pair $v$ in $V(\Gamma)\times V(\Gamma)$ under $\Aut(\Gamma)$.
We still have the equality:
\be\label{orbit-v3}
\frac{|\Aut(\Gamma)|}{|\Aut_v(\Gamma)|}=|O_v|.
\ee
Assume there are $l_i$ external edges attached to $v_i$ in $\Gamma'$, then
\be\label{choose-edge3}
\frac{|\Aut_v(\Gamma)|}{|\Aut(\Gamma')|}=\frac{l_1!}{(l_1-1)!}\cdot\frac{l_2!}{(l_2-1)!}\cdot\delta=l_1l_2\cdot\delta,
\ee
here $\delta=\frac{1}{2}$ if $(v_2,v_1)\in O_v$, and $\delta=1$ if $(v_2,v_1)\not\in O_v$.
The equations \eqref{orbit-v3} and \eqref{choose-edge3} give us
\be
\frac{|\Aut(\Gamma)|}{|\Aut(\Gamma')|}=l_1l_2|O_v|\cdot\delta,
\ee
which equals $|H_{\widetilde{\Gamma}}|$.

Similarly define $\Aut_v(\widetilde{\Gamma})$ to be the group of automorphisms
of $\widetilde\Gamma$ that fix $v$
(e.g., for Figure \ref{eg-fig2}, $\Aut(\widetilde{\Gamma})\cong S_3$).
Let $\widetilde O_v$ be the orbit of $v$ in $V(\widetilde\Gamma)\times V(\widetilde\Gamma)$ under $\Aut(\widetilde\Gamma)$,
then
\be\label{orbit-v4}
\frac{|\Aut(\widetilde\Gamma)|}{|\Aut_v(\widetilde\Gamma)|}=|\widetilde O_v|.
\ee
Assume there are $k$ internal edges connecting $v_1$ and $v_2$ in $\Gamma''$ including a dotted one, then
\be\label{choose-edge4}
\frac{|\Aut_v(\widetilde\Gamma)|}{|\Aut(\Gamma'')|}=\frac{k!}{(k-1)!}\cdot\delta=k\cdot\delta.
\ee
Then the equations \eqref{orbit-v4} and \eqref{choose-edge4} imply
\ben
\frac{|\Aut(\widetilde\Gamma)|}{|\Aut(\Gamma'')|}=k |\widetilde O_v|\cdot\delta,
\een
which equals to $N(\widetilde{\Gamma})$.
Therefore the first and third equalities in \eqref{eq-claim}
also hold in this case.

In summary, we've proved the identity
\be\label{eqn-I}
I_{g,n}=\binom{n+2}{2}\widehat{\cF}_{g-1,n+2}.
\ee

The relation
\be\label{eqn-J}
J_{g,n}=\frac{1}{2}\sum_{\substack{g_1+g_2=g,n_1+n_2=n+2,\\n_1\geq 1,n_2\geq 1}}(n_1\widehat{\cF}_{g_1,n_1})(n_2\widehat{\cF}_{g_2,n_2})
\ee
can be established in a similar fashion.
For two graphs $\Gamma_1\in\cG_{g_1,n_1}^c$, $\Gamma_2\in\cG_{g_2,n_2}^c$,
denote by $\widetilde{\Gamma}$  one of the equivalence classes of graphs
obtained from $\Gamma_1$ and $\Gamma_2$ by gluing two external edges, one from each graph.
Then we need to show that
\be
K''\biggl(\sum_{\widetilde{\Gamma}\in \cG''_\Gamma} \frac{1}{|\Aut(\widetilde{\Gamma})|}
\cdot \widetilde{\Gamma} \biggr)
= \frac{n_1}{|\Aut(\Gamma_1)|}\Gamma_1 \cdot \frac{n_2}{|\Aut(\Gamma_2)|}\Gamma_2,
\ee
where $\cG_\Gamma''$ is the set of equivalence classes of stable graphs $\widetilde{\Gamma}$ such that
one can obtain $\Gamma_1$ and $\Gamma_2$ by cutting an edge of $\widetilde{\Gamma}$,
and
$$K''(\widetilde{\Gamma}) = N''(\widetilde{\Gamma})\cdot \Gamma_1 \Gamma_2,$$
where $N''(\widetilde{\Gamma})$ is the number of ways to cut an edge of $\widetilde{\Gamma}$
to get $\Gamma_1$ and $\Gamma_2$.
It suffices to show that
\be
\sum_{\widetilde{\Gamma}\in \cG''_\Gamma} N''(\widetilde{\Gamma})\cdot
\frac{1}{|\Aut(\widetilde{\Gamma})|}
= \frac{n_1}{|\Aut(\Gamma_1)|}  \cdot \frac{n_2}{|\Aut(\Gamma_2)|} .
\ee
Consider the set of the choices of picking one external edge from each of $\Gamma_1$ and $\Gamma_2$.
This set has $n_1n_2$ elements.
We partition this set into disjoint union of $H''_{\widetilde{\Gamma}}$
consisting of those choices for which one can get $\widetilde{\Gamma}$
by gluing the pair of external edges,
where $\widetilde{\Gamma}$ runs over the set $\cG''_\Gamma$ of all possibilities.
Therefore,
\be
\sum_{\widetilde{\Gamma}\in \cG_\Gamma''} |H''_{\widetilde{\Gamma}}| = n_1n_2.
\ee
Now we show that
\be
|H''_{\widetilde{\Gamma}}| = N''(\widetilde{\Gamma})\cdot
\frac{|\Aut(\Gamma_1)| \cdot |\Aut(\Gamma_2)|}{|\Aut(\widetilde{\Gamma})|},
\ee
or equivalently,
\be
\frac{|H''_{\widetilde{\Gamma}}|}{|\Aut(\Gamma_1)|\cdot|\Aut(\Gamma_2)|} =
\frac{N''(\widetilde{\Gamma})}{|\Aut(\widetilde{\Gamma})|}.
\ee

Let $\Gamma''$ be the graph  obtained from $\widetilde{\Gamma}$ by
changing the internal edge that we cut to get $\Gamma$ into a dotted edge,
and $\Gamma'$ be the disconnected graph obtained from $\Gamma''$
by cutting the dotted edge. Then the equality
\be\label{eq-claim2}
\frac{|H''_{\widetilde{\Gamma}}|}{|\Aut(\Gamma_1)|\cdot|\Aut(\Gamma_2)|} =
\frac{1}{|\Aut(\Gamma')|}=
\frac{1}{|\Aut(\Gamma'')|}=
\frac{N''(\widetilde{\Gamma})}{|\Aut(\widetilde{\Gamma})|}
\ee
can be proved using exactly the same method as the proof of \eqref{eq-claim}.
Suppose that the two external edges
to be glued together are incident at $v_1$ in $\Gamma_1$ and $v_2$ in $\Gamma_2$ respectively.
Denote by $v=(v_1, v_2)$ the ordered pair of these two vertices.
Then the subgroups $\Aut_{v_i}(\Gamma_i)$,
$\Aut_{v}(\widetilde\Gamma)$ and the orbits $O_{v_i}$,
$\widetilde O_{v}$ are similarly defined. Then the following equalities hold:
\ben
&&\frac{|\Aut(\Gamma_i)|}{|\Aut_{v_i}(\Gamma_i)|}=|O_{v_i}|,
\quad\frac{|\Aut_{v_1}(\Gamma_1)|\cdot |\Aut_{v_1}(\Gamma_1)|}{|\Aut(\Gamma')|}=\delta'',\\
&&\frac{|\Aut(\widetilde\Gamma)|}{|\Aut_{v}(\widetilde\Gamma)|}=|\widetilde O_{v}|,
\quad\frac{|\Aut_v(\widetilde\Gamma)|}{|\Aut(\Gamma'')|}=\delta'',
\een
where $\delta''=\frac{1}{2}$ if $\Gamma_1=\Gamma_2$ and $v_1\in O_{v_2}$
under this identification, and $\delta''=1$ otherwise. Furthermore, we have
\ben
|O_{v_1}|\cdot|O_{v_2}|\cdot\delta''=|H''_{\widetilde{\Gamma}}|,
\quad |\widetilde O_{v}|\cdot\delta''=N''(\widetilde{\Gamma}).
\een
These proves \eqref{eq-claim2}, therefore the relation \eqref{eqn-J} holds.

The equalities \eqref{eq0}, \eqref{eqn-I}, \eqref{eqn-J} completes the proof of the theorem.
\end{proof}

Furthermore, the right hand side of \eqref{eq-thm1} can be rewritten in terms of operators $\cD$.
To explain this we need the following observation.

\begin{Lemma}\label{lem1}
For $2g-2+n>0$, we have $\cD\widehat{\cF}_{g,n}=(n+1)\widehat{\cF}_{g,n+1}$.
\end{Lemma}

\begin{proof}
The proof is similar to the proof of Theorem \ref{thm1}.
First the graphs appearing in the two sums are the same,
since we've already taken all possible unstable contractions into
consideration in the definition of the operator $\cD$.
Thus we only need to check the coefficients.

Let $\Gamma\in\cG_{g,n+1}^c$ be a stable graph,
and $\cG_\Gamma$ be the subset of $\cG_{g,n}^c$ consisting of graphs
obtained from $\Gamma$ by removing an external edge and stabilizing.
We only need to show
\ben
\cD'(\sum_{\widetilde{\Gamma}\in\cG_\Gamma}\frac{1}{|\Aut(\widetilde\Gamma)|}\widetilde\Gamma)
=(n+1)\cdot\frac{1}{|\Aut(\Gamma)|}\Gamma,
\een
where $D'(\Gamma)=N(\widetilde{\Gamma})\cdot\Gamma$ with $N(\widetilde{\Gamma})$ the number of ways
to obtain $\Gamma$ from $\tilde\Gamma$ via one of the operations in $\cD=\pd+\gamma$.
Then it suffices to show
\be
\sum_{\widetilde{\Gamma}\in\cG_\Gamma}N(\widetilde{\Gamma})\cdot
\frac{|\Aut(\Gamma)|}{|\Aut(\widetilde\Gamma)|}=n+1.
\ee

Consider the set of the choices of picking an external edge of $\Gamma$,
and partition it to subsets $H_{\widetilde\Gamma}$ consisting of those choices
for which one can obtain $\widetilde\Gamma$ after removing this edge and stabilizing.
Clearly
\ben
\sum_{\widetilde\Gamma\in\cG_\Gamma}|H_{\widetilde\Gamma}|=\binom{n+1}{1}=n+1.
\een
Now we show that
\ben
N(\widetilde{\Gamma})\cdot
\frac{|\Aut(\Gamma)|}{|\Aut(\widetilde\Gamma)|}=|H_{\widetilde\Gamma}|,
\een
or equivalently,
\ben
\frac{|H_{\widetilde\Gamma}|}{|\Aut(\Gamma)|}=
\frac{N(\widetilde{\Gamma})}{|\Aut(\widetilde\Gamma)|}.
\een
Denote by $\Gamma'$ the graph obtained from $\Gamma$ by changing the external edge
that we remove to a dotted edge. We will show that
\be\label{eq-claim3}
\frac{|H_{\widetilde\Gamma}|}{|\Aut(\Gamma)|}=
\frac{1}{|\Aut(\Gamma')|}=
\frac{N(\widetilde{\Gamma})}{|\Aut(\widetilde\Gamma)|}.
\ee
Let $v$ be the vertex of $\Gamma$ and $\Gamma'$ that the dotted edge attaches to. The first equality follows from
\ben
\frac{|\Aut(\Gamma)|}{|\Aut_v(\Gamma)|}=|O_v|, \quad \frac{|\Aut_v(\Gamma)|}{|\Aut(\Gamma')|}=\frac{k!}{(k-1)!};
\een
\ben
|H_{\widetilde\Gamma}|=\binom{k}{1}\cdot|O_v|,
\een
where $\Aut_v(\Gamma)$ is the set of automorphisms of $\Gamma$ that preserve $v$,
$O_v$ is the orbit of $v$ in $\Gamma$ under $\Aut(\Gamma)$,
and $k$ is the number of external edges attaching to $v$ in $\Gamma$.

Let $\tilde v$ be the corresponding vertex of $v$ in $\widetilde\Gamma$, that is,
$\tilde v=v$ if the graph is still stable after removing the dotted edge,
and $\tilde v$ the vertex that $v$ contracts to if the graph is unstable.
The second equality in \eqref{eq-claim3} follows from
\ben
\frac{|\Aut(\widetilde\Gamma)|}{|\Aut_{\tilde v}(\widetilde\Gamma)|}=|\widetilde{O}_{\tilde v}|,
\quad |\Aut_{\tilde v}(\widetilde\Gamma)|=|\Aut(\Gamma')|;
\een
\ben
N(\widetilde\Gamma)=|\widetilde{O}_{\tilde v}|,
\een
where $\Aut_{\tilde v}(\widetilde\Gamma)$ is the set of automorphisms of $\widetilde\Gamma$
that preserve $\tilde v$,
and $\widetilde{O}_{\tilde v}$ is the orbit of $\tilde v$
in $\widetilde\Gamma$ under $\Aut(\widetilde\Gamma)$.

Therefore the equality \eqref{eq-claim3} holds, which proves the lemma.
\end{proof}

Even though so far we have not defined $\widehat{\cF}_{1}$, $\widehat{\cF}_{0,2}$, and $\widehat{\cF}_{0,1}$,
it will be convenient for later use to define $\cD\widehat{\cF}_{1}$, $\cD\widehat{\cF}_{0,2}$, $\cD \cD\widehat{\cF}_{0,1}$
by formally applying Lemma \ref{lem1} as follows:
\be\label{unstable}
\begin{split}
&\cD\widehat{\cF}_{1}:=\widehat{\cF}_{1,1},\\
&\cD\widehat{\cF}_{0,2}:=3\widehat{\cF}_{0,3},\\
&\cD \cD\widehat{\cF}_{0,1}:=6\widehat{\cF}_{0,3}.
\end{split}
\ee

Now using Lemma \ref{lem1} and conventions in \eqref{unstable},
our Theorem \ref{thm1} can be rewritten in the following way.

\begin{Theorem}\label{thm2}
For $2g-2+n>0$, we have
\be\label{eq-thm2}
K\widehat{\cF}_{g,n}=\frac{1}{2}(\cD \cD\widehat{\cF}_{g-1,n}+\sum_{\substack{g_1+g_2=g,\\n_1+n_2=n}}\cD\widehat{\cF}_{g_1,n_1}\cD\widehat{\cF}_{g_2,n_2}).
\ee
In particular, by taking $n=0$ we get a recursion relation for the abstract free energy:
\be\label{thm-free}
K\widehat{\cF}_g=\frac{1}{2}(\cD \pd\widehat{\cF}_{g-1}+\sum_{r=1}^{g-1}\pd\widehat{\cF}_{r}\pd\widehat{\cF}_{g-r}).
\ee
\end{Theorem}

\section{A Generalization to Labelled Graphs}\label{sec3}

In this section, we generalize the construction of Section \ref{sec2}
to the case where all the half-edges of graphs may have labels.

\subsection{Diagrammatics for labelled graphs}

\label{sec:N}

Fix a positive integer $N$,
for a stable graph $\Gamma\in\cG_{g,n}$ we label some indices on the edges as follows.
For an external edge, we label an index in $\{1,2,\cdots,N\}$; and for an internal edge,
we label an index in $\{1,2,\cdots,N\}$ at each of its end points.
In other words, we label each half edge incident at a vertex.

As generalizations of the edge-cutting operator $K$
defined in Section \ref{sec:Operators-1},
we define the edge-cutting operators $K_{ij}$ as the operator that cuts
off an internal edges with two labels $i$ and $j$, then sum over all such internal edges.

\begin{Example}
Take $N=2$, then we have the following examples:
\ben
&&\begin{tikzpicture}
\node [align=center,align=center] at (-0.5,0) {$K_{11}$};
\draw (1.2,0) circle [radius=0.2];
\draw (0.4,0) circle [radius=0.2];
\draw (0.6,0)--(1,0);
\draw (1.36,0.1) .. controls (1.7,0.2) and (1.7,-0.2) ..  (1.36,-0.1);
\draw (0.24,0.1) .. controls (-0.1,0.2) and (-0.1,-0.2) ..  (0.24,-0.1);
\node [align=center,align=center] at (1.2,0) {$0$};
\node [align=center,align=center] at (0.4,0) {$0$};
\node [above,align=center] at (0.6,0) {$1$};
\node [above,align=center] at (1,0) {$2$};
\node [above,align=center] at (0.24,0.1) {$1$};
\node [below,align=center] at (0.24,-0.1) {$1$};
\node [above,align=center] at (1.36,0.1) {$1$};
\node [below,align=center] at (1.36,-0.1) {$1$};
\draw (3.7,0) circle [radius=0.2];
\draw (2.9,0) circle [radius=0.2];
\draw (3.1,0)--(3.5,0);
\draw (2.5,0.15)--(2.74,0.1);
\draw (2.5,-0.15)--(2.74,-0.1);
\draw (3.1,0)--(3.5,0);
\draw (3.86,0.1) .. controls (4.2,0.2) and (4.2,-0.2) ..  (3.86,-0.1);
\node [align=center,align=center] at (3.7,0) {$0$};
\node [align=center,align=center] at (2.9,0) {$0$};
\node [above,align=center] at (3.1,0) {$1$};
\node [above,align=center] at (3.5,0) {$2$};
\node [above,align=center] at (2.6,0.1) {$1$};
\node [below,align=center] at (2.6,-0.1) {$1$};
\node [above,align=center] at (3.86,0.1) {$1$};
\node [below,align=center] at (3.86,-0.1) {$1$};
\node [align=center,align=center] at (2,0) {$=$};
\node [align=center,align=center] at (4.5,0) {$+$};
\draw (3.7+2.4,0) circle [radius=0.2];
\draw (2.9+2.4,0) circle [radius=0.2];
\draw (3.1+2.4,0)--(3.5+2.4,0);
\draw (2.5+2.4,0.15)--(2.74+2.4,0.1);
\draw (2.5+2.4,-0.15)--(2.74+2.4,-0.1);
\draw (3.1+2.4,0)--(3.5+2.4,0);
\draw (3.86+2.4,0.1) .. controls (4.2+2.4,0.2) and (4.2+2.4,-0.2) ..  (3.86+2.4,-0.1);
\node [align=center,align=center] at (3.7+2.4,0) {$0$};
\node [align=center,align=center] at (2.9+2.4,0) {$0$};
\node [above,align=center] at (3.1+2.4,0) {$2$};
\node [above,align=center] at (3.5+2.4,0) {$1$};
\node [above,align=center] at (2.6+2.4,0.1) {$1$};
\node [below,align=center] at (2.6+2.4,-0.1) {$1$};
\node [above,align=center] at (3.86+2.4,0.1) {$1$};
\node [below,align=center] at (3.86+2.4,-0.1) {$1$};
\end{tikzpicture},\\
&&\begin{tikzpicture}
\node [align=center,align=center] at (-0.5,0) {$K_{12}$};
\draw (1.2,0) circle [radius=0.2];
\draw (0.4,0) circle [radius=0.2];
\draw (0.6,0)--(1,0);
\draw (1.36,0.1) .. controls (1.7,0.2) and (1.7,-0.2) ..  (1.36,-0.1);
\draw (0.24,0.1) .. controls (-0.1,0.2) and (-0.1,-0.2) ..  (0.24,-0.1);
\node [align=center,align=center] at (1.2,0) {$0$};
\node [align=center,align=center] at (0.4,0) {$0$};
\node [above,align=center] at (0.6,0) {$1$};
\node [above,align=center] at (1,0) {$2$};
\node [above,align=center] at (0.24,0.1) {$1$};
\node [below,align=center] at (0.24,-0.1) {$1$};
\node [above,align=center] at (1.36,0.1) {$1$};
\node [below,align=center] at (1.36,-0.1) {$1$};
\draw (3.8,0) circle [radius=0.2];
\draw (2.8,0) circle [radius=0.2];
\draw (3.0,0)--(3.2,0);
\draw (3.4,0)--(3.6,0);
\draw (3.96,0.1) .. controls (4.3,0.2) and (4.3,-0.2) ..  (3.96,-0.1);
\draw (2.64,0.1) .. controls (2.3,0.2) and (2.3,-0.2) ..  (2.64,-0.1);
\node [align=center,align=center] at (3.8,0) {$0$};
\node [align=center,align=center] at (2.8,0) {$0$};
\node [above,align=center] at (3.1,0) {$1$};
\node [above,align=center] at (3.5,0) {$2$};
\node [above,align=center] at (2.64,0.1) {$1$};
\node [below,align=center] at (2.64,-0.1) {$1$};
\node [above,align=center] at (3.96,0.1) {$1$};
\node [below,align=center] at (3.96,-0.1) {$1$};
\node [align=center,align=center] at (2,0) {$=$};
\end{tikzpicture},\\
&&\begin{tikzpicture}
\node [align=center,align=center] at (-0.5,0) {$K_{22}$};
\draw (1.2,0) circle [radius=0.2];
\draw (0.4,0) circle [radius=0.2];
\draw (0.6,0)--(1,0);
\draw (1.36,0.1) .. controls (1.7,0.2) and (1.7,-0.2) ..  (1.36,-0.1);
\draw (0.24,0.1) .. controls (-0.1,0.2) and (-0.1,-0.2) ..  (0.24,-0.1);
\node [align=center,align=center] at (1.2,0) {$0$};
\node [align=center,align=center] at (0.4,0) {$0$};
\node [above,align=center] at (0.6,0) {$1$};
\node [above,align=center] at (1,0) {$2$};
\node [above,align=center] at (0.24,0.1) {$1$};
\node [below,align=center] at (0.24,-0.1) {$1$};
\node [above,align=center] at (1.36,0.1) {$1$};
\node [below,align=center] at (1.36,-0.1) {$1$};
\node [align=center,align=center] at (2.1,0) {$=0$};
\end{tikzpicture}.
\een
\end{Example}

Similarly, we also have operators $\partial_{i}$ and $\gamma_{i}$
generalizing the operators $\pd$ and $\gamma$ in Section \ref{sec:Operators-1} respectively.
The operator $\partial_{i}$ has two parts, one is to add an external edge labelled by $i$,
and the other is to break up an internal edge and attach a $3$-pointed Riemann sphere
such that the new external edge has label $i$.
The operator $\gamma_{i}$ is to attach a $3$-pointed Riemann sphere to an external edge $e$,
and move the label of $e$ to be the label of one of the new external edge,
while the other new external edge has label $i$.
The two labels of the new internal edge can be chosen arbitrarily.

\begin{Example}
For $N=2$, here are some examples of the operator $\partial_{i}$.
\ben
&& \begin{tikzpicture}
\node [align=center,align=center] at (0.4,0) {$\partial_{1}$};
\draw (1,0) circle [radius=0.2];
\draw (1.2,0)--(1.5,0);
\node [align=center,align=center] at (1,0) {$1$};
\node [above,align=center] at (1.5,0) {$2$};
\node [align=center,align=center] at (1.9,0) {$=$};
\draw (2.5,0) circle [radius=0.2];
\draw (2.67,0.1)--(2.9,0.15);
\draw (2.67,-0.1)--(2.9,-0.15);
\node [align=center,align=center] at (2.5,0) {$1$};
\node [align=center,right] at (2.9,0.15) {$2$};
\node [align=center,right] at (2.9,-0.15) {$1$};
\end{tikzpicture},
\\
&&\begin{tikzpicture}
\node [align=center,align=center] at (0.4,0) {$\partial_{1}$};
\draw (1.2,0) circle [radius=0.2];
\draw (1.4,0)--(1.7,0);
\draw (1.04,0.1) .. controls (0.7,0.2) and (0.7,-0.2) ..  (1.04,-0.1);
\node [align=center,align=center] at (1.2,0) {$0$};
\node [above,align=center] at (1.04,0.1) {$2$};
\node [below,align=center] at (1.04,-0.1) {$2$};
\node [align=center,above] at (1.7,0) {$1$};
\node [align=center,right] at (1.9,0) {$=$};
\draw (3,0) circle [radius=0.2];
\draw (3.17,0.1)--(3.4,0.15);
\draw (3.17,-0.1)--(3.4,-0.15);
\draw (2.84,0.1) .. controls (2.5,0.2) and (2.5,-0.2) ..  (2.84,-0.1);
\node [align=center,align=center] at (3,0) {$0$};
\node [above,align=center] at (2.84,0.1) {$2$};
\node [below,align=center] at (2.84,-0.1) {$2$};
\node [align=center,right] at (3.4,0.15) {$1$};
\node [align=center,right] at (3.4,-0.15) {$1$};
\node [align=center,right] at (3.8,0) {$+$};
\draw (5,0) circle [radius=0.2];
\draw (4.5,0)--(4.8,0);
\draw (5.18,0.07)--(5.62,0.07);
\draw (5.18,-0.07)--(5.62,-0.07);
\draw (6,0)--(6.3,0);
\draw (5.8,0) circle [radius=0.2];
\node [align=center,align=center] at (5,0) {$0$};
\node [align=center,align=center] at (5.8,0) {$0$};
\node [above,align=center] at (4.5,0) {$1$};
\node [above,align=center] at (6.3,0) {$1$};
\node [above,align=center] at (5.18,0.07) {$1$};
\node [above,align=center] at (5.62,0.07) {$2$};
\node [below,align=center] at (5.18,-0.07) {$1$};
\node [below,align=center] at (5.62,-0.07) {$2$};
\node [align=center,align=center] at (6.7,0) {$+$};
\draw (7.5,0) circle [radius=0.2];
\draw (7,0)--(7.3,0);
\draw (7.68,0.07)--(8.12,0.07);
\draw (7.68,-0.07)--(8.12,-0.07);
\draw (8.5,0)--(8.8,0);
\draw (8.3,0) circle [radius=0.2];
\node [align=center,align=center] at (7.5,0) {$0$};
\node [align=center,align=center] at (8.3,0) {$0$};
\node [above,align=center] at (7,0) {$1$};
\node [above,align=center] at (8.8,0) {$1$};
\node [above,align=center] at (7.68,0.07) {$2$};
\node [above,align=center] at (8.12,0.07) {$2$};
\node [below,align=center] at (7.68,-0.07) {$2$};
\node [below,align=center] at (8.12,-0.07) {$2$};
\node [align=center,align=center] at (9.2,0) {$+$};
\node [align=center,align=center] at (9.5,0) {$2$};
\draw (10.3,0) circle [radius=0.2];
\draw (9.8,0)--(10.1,0);
\draw (10.48,0.07)--(10.92,0.07);
\draw (10.48,-0.07)--(10.92,-0.07);
\draw (11.3,0)--(11.6,0);
\draw (11.1,0) circle [radius=0.2];
\node [align=center,align=center] at (10.3,0) {$0$};
\node [align=center,align=center] at (11.1,0) {$0$};
\node [above,align=center] at (9.8,0) {$1$};
\node [above,align=center] at (11.6,0) {$1$};
\node [above,align=center] at (10.48,0.07) {$1$};
\node [above,align=center] at (10.92,0.07) {$2$};
\node [below,align=center] at (10.48,-0.07) {$2$};
\node [below,align=center] at (10.92,-0.07) {$2$};
\end{tikzpicture}.
\een
\end{Example}

\begin{Example}
For $N=2$, here are some examples of the operator $\gamma_i$.
\ben
&&\begin{tikzpicture}
\node [align=center,align=center] at (0.4,0) {$\gamma_1$};
\draw (1,0) circle [radius=0.2];
\draw (1.2,0)--(1.5,0);
\node [align=center,align=center] at (1,0) {$1$};
\node [above,align=center] at (1.5,0) {$2$};
\node [align=center,align=center] at (1.9,0) {$=$};
\draw (2.5,0) circle [radius=0.2];
\draw (3.3,0) circle [radius=0.2];
\draw (2.7,0)--(3.1,0);
\draw (3.47,0.1)--(3.7,0.15);
\draw (3.47,-0.1)--(3.7,-0.15);
\node [align=center,align=center] at (2.5,0) {$1$};
\node [align=center,align=center] at (3.3,0) {$0$};
\node [align=center,right] at (3.7,0.15) {$2$};
\node [align=center,right] at (3.7,-0.15) {$1$};
\node [above,align=center] at (2.7,0) {$1$};
\node [above,align=center] at (3.1,0) {$1$};
\node [align=center,align=center] at (4.4,0) {$+$};
\draw (5,0) circle [radius=0.2];
\draw (5.8,0) circle [radius=0.2];
\draw (5.2,0)--(5.6,0);
\draw (5.97,0.1)--(6.2,0.15);
\draw (5.97,-0.1)--(6.2,-0.15);
\node [align=center,align=center] at (5,0) {$1$};
\node [align=center,align=center] at (5.8,0) {$0$};
\node [align=center,right] at (6.2,0.15) {$2$};
\node [align=center,right] at (6.2,-0.15) {$1$};
\node [above,align=center] at (5.2,0) {$1$};
\node [above,align=center] at (5.6,0) {$2$};
\node [align=center,align=center] at (6.9,0) {$+$};
\draw (7.5,0) circle [radius=0.2];
\draw (8.3,0) circle [radius=0.2];
\draw (7.7,0)--(8.1,0);
\draw (8.47,0.1)--(8.7,0.15);
\draw (8.47,-0.1)--(8.7,-0.15);
\node [align=center,align=center] at (7.5,0) {$1$};
\node [align=center,align=center] at (8.3,0) {$0$};
\node [align=center,right] at (8.7,0.15) {$2$};
\node [align=center,right] at (8.7,-0.15) {$1$};
\node [above,align=center] at (7.7,0) {$2$};
\node [above,align=center] at (8.1,0) {$1$};
\node [align=center,align=center] at (9.4,0) {$+$};
\draw (10,0) circle [radius=0.2];
\draw (10.8,0) circle [radius=0.2];
\draw (10.2,0)--(10.6,0);
\draw (10.97,0.1)--(11.2,0.15);
\draw (10.97,-0.1)--(11.2,-0.15);
\node [align=center,align=center] at (10,0) {$1$};
\node [align=center,align=center] at (10.8,0) {$0$};
\node [align=center,right] at (11.2,0.15) {$2$};
\node [align=center,right] at (11.2,-0.15) {$1$};
\node [above,align=center] at (10.2,0) {$2$};
\node [above,align=center] at (10.6,0) {$2$};
\end{tikzpicture},
\\
&&\begin{tikzpicture}
\node [align=center,align=center] at (0.5,0) {$\gamma_1$};
\draw (1.2,0) circle [radius=0.2];
\draw (1.4,0)--(1.7,0);
\draw (1.04,0.1) .. controls (0.7,0.2) and (0.7,-0.2) ..  (1.04,-0.1);
\node [align=center,align=center] at (1.2,0) {$0$};
\node [above,align=center] at (1.04,0.1) {$2$};
\node [below,align=center] at (1.04,-0.1) {$2$};
\node [align=center,above] at (1.7,0) {$1$};
\node [align=center,right] at (1.8,0) {$=$};
\draw (3.5,0) circle [radius=0.2];
\draw (2.8,0) circle [radius=0.2];
\draw (3,0)--(3.3,0);
\draw (3.66,0.1)--(3.9,0.15);
\draw (3.66,-0.1)--(3.9,-0.15);
\draw (2.64,0.1) .. controls (2.3,0.2) and (2.3,-0.2) ..  (2.64,-0.1);
\node [align=center,align=center] at (3.5,0) {$0$};
\node [align=center,align=center] at (2.8,0) {$0$};
\node [above,align=center] at (2.64,0.1) {$2$};
\node [below,align=center] at (2.64,-0.1) {$2$};
\node [above,align=center] at (3,0) {$1$};
\node [above,align=center] at (3.3,0) {$1$};
\node [align=center,right] at (3.9,0.15) {$1$};
\node [align=center,right] at (3.9,-0.15) {$1$};
\node [align=center,align=center] at (4.5,0) {$+$};
\draw (6,0) circle [radius=0.2];
\draw (5.3,0) circle [radius=0.2];
\draw (5.5,0)--(5.8,0);
\draw (6.16,0.1)--(6.4,0.15);
\draw (6.16,-0.1)--(6.4,-0.15);
\draw (5.14,0.1) .. controls (4.8,0.2) and (4.8,-0.2) ..  (5.14,-0.1);
\node [align=center,align=center] at (6,0) {$0$};
\node [align=center,align=center] at (5.3,0) {$0$};
\node [above,align=center] at (5.14,0.1) {$2$};
\node [below,align=center] at (5.14,-0.1) {$2$};
\node [above,align=center] at (5.5,0) {$1$};
\node [above,align=center] at (5.8,0) {$2$};
\node [align=center,right] at (6.4,0.15) {$1$};
\node [align=center,right] at (6.4,-0.15) {$1$};
\node [align=center,align=center] at (7,0) {$+$};
\draw (8.5,0) circle [radius=0.2];
\draw (7.8,0) circle [radius=0.2];
\draw (8,0)--(8.3,0);
\draw (8.66,0.1)--(8.9,0.15);
\draw (8.66,-0.1)--(8.9,-0.15);
\draw (7.64,0.1) .. controls (7.3,0.2) and (7.3,-0.2) ..  (7.64,-0.1);
\node [align=center,align=center] at (8.5,0) {$0$};
\node [align=center,align=center] at (7.8,0) {$0$};
\node [above,align=center] at (7.64,0.1) {$2$};
\node [below,align=center] at (7.64,-0.1) {$2$};
\node [above,align=center] at (8,0) {$2$};
\node [above,align=center] at (8.3,0) {$1$};
\node [align=center,right] at (8.9,0.15) {$1$};
\node [align=center,right] at (8.9,-0.15) {$1$};
\node [align=center,align=center] at (9.5,0) {$+$};
\draw (11,0) circle [radius=0.2];
\draw (10.3,0) circle [radius=0.2];
\draw (10.5,0)--(10.8,0);
\draw (11.16,0.1)--(11.4,0.15);
\draw (11.16,-0.1)--(11.4,-0.15);
\draw (10.14,0.1) .. controls (9.8,0.2) and (9.8,-0.2) ..  (10.14,-0.1);
\node [align=center,align=center] at (11,0) {$0$};
\node [align=center,align=center] at (10.3,0) {$0$};
\node [above,align=center] at (10.14,0.1) {$2$};
\node [below,align=center] at (10.14,-0.1) {$2$};
\node [above,align=center] at (10.5,0) {$2$};
\node [above,align=center] at (10.8,0) {$2$};
\node [align=center,right] at (11.4,0.15) {$1$};
\node [align=center,right] at (11.4,-0.15) {$1$};
\end{tikzpicture}.
\een
\end{Example}

Now let $\cG_{g,n}(N)$ be the set of all labelled stable graphs of genus $g$ with $n$ external edges
and $N$ possible indices,
and $\cG_{g,n}^c(N)$ be the set of connected ones.
In the labelled case we also define the abstract free energy in the same way.

\begin{Definition}
For $g\geq 2$,   the abstract free energy $\widehat{\cF}_g$  is defined by:
\be
\widehat{\cF}_g=\sum_{\Gamma\in\cG_{g,0}^c(N)}\frac{1}{|\Aut(\Gamma)|}\Gamma.
\ee
\end{Definition}

For example, take $N=2$, then we have
\ben
\begin{tikzpicture}
\node [align=center,align=center] at (0.2,0) {$\widehat{\cF}_2=$};
\draw (1,0) circle [radius=0.2];
\node [align=center,align=center] at (1,0) {$2$};
\node [align=center,align=center] at (1.5,0) {$+$};
\node [align=center,align=center] at (2,0) {$\frac{1}{2}$};
\draw (2.5,0) circle [radius=0.2];
\draw (3.3,0) circle [radius=0.2];
\draw (2.7,0)--(3.1,0);
\node [align=center,align=center] at (2.5,0) {$1$};
\node [align=center,align=center] at (3.3,0) {$1$};
\node [above,align=center] at (2.7,0) {$1$};
\node [above,align=center] at (3.1,0) {$1$};
\node [align=center,align=center] at (4,0) {$+$};
\node [align=center,align=center] at (4.5,0) {$\frac{1}{2}$};
\draw (5,0) circle [radius=0.2];
\draw (5.8,0) circle [radius=0.2];
\draw (5.2,0)--(5.6,0);
\node [align=center,align=center] at (5,0) {$1$};
\node [align=center,align=center] at (5.8,0) {$1$};
\node [above,align=center] at (5.2,0) {$2$};
\node [above,align=center] at (5.6,0) {$2$};
\node [align=center,align=center] at (6.5,0) {$+$};
\draw (7,0) circle [radius=0.2];
\draw (7.8,0) circle [radius=0.2];
\draw (7.2,0)--(7.6,0);
\node [align=center,align=center] at (7,0) {$1$};
\node [align=center,align=center] at (7.8,0) {$1$};
\node [above,align=center] at (7.2,0) {$1$};
\node [above,align=center] at (7.6,0) {$2$};
\end{tikzpicture}
\\
\begin{tikzpicture}
\node [align=center,align=center] at (1.0,0) {$+$};
\node [align=center,align=center] at (1.5,0) {$\frac{1}{2}$};
\draw (2.3,0) circle [radius=0.2];
\node [above,align=center] at (2.14,0.1) {$1$};
\node [below,align=center] at (2.14,-0.1) {$1$};
\draw (2.14,0.1) .. controls (1.8,0.2) and (1.8,-0.2) ..  (2.14,-0.1);
\node [align=center,align=center] at (2.3,0) {$1$};
\node [align=center,align=center] at (3,0) {$+$};
\node [align=center,align=center] at (3.5,0) {$\frac{1}{2}$};
\draw (4.3,0) circle [radius=0.2];
\draw (4.14,0.1) .. controls (3.8,0.2) and (3.8,-0.2) ..  (4.14,-0.1);
\node [above,align=center] at (4.14,0.1) {$2$};
\node [below,align=center] at (4.14,-0.1) {$2$};
\node [align=center,align=center] at (4.3,0) {$1$};
\node [align=center,align=center] at (5,0) {$+$};
\draw (5.8,0) circle [radius=0.2];
\draw (5.64,0.1) .. controls (5.3,0.2) and (5.3,-0.2) ..  (5.64,-0.1);
\node [above,align=center] at (5.64,0.1) {$1$};
\node [below,align=center] at (5.64,-0.1) {$2$};
\node [align=center,align=center] at (5.8,0) {$1$};
\node [align=center,align=center] at (6.8,0) {$+\cdots$};
\end{tikzpicture}
\een
The complete expressions for $\widehat{\cF}_{2}$ can be obtained by adding all possible labels on the graphs of $\widehat{\cF}_2$ for the $N=1$ case and modifying the number of automorphisms.

In this case we also have the notions of abstract $n$-point functions.
Denote by $\cG_{g;l_1,\cdots,l_N}^c(N)\subset\cG_{g,n}^c(N)$ the subset
consisting of all connected labelled
stable graphs of genus $g$, with $l_j$ external edges
labelled by $j$ for every $j\in\{1,\cdots,N\}$
(we have $l_1+\cdots+\l_N=n$).
Then we define the abstract $n$-point function $\wcF_{g;l_1,\cdots,l_N}$ to be
the linear combination
\be
\wcF_{g;l_1,\cdots,l_N}:=
\sum_{\Gamma\in \cG_{g;l_1,\cdots,l_N}^c(N)}
\frac{1}{|\Aut(\Gamma)|}\Gamma,
\ee
for $2g-2+\sum l_j>0$.

\subsection{A recursion relation for the abstract free energy}
\label{sec:Recursion-N}

In the labelled case, we also have a recursion relation generalizing \eqref{thm-free}.
Here we define $\cD_{i}=\partial_{i}+\gamma_i$.

\begin{Theorem}\label{N-graph}
For $g\geq 2$, we have
\be
K_{ij}\widehat{\cF}_g=\cD_i\partial_j\widehat{\cF}_{g-1}+\sum_{r=1}^{g-1}
\partial_i\widehat{\cF}_r\partial_j\widehat{\cF}_{g-r}
\ee
for $i\not= j$, and
\be
K_{ii}\widehat{\cF}_g=\frac{1}{2}\big(\cD_i\partial_i\widehat{\cF}_{g-1}+\sum_{r=1}^{g-1}
\partial_i\widehat{\cF}_r\partial_i\widehat{\cF}_{g-r}\big).
\ee
\end{Theorem}

The proof is an easy modification of the proof of Theorem \ref{thm2} and is omitted.
For the case of $g=2$ we use the following convention£º
\be
\begin{tikzpicture}
\node [align=center,align=center] at (-0.2,0) {$\partial_{j}\widehat{\cF}_1:=$};
\draw (1,0) circle [radius=0.2];
\draw (1.2,0)--(1.45,0);
\node [align=center,align=center] at (1,0) {$1$};
\node [above,align=center] at (1.5,0) {$j$};
\node [align=center,align=center] at (2.1,0) {$+\sum_k\frac{1}{2}$};
\draw (3.2,0) circle [radius=0.2];
\draw (3.4,0)--(3.7,0);
\draw (3.04,0.1) .. controls (2.7,0.2) and (2.7,-0.2) ..  (3.04,-0.1);
\node [align=center,align=center] at (3.2,0) {$0$};
\node [above,align=center] at (3.04,0.1) {$k$};
\node [below,align=center] at (3.04,-0.1) {$k$};
\node [align=center,above] at (3.7,0) {$j$};
\node [align=center,align=center] at (4.5,0) {$+\sum_{k\not= l}$};
\draw (5.7,0) circle [radius=0.2];
\draw (5.9,0)--(6.2,0);
\draw (5.54,0.1) .. controls (5.2,0.2) and (5.2,-0.2) ..  (5.54,-0.1);
\node [align=center,align=center] at (5.7,0) {$0$};
\node [above,align=center] at (5.54,0.1) {$k$};
\node [below,align=center] at (5.54,-0.1) {$l$};
\node [align=center,above] at (6.2,0) {$j$};
\end{tikzpicture}.
\ee

The recursion relations for the abstract $n$-point functions
can also be generalized to the case of labelled graphs.
The generalizations of
Lemma \ref{lem1} and Theorem \ref{thm2}
are the following:

\begin{Lemma}\label{lem-N-rec-npt}
For $2g-2+\sum l_j>0$,
we have
\be
\cD_j\wcF_{g;l_1,\cdots,l_N}=(l_j+1)\cdot\wcF_{g;l_1,\cdots,l_j+1,\cdots,l_N}.
\ee
\end{Lemma}

\begin{Theorem}
\label{thm-N-rec-npt}
For $2g-2+\sum l_j>0$, we have
\begin{equation*}
\begin{split}
&K_{ij}\widehat{\cF}_{g;l_1,\cdots,l_N}=\cD_i\cD_j\widehat{\cF}_{g-1;l_1,\cdots,l_N}+
\sum_{\substack{g_1+g_2=g\\p_k+q_k=l_k}}
\cD_i\widehat{\cF}_{g_1;p_1,\cdots,p_N}\cD_j\widehat{\cF}_{g_2;q_1,\cdots,q_N},
\quad i\not=j;\\
&K_{ii}\widehat{\cF}_{g;l_1,\cdots,l_N}=\half\biggl(
\cD_i\cD_i\widehat{\cF}_{g-1;l_1,\cdots,l_N}+
\sum_{\substack{g_1+g_2=g\\p_k+q_k=l_k}}
\cD_i\widehat{\cF}_{g_1;p_1,\cdots,p_N}\cD_i\widehat{\cF}_{g_2;q_1,\cdots,q_N}\biggr).
\end{split}
\end{equation*}
In particular,
for $l_1=\cdots=l_N=0$
these recover the recursion relations in Theorem \ref{N-graph}.
Here we use the convention
\ben
&&\pd_j \wcF_{1;l_1,\cdots,l_N}:=\wcF_{1;l_1,\cdots,l_j+1,\cdots,l_N},
\quad \text{for } l_1=\cdots=l_N=0;\\
&&\cD_j \wcF_{0;l_1,\cdots,l_N}= (l_j+1)\wcF_{0;l_1,\cdots,l_j+1,\cdots,l_N},
\quad \text{for } l_1+\cdots+l_N=2;\\
&&\cD_i\cD_j \wcF_{0;l_1,\cdots,l_N}:=(l_j+1)\cD_i \wcF_{0;l_1,\cdots,l_j+1,\cdots,l_N},
\quad \text{for } l_1+\cdots+l_N=1.\\
\een
\end{Theorem}

\section{Realization of the Abstract Quantum Field Theory}\label{sec4}

In this section we formulate a Feynman rule for the diagrammatics
discussed in the previous sections,
in such way we are able to define a new free energy $\wF$ from a given holomorphic free energy $F$
and a choice of propagator $\kappa$.
We also represent this new field theory as formal Gaussian integrals.
This construction realizes the abstract free energy and its quadratic
recursion relations discussed in preceding sections.

\subsection{Realization of the abstract QFT by Feynman rules}

In this subsection we will present a construction to realize the
abstract quantum field theory defined in Section \ref{sec:Abstract-1} and
Section \ref{sec:N} by Feynman rules that assign the contribution
of a vertex and the contribution of edge:
\be
\Gamma \mapsto \omega_\Gamma = \prod_{v\in V(\Gamma)} \omega_v \cdot
\prod_{e\in E(\Gamma)} \omega_e.
\ee
For $\omega_v$,
we need to fix a sequence of holomorphic functions $F_g$ in $t_1,\dots, t_N$
as input;
for $\omega_e$, we need fix a nondegenerate symmetric matrix $\kappa
= (\kappa_{ij})_{1 \leq i, j \leq N}$ of functions in $t_1, \dots, t_N$
as propagators.

For a given positive integer $N$, let $t_1$,$t_2$,$\cdots$,$t_N$
be the coordinates on an $N$-dimensional vector space,
and $\kappa=(\kappa_{ij})$ be a nondegenerate symmetric matrix of size $N\times N$.
Moreover, we fix a sequence of holomorphic function $F_g(t_1,\cdots,t_N)$ for $g\geq 0$.
Let
\be
F(t; \lambda)=\sum_{g=0}^{\infty}\lambda^{2g-2}F_g(t).
\ee

Now for a labelled stable graph $\Gamma\in\cG_{g,0}^c(N)$,
we associate a polynomial $\omega_\Gamma$ as follows.
Let $V(\Gamma)$, $E^{\ext}(\Gamma)$, and $E(\Gamma)$ be the sets of vertices,
external half edges, and internal edges of $\Gamma$ respectively.
For each $e\in E^{\ext}(\Gamma)$, write $l(e)$ as the label on the external $e$;
and for each $e\in E(\Gamma)$, write $l_1(e)$ and $l_2(e)$ as the two labels on $e$
(the order is not important).
For each $v\in V(\Gamma)$, write $g(v)$ as the genus associated to $v$,
and $\val_i(v)$ as the number of half-edges labelled by $i$ incident at $v$,
then $\sum_{i}\val_i(v)$ is the valence of $v$.
Then $\omega_\Gamma$ is defined by
\be
\omega_{\Gamma}=\prod_{v\in V(\Gamma)}F_{g(v)}^{(\val_1(v),\cdots,\val_N(v))}(t)
\prod_{e\in E(\Gamma)}\kappa_{l_1(e)l_2(e)},
\ee
where $F_{g(v)}^{(\val_1(v),\cdots,\val_N(v))}(t):
=(\frac{\partial}{\partial t_1})^{\val_1(v)}\cdots(\frac{\partial}{\partial t_N})^{\val_N(v)}F_{g(v)}(t)$. For $N=1$, this expression simplifies to
\be\label{omega}
\omega_{\Gamma}=(\prod_{v\in V(\Gamma)}F_{g(v)}^{(\val(v))}(t))\kappa^{|E(\Gamma)|}.
\ee

\begin{Definition}
For $g\geq 2$ we define the free energy to be
\be\label{free-expression}
\wF_g=\sum_{\Gamma\in\cG_{g,0}^c(N)}\frac{1}{|\Aut(\Gamma)|}\omega_\Gamma.
\ee
In particular, the degree of $\kappa$ in the expression of $\wF_g$ is $3g-3$.
\end{Definition}

\begin{Example}
For $N=1$, we have
\be\label{f2}
\wF_2=F_2+\kappa[\frac{1}{2}F_1''+\frac{1}{2}(F_1')^2]+\kappa^2(\frac{1}{8}F_0^{(4)}+\frac{1}{2}F_1'F_0''')+\frac{5}{24}\kappa^3(F_0''')^2.
\ee
And for $N=2$, we have
\be
\begin{split}
\wF_2&=F_2+(\frac{1}{2}F_1^{(2,0)}+\frac{1}{2}(F_1^{(1,0)})^2)\kappa_{11}+(\frac{1}{2}F_1^{(0,2)}+\frac{1}{2}(F_1^{(0,1)})^2)\kappa_{22}
\\
&+(\frac{1}{8}F_0^{(4,0)}+\frac{1}{2}F_1^{(1,0)}F_0^{(3,0)})\kappa_{11}^2+(\frac{1}{8}F_0^{(0,4)}+\frac{1}{2}F_1^{(0,1)}F_0^{(0,3)})\kappa_{22}^2\\
&+(F_1^{(1,1)}+F_1^{(1,0)}F_1^{(0,1)})\kappa_{12}+(\frac{1}{2}F_0^{(2,2)}+F_1^{(1,0)}F_0^{(1,2)}+F_1^{(0,1)}F_0^{(2,1)})\kappa_{12}^2\\
&+(\frac{1}{4}F_0^{(2,2)}+\frac{1}{2}F_1^{(1,0)}F_0^{(1,2)}+\frac{1}{2}F_1^{(0,1)}F_0^{(2,1)})\kappa_{11}\kappa_{22}\\
&+(\frac{1}{2}F_0^{(3,1)}+\frac{3}{2}F_1^{(1,0)}F_0^{(2,1)}+\frac{1}{2}F_1^{(0,1)}F_0^{(3,0)})\kappa_{11}\kappa_{12}\\
&+(\frac{1}{2}F_0^{(1,3)}+\frac{3}{2}F_1^{(0,1)}F_0^{(1,2)}+\frac{1}{2}F_1^{(1,0)}F_0^{(0,3)})\kappa_{22}\kappa_{12}\\
&+\frac{5}{24}(F_0^{(3,0)})^2\kappa_{11}^3++\frac{5}{24}(F_0^{(0,3)})^2\kappa_{22}^3+(\frac{1}{6}F_0^{(3,0)}F_0^{(0,3)}+\frac{3}{2}F_0^{(2,1)}F_0^{(1,2)})\kappa_{12}^3\\
&+(\frac{3}{8}(F_0^{(2,1)})^2+\frac{1}{4}F_0^{(3,0)}F_0^{(1,2)})\kappa_{11}^2\kappa_{22}\\
&+(\frac{3}{8}(F_0^{(1,2)})^2+\frac{1}{4}F_0^{(0,3)}F_0^{(2,1)})\kappa_{22}^2\kappa_{11}\\
&+(\frac{3}{2}(F_0^{(2,1)})^2+F_0^{(3,0)}F_0^{(1,2)})\kappa_{11}\kappa_{12}^2+(\frac{3}{2}(F_0^{(1,2)})^2+F_0^{(0,3)}F_0^{(2,1)})\kappa_{22}\kappa_{12}^2\\
&+\frac{5}{4}F_0^{(3,0)}F_0^{(2,1)}\kappa_{11}^2\kappa_{12}+\frac{5}{4}F_0^{(0,3)}F_0^{(1,2)}\kappa_{22}^2\kappa_{12}\\
&+(\frac{1}{4}F_0^{(3,0)}F_0^{(0,3)}+\frac{9}{4}F_0^{(2,1)}F_0^{(1,2)})\kappa_{11}\kappa_{22}\kappa_{12}.
\end{split}
\ee
\end{Example}

\subsection{Representation by formal Gaussian integrals}
In fact, the free energy $\wF_g$ defined in last subsection
has a representation using formal Gaussian integrals.
Consider the following partition function
\be\label{gaussian}
\begin{split}
Z(t,\kappa)=\int \exp\biggl\{&F(\eta)-\lambda^{-2}
\biggl[\frac{\partial F_0(t)}{\partial t}(\eta-t)
+\frac{1}{2}(\eta-t)^T\frac{\partial^2 F_0(t)}{\partial t^2}(\eta-t)\\
&+\frac{1}{2}(\eta-t)^T\kappa^{-1}(\eta-t)\biggr]\biggr\}d\eta,
\end{split}
\ee
where the integral is over $\eta\in\bR^N$.
Now do the Taylor expansion of $F(\eta)$, we get
\be\label{partition-fn}
\begin{split}
Z(t,\kappa)=&Z(t)\int \exp\biggl[\sum_{2g-2+\sum l^{(i)}>0}
\frac{\lambda^{2g-2}}{\prod_{i=1}^{N}l^{(i)}!}F_g^{(l^{(1)},\cdots,l^{(N)})}(t)
\prod_{j=1}^{N}(\eta_i-t_i)^{l^{(i)}}
\\
&-\frac{\lambda^{-2}}{2}(\eta-t)^T\kappa^{-1}(\eta-t)\biggr]d\eta\\
=&Z(t)\sum_{k\geq 0}\sum_{2g_j-2+\sum_{i} l_j^{(i)}>0}\frac{1}{k!}
\frac{\lambda^{2\sum_k g_k -2k}}{\prod_{1\leq j\leq k}\prod_i l_j^{(i)}!}
\prod_{j=1}^{k}F_{g_j}^{(l_j^{(1)},\cdots,l_j^{(N)})}\times\\
&\int \prod_{i=1}^{N}(\eta_i-t_i)^{\sum_{j=1}^k l_j^{(i)}}\cdot
\exp\big[-\frac{\lambda^{-2}}{2}(\eta-t)^T\kappa^{-1}(\eta-t)\big]d\eta,
\end{split}
\ee
where $Z(t)=\exp(F(t))$ is the holomorphic partition function.
Then evaluate the Gaussian integral on the right-hand side of the last equality,
we have
\be\label{gauss}
\begin{split}
&\int \prod_{i=1}^{N}(\eta_i-t_i)^{\sum_{j=1}^k l_j^{(i)}}\cdot
\exp[-\frac{\lambda^{-2}}{2}(\eta-t)^T\kappa^{-1}(\eta-t)]d\eta\\
=&(\frac{(2\pi)^n}{\det(\lambda^{-2}\kappa^{-1})})^\frac{1}{2}\cdot\frac{1}{2^l\cdot l!}
\sum_{\sigma\in S_{2l}}\prod_{j=1}^{l}(\lambda^2\kappa_{a_{\sigma(2j-1)}a_{\sigma(2j)}}),
\end{split}
\ee
for $\sum_{i,j}l_j^{(i)}=2l$ even, and equals zero if $\sum_{i,j}l_j^{(i)}$ is odd,
where the sequence $(a_1,\cdots,a_{2l})$ consists of indices in $\{1,\cdots,N\}$
with the number of $i$ equals $\sum_j l_j^{(i)}$.

Let the normalized partition function be
\be
\widetilde{Z}(t,\kappa)=Z(t,\kappa)/(2\pi\lambda^2)^{\frac{N}{2}},
\ee
then this relates to our free energy (\ref{free-expression}) in the following way.

\begin{Theorem}
We have
\be\label{non-holoF}
\widetilde{Z}(t,\kappa)=\exp(\sum_{g=0}^{\infty}\lambda^{2g-2}\wF_g(t)),
\ee
where $\wF_0$ and $\wF_1$ are defined by
\be\label{f0f1}
\wF_0=F_0,\quad \wF_1=F_1+\frac{1}{2}\log(\det(\kappa)).
\ee
\end{Theorem}

\begin{proof}
From (\ref{partition-fn}) and (\ref{gauss}), the coefficient of $\lambda^{n}$ in $\widetilde{Z}(t,\kappa)/(Z(t)\sqrt{\det(\kappa)})$ is
\ben
\sum_{2\sum g_j-2k+2l=n}\frac{1}{k!}\cdot\frac{1}{2^l\cdot l!}\prod_{j=1}^{k}\frac{F_{g_j}^{(l_j^{(1)},\cdots,l_j^{(N)})}}{l_j^{(1)}!\cdots l_j^{(N)}!}\cdot\sum_{\sigma\in S_{2l}}\prod_{j=1}^{l}\kappa_{a_{\sigma(2j-1)}a_{\sigma(2j)}}.
\een
This corresponds to the summation of $\omega_{\Gamma}$ over all possible ways to glue some stable vertices of types $\{F_{g_j}^{(l_j^{(1)},\cdots,l_j^{(N)})}\}_{1\leq j\leq l}$ together to obtain a stable graph $\Gamma$ (not necessarily connected), with coefficients $\frac{1}{|\Aut(\Gamma)|}$. Then since
\ben
l=|E(\Gamma)|,\quad k=|V(\Gamma)|,
\een
we have
\ben
g(\Gamma)=1+l-k+\sum_{j=1}^{k} g_j.
\een
Thus
\ben
\widetilde{Z}(t,\kappa)=\sum_{\Gamma\in\cG_{g,0}}\frac{\lambda^{2g-2}}{|\Aut(\Gamma)|}\omega_\Gamma
\een
and
\ben
\log(\widetilde{Z})=\sum_{\Gamma\in\cG_{g,0}^c}\frac{\lambda^{2g-2}}{|\Aut(\Gamma)|}\omega_\Gamma
=\sum_{g=0}^{\infty}\lambda^{2g-2}\wF_g.
\een
\end{proof}

\begin{Example}
For the case of $N=1$,
the partition function is
\be\label{N=1-partition fn}
\begin{split}
\tilde{Z}(t,\kappa)=& \exp(\sum_{g=0}^\infty \lambda^{2g-2}\wF_g(t))\\
=& Z(t) \sum_{k\geq 0}\sum_{2g_i-2+l_i>0}\sum_l\frac{1}{k!}\frac{\lambda^{2(g_1+\cdots+g_k)-2k+2l}}{l_1!\cdots l_k!}\times\\
&(2l-1)!!\cdot F_{g_1}^{(l_1)}\cdots F_{g_k}^{(l_k)}\cdot\kappa^l\cdot\delta_{l_1+\cdots+l_k-2l,0}.
\end{split}
\ee
\end{Example}

As mentioned in the Introduction,
the formal Gaussian integral (\ref{gaussian}) appeared in
earlier work on holomorphic anomaly equations
\cite{ey, emo, gkmw},
where the propagator $\kappa$ was chosen to be the Zamolodchikov metric
which is non-holomorphic.
Later in this paper we will see that the propagator $\kappa$ admits different choices,
which give us different theories in mathematical physics.
In particular,
by taking suitable holomorphic propagators
we will see that partition functions for some well-known quantum spectral
curves can be constructed in our formalism.

\subsection{Realization of the operators on stable graphs}

Now we discuss the realization of the operators $K$, $\pd$ and $\gamma$ defined in
Section \ref{sec:Operators-1}.

The operator $K$ is realized by the differential operator $\pd_\kappa$.
In fact, by \eqref{omega},
\ben
\pd_\kappa \omega_\Gamma &=&
\pd_\kappa \biggl[(\prod_{v\in V(\Gamma)}F_{g(v)}^{(\val(v))}(t))\kappa^{|E(\Gamma)|} \biggr] \\
& = & |E(\Gamma)| \cdot (\prod_{v\in V(\Gamma)}F_{g(v)}^{(\val(v))}(t))\kappa^{|E(\Gamma)|-1}\\
& = & \sum_{\Gamma'} \omega_{\Gamma'} = \omega_{K\Gamma},
\een
where the summation on the right-hand side of the second equality is
over all graphs $\Gamma'$ obtained from $\Gamma$ by cutting an edge.

The operator $\gamma$ is realized by the operator of multiplication by
$|E^{\ext}(\Gamma)| \cdot \kappa F_0'''$ on $\omega_\Gamma$.
In fact,
the effect of $\gamma$ is to change each external edge incident at $v$ of a stable graph
$\Gamma$ to an internal edge joining the vertex $v$ to an additional
trivalent vertex of genus $0$.
Let $\Gamma'$ be one of stable graphs obtained from $\Gamma$ in this way,
then by the Feynman rule,
\be
\omega_{\Gamma'}=  \kappa F_0''' \cdot \omega_\Gamma .
\ee
Therefore,
\be
\omega_{\gamma \Gamma}= \sum_{\Gamma'} \omega_{\Gamma'}
= |E^{ext}(\Gamma)| \cdot \kappa F_0''' \cdot \omega_\Gamma.
\ee

The operator $\pd$ is realized by the differential operator $\pd_t$.
In fact, $\omega_\Gamma$ is a product of contributions from the vertices
and the contributions from the internal edges.
Hence $\pd_t \omega_\Gamma$ is a summation of the contributions
of the actions on each vertices and the contributions on each internal edges.
In fact by \eqref{omega},
\ben
\pd_t\omega_{\Gamma}
& = & \pd_t \biggl[(\prod_{v\in V(\Gamma)}F_{g(v)}^{(\val(v))}(t))\kappa^{|E(\Gamma)|} \biggr]\\
& = & \sum_{v\in V(\Gamma)} F_{g(v)}^{(\val(v)+1)}(t)
\cdot \prod_{v'\in V(\Gamma)-\{v\}}F_{g(v')}^{(\val(v'))}(t)) \cdot \kappa^{|E(\Gamma)|} \\
&+& \prod_{v\in V(\Gamma)}F_{g(v)}^{(\val(v))}(t))\cdot |E(\Gamma)| \kappa^{|E(\Gamma)|-1}
\cdot \pd_t\kappa.
\een
On the other hand,
\be
\pd \Gamma = \sum_{v\in V(\Gamma)} \Gamma^v + \sum_{e\in E(\Gamma)}\Gamma^e,
\ee
where $\Gamma^v$ is obtained from $\Gamma$ by attaching an additional half-edge
at a vertex $v \in V(\Gamma)$,
and $\Gamma^e$ is obtained from $\Gamma$ by changing an internal edge $e\in E(\Gamma)$
to two internal edges incident at an addition trivalent vertex $w$ of genus zero.
Therefore,
by the Feynman rule we get:
\ben
\omega_{\pd \Gamma} &=& \sum_{v\in V(\Gamma)} \omega_{\Gamma^v}
+ \sum_{e\in E(\Gamma)} \omega_{\Gamma^e}\\
& = & \sum_{v\in V(\Gamma)} F_{g(v)}^{(\val(v)+1)}(t)
\cdot \prod_{v'\in V(\Gamma)-\{v\}}F_{g(v')}^{(\val(v'))}(t)) \cdot \kappa^{|E(\Gamma)|} \\
&+& \sum_{e\in E(\Gamma)}
\prod_{v\in V(\Gamma)}F_{g(v)}^{(\val(v))}(t)) \cdot
\kappa^{|E(\Gamma)|+1} \cdot F_0'''(t).
\een
Therefore,
in order to have $\omega_{\pd \Gamma}= \pd_t \omega_\Gamma$,
we require the propagator $\kappa$ to be a function of $t$ which satisfies
\be\label{constraint}
\frac{\partial\kappa}{\partial t}=\kappa^2\cdot F_0'''(t),
\ee
here $\kappa$ needs not to be holomorphic. A general solution to this equation is
\ben
\kappa(t)=\frac{1}{C-F_0''(t)},
\een
where $C$ may be either a constant or an anti-holomorphic function of $t$.

\subsection{Realization of the recursion relations}

By applying Feynman rules to the
quadratic recursion relations in Theorem \ref{thm2}
and using the realizations of the operators on graphs discussion in the preceding subsection,
we obtain the recursion relations for $\wF_g$ immediately.

\begin{Theorem}\label{1-recursion}
When the propagator $\kappa$ satisfies the condition
\eqref{constraint},
the free energy $\wF_g$ for $g \geq 2$ satisfies the equations
\be
\partial_\kappa\wF_g=\frac{1}{2}(D_t \partial_t\wF_{g-1}
+\sum_{r=1}^{g-1}\partial_t\wF_{r}\partial_t\wF_{g-r}),
\ee
where $D_t=\partial_t+\kappa F_0'''$ is the covariant derivative.
\end{Theorem}

\begin{Example}
For the expression (\ref{f2}) of $\wF_2$, we have
\ben
\partial_\kappa\wF_2&=&\frac{1}{2}F_1''+\frac{1}{2}(F_1')^2
+\kappa(\frac{1}{4}F_0^{(4)}+F_1'F_0''')+\frac{5}{8}\kappa^2(F_0''')^2,\\
\partial_t\wF_2&=&F_2'+\kappa[\frac{1}{2}F_1'''+F_1'F_1'']
+\kappa^2(\frac{1}{8}F_0^{(5)}+\frac{1}{2}F_1'F_0^{(4)}
+\frac{1}{2}(F_1')^2F_0'''+F_1''F_0''')\\
&+&\kappa^3[F_1'(F_0''')^2+\frac{2}{3}F_0'''F_0^{(4)}]+\frac{5}{8}\kappa^4(F_0''')^3.
\een
Here in the second equality we have used \eqref{constraint}.
\end{Example}

\begin{Example}
Direct computation gives
\ben
&& D_t\partial_t\wF_1=F_1''+\kappa(\frac{1}{2}F_0^{(4)}+F_1'F_0''')+\kappa^2 (F_0''')^2,
\een
one may check that
\ben
&& \partial_\kappa\wF_2=\frac{1}{2}(D_t \partial_t\wF_{1}+(\partial_t\wF_{1})^2).
\een
\end{Example}

\subsection{The general case}

Now we return to the general case for an arbitrary dimension $N$.
Similar to $N=1$ case,
the operators $K_{ij}$ and $\partial_{i}$
defined in Section \ref{sec3}
are realized by $\partial_{\kappa_{ij}}$ and    $\partial_{t_i}$  respectively.
To realize the operator $\pd_i$ as the differential operator $\pd_{t_i}$,
we need to impose some conditions on the propagators $\kappa_{ij}$
analogous to \eqref{constraint}:
\be
\frac{\partial\kappa_{jk}}{\partial t_i}
=\sum_{l,m}\kappa_{jl}\kappa_{km}\frac{\partial^3 F_0}{\partial t_i\partial t_l\partial t_m},
\ee
or in matrix form:
\be\label{N-constraint}
\frac{\partial\kappa}{\partial t_k}=\kappa\frac{\partial H(F_0)}{\partial t_k}\kappa
\ee
where $H(F_0)=(\frac{\partial^2 F_0}{\partial t_i \partial t_j})$ is the Hessian of $F_0$. A general solution is given by
\ben
\kappa(t)=(C-H(F_0))^{-1},
\een
where $C$ is a symmetric anti-holomorphic or constant matrix
such that $(C- H(F_0))$ is invertible.

The realization of the operator $D_{i}=\partial_{i}+\gamma_i$ is more subtle.
For our purpose,
it suffices to define the operator $D_{t_i}\partial_{t_j}$ as a realization of
the operator $D_i\pd_j$ as follows.

\begin{Definition}
For $\Gamma\in\cG_{g,0}^c$, we define
\be
\begin{split}
D_{t_i}\partial_{t_j}\omega_\Gamma=&\omega_{D_i\partial_j\Gamma}\\
=&\frac{\partial^2}{\partial t_i \partial t_j}\omega_\Gamma
+\sum_{l,m=1}^N \frac{\partial}{\partial t_l}\omega_\Gamma\cdot\kappa_{lm}(t)
\frac{\partial^3 F_0}{\partial t_m \partial t_i \partial t_j},
\end{split}
\ee
where
\ben
\omega_{\Gamma}=\prod_{v\in V(\Gamma)}F_{g(v)}^{(\val_1(v),\cdots,\val_N(v))}(t)
\cdot \prod_{e\in E(\Gamma)}\kappa_{l_1(e)l_2(e)}(t, \bar{t}).
\een
In particular, we see that $D_{t_i}\partial_{t_j}\omega_\Gamma=D_{t_j}\partial_{t_i}\omega_\Gamma$.
\end{Definition}

Then by applying the Feynman rules, Theorem \ref{N-graph} gives us

\begin{Theorem}\label{N-recursion}
When the condition \eqref{N-constraint} is satisfied,
for $g\geq 2$, we have the equations
\be
\partial_{\kappa_{ij}}\wF_g=D_{t_i}\partial_{t_j}\wF_{g-1}+\sum_{r=1}^{g-1}\partial_{t_i}\wF_r\partial_{t_j}\wF_{g-r}
\ee
for $i\not= j$, and
\be
\partial_{\kappa_{ii}}\wF_g=\frac{1}{2}(D_{t_i}\partial_{t_i}\wF_{g-1}+\sum_{r=1}^{g-1}\partial_{t_i}\wF_r\partial_{t_i}\wF_{g-r})
\ee
for the free energy $\wF_g$ defined by (\ref{free-expression}).
\end{Theorem}

\section{The Holomorphic Anomaly Equation} \label{sec5}

In the rest of this paper we will present some examples
of our formalism developed above.
In this section we will focus on the case of holomorphic anomaly equations.

As we have mentioned in the Introduction,
the holomorphic anomaly equation was introduced by Bershadsky {\em et al}
in physics literatures \cite{bcov1,bcov2}
in order to compute the Gromov-Witten invariants of the quintic Calabi-Yau threefold.
They developed a method to solve the non-holomorphic free energies $\cF_g(t,\bar{t})$
recursively using the holomorphic anomaly equations
\be
\bar\partial_{\bar{i}}\cF_g=\frac{1}{2}\bar{C}_{\bar{i}\bar{j}\bar{k}}e^{2K}G^{j\bar{j}}G^{k\bar{k}}(D_j D_k\cF_{g-1}+\sum_{r=1}^{g-1}D_j\cF_r D_k\cF_{g-r}), \quad g\geq 2,
\ee
and by formally `taking a limit'
\ben
F_g(t)=\lim_{\bar{t}\to\infty}\cF_g(t,\bar{t})
\een
they got the generating series of Gromov-Witten invariants $F(t)$ of
the threefold on the small phase space.
The non-holomorphic free energies $\cF_g(t,\bar{t})$
are supposed to have modularity while the holomorphic free energies $F_g(t)$ does not,
and this has been exploited to use the theory of modular forms.
Since the ring of modular forms are often polynomial algebras with finitely many  generators,
this suggests some polynomiality properties of the non-holomorphic
free energies \cite{yy}.
This has led to many recent progress on the computation of
Gromov-Witten invariants of quintic Calabi-Yau 3-folds.

Unfortunately the mathematical or geometric definition of
the nonholomorphic free energy
is lacking in the mathematical literature.
According to \cite{bcov2},
the nonholomorphic free energy should be closely related to some geometric structures
on the relevant moduli spaces.
In the genus zero case,
it is known that the free energy is related to the special K\"ahler geometry
of the moduli space \cite{cdgp}.
In the genus one case,
the free energy is related to the $tt^*$-geometry of the moduli space \cite{cv,dub}
and the theory of analytic torsions \cite{bcov1,cfiv}.
In the higher genus case,
the nonholomorphic free energy is supposed to be
a suitable non-holomorphic
section of a holomorphic line bundle on the moduli space \cite{bcov2}.
Such an approach depends heavily on the special properties of Calabi-Yau 3-folds.

It is interesting to further clarify the geometric meaning of non-holomorphic
free energy,
and furthermore, to find analogue of holomorphic
anomaly equation for Gromov-Witten theory of general symplectic manifolds
or other related theories.
The formalism we develop in this paper shed some lights on this problem.
In fact,
our formalism is inspired by the direct integration
approach to holomorphic anomaly equations,
and it does not need to start with a Calabi-Yau 3-fold.

Let us briefly recall the direct integration approach.
According to \cite{bcov2},
the non-holomorphic free energy
should count the contributions from degenerate Riemann surfaces,
and so it should involve the boundary strata of the Deligne-Mumford
moduli spaces of algebraic curves.
Witten \cite{wit1} interpreted the holomorphic anomaly  equation
from the point of view of geometric quantization of symplectic vector spaces
and background independence.
Inspired by this work,
Aganagic, Bouchard and Klemm \cite{abk} obtained via change of polarizations
the following general form for $\wF_g$:
\be
\wF_g(t, \bar{t}) = F_g(t) + \Gamma_g \biggl(\Delta^{IJ},
\pd_{I_1} \cdots \pd_{I_n}F_{r<g}(t)\biggr),
\ee
in  some particular cases,
$\Delta^{IJ}$ can be taken to be $-(\tau -\bar{\tau})^{IJ}$,
where $\tau = (\tau_{ij}) = \big(\frac{\pd^2F_0}{\pd t_i\pd t_j})$.
These particular cases were interpreted for matrix models
using Eynard-Orantin topological recursion
by Eynard, Mari\~no and Orantin \cite{emo}.
Furthermore,
they also represented the partition function as a formal Gaussian integral \cite[(4.27)]{emo}
and presented the Feynman graphs and Feynman rules for the terms that contribute to $\Gamma_g$.
These results have been generalized to other models
by Grimm, Klemm, Mari\~no and Weiss \cite{gkmw}.
These authors reformulated the holomorphic anomaly equation as a
quadratic recursion relation for the derivative of $\wF_g$ with respective to the propagators
$\Delta^{IJ}$ (cf. \cite[(7.50)]{gkmw}:
\be
\frac{\pd \wF_g}{\pd \Delta^{IJ}} =
\frac{1}{2}
D_I\pd_J\wF_{g-1} + \frac{1}{2}
\sum^{g-1}_{r=1}
\pd_I\wF_r\pd_J\wF_{g-r}.
\ee
From the consideration of modularity
of $\wF_g$,
the propagators $\Delta^{IJ}$ has the form
\be
\Delta^{IJ} = \frac{1}{\sqrt{-1}} ((\tau-\bar{\tau})^{-1})^{IJ} +\cE^{IJ},
\ee
where $\cE^{IJ}$ is a holomorphic function.
They also derived the formal Gaussian integral representation of the partition function
and the Feynman expansions for $\wF_g$.

It is clear that our formalism is inspired by \cite{abk, emo, gkmw}.
The results in these works
when the propagators are
\be
\kappa = -(\tau - \bar{\tau})^{-1}
\ee
 are special cases of our formalism.
Indeed,
the condition \eqref{N-constraint} clearly holds for such propagators,
and so our formalism can be applied to the  free energy $F_g$
of Gromov-Witten theory of a compact Calabi-Yau 3-fold to get a sequence $\wF_g$
that satisfies the holomorphic anomaly equations.
It is remarkable that the application of our formalism
does not need the Calabi-Yau condition,
and so it applies to the Gromov-Witten theory of {\em any}
compact symplectic manifold,
any Gromov-Witten type theory as the FJRW theory.

\begin{Example}
For FJRW theory of type $A_1$,
$F_0(t_0) = \frac{t_0^3}{3!}$,
and so we have
\be
\tau = \frac{\pd^2F_0}{\pd t_0^2} = t_0,
\ee
and so
\be
\kappa = \frac{1}{\bar{t}_0 - t_0}.
\ee
The free energy functions $F_g(t_0)$ are generating functions of intersection
numbers on $\Mbar_{g,n}$,
restricted to the small phase space $t_i = 0$ for $i > 0$.
Therefore, for $g \geq 1$,
\be
F_g(t_0)= 0.
\ee
Then our construction yields a sequence $\wF_g(t, \bar{t})$:

\end{Example}

\begin{Example}
For the Gromov-Witten theory of $\bP^1$,
the free energy function restricted to the small phase space is given by
\begin{align*}
F_0(t_0, t_1) &= \frac{t_0^2t_1}{2} +q e^{t_1},\\
F_1(t_0,t_1) &= - \frac{t_1}{24}, \\
F_g(t_0, t_1) &= 0, \;\; g \geq 2,
\end{align*}
where $q$ is the degree tracking parameter.
The period matrix $\tau$ can be directly computed as follows:
\be
\tau = \begin{pmatrix}
t_1 & t_0 \\
t_0 & q e^{t_1}
\end{pmatrix}
\ee
and so the propagator is
\be
\kappa = - \begin{pmatrix}
t_1 - \bar{t}_1 & t_0 - \bar{t}_0 \\
t_0 - \bar{t}_0 & q e^{t_1}- qe^{\bar{t}_1}
\end{pmatrix}^{-1}
\ee
\end{Example}

It is an interesting problem to extend our construction to the big phase space
that includes all the gravitational descendants of primary observables.
It is well-known that Gromov-Witten theory and its various generalizations
in the mathematical literature describes what physicists
refer to as the topological string theory.
The mathematical meaning or definition of the partition function
of what physicists refer to as the string theory needs to be clarified at present.
For example, for Gromov-Witten theory,
the free energy $F_g$ depends on infinitely many coupling constants
$\{t_{i, n}\}_{1\leq i \leq N, n \geq 0}$,
we expect the string theory behind it to have free energy $\wF_g$
that also depends on $\{\bar{t}_{i, n}\}_{1\leq i \leq N, n \geq 0}$.
We hope our construction shed some lights on this problem.

When the propagators take the form
\be
\kappa = -(\tau - \bar{\tau})^{-1}+\cE,
\ee
where $\cE$ is a symmetric matrix whose entries are holomorphic functions,
the condition \eqref{N-constraint} do not hold in general.
This means that we need to modify the diagrammatics we discuss in this paper.
We will do this in a subsequent work \cite{wz3}.

\section{Relationship to Eynard-Orantin Topological Recursion and Quantum Spectral Curves}
\label{sec6}

In this section we write the wave function of the quantum spectral curve
of Gukov-Su{\l}kowski \cite{gs} as a summation over stable graphs,
and interpret it as a particular realization of the abstract QFT.

\subsection{Graph sum for Eynard-Orantin topological recursion}

The theory of Eynard-Orantin topological recursion was developed by
Eynard and Orantin in \cite{eo},
inspired by the theory of matrix models.
The input data of this theory are a Riemann surface $\mathcal{C}$ (spectral curve),
two holomorphic functions $x$ and $y$ on this curve,
and a symmetric bilinear meromorphic 2-form of the 2nd kind $B(p,q)$ (Bergman kernel)
on $\mathcal{C}\times\mathcal{C}$ which has a double pole along the diagonal and no other pole.
Then a family of meromorphic n-differentials $W^{(g)}_n(z_1,\cdots,z_n)$
on $\mathcal{C}^n$ are defined recursively from these input data.
The Eynard-Orantin invariants $W_n^{(g)}$ can be represented
as a sum over trivalent graphs (with both oriented edges and non-oriented edges)
of genus $g$ with $n$ leaves marked by $z_1,\cdots, z_n$ respectively
(see \cite[\S 4.5]{eo}).

In the case of one branch point,
Eynard \cite{ey1} showed that the invariants $W_n^{(g)}$
can be expressed in terms of intersection numbers on the moduli space of stable curves $\Mbar_{g,n}$.
In the case of $N$ branch points,
$W_n^{(g)}$ can be expressed in terms of intersection numbers on the moduli space
of `colored' stable curves $\Mbar_{g,n}^N$, see Eynard \cite{ey2}.
The vertices of trivalent graphs in this case are indexed by $\{1,\cdots, N\}$.
By cutting these trivalent graphs into clusters of constant indices,
summations over these trivalent graphs can be reformulated as summations over
labelled stable graphs of type $(g,n)$,
see \cite[Appendix A.3]{ey2} and Kostov-Orantin \cite{ko}.

Dunin-Barkowski, Orantin, Shadrin, and Spitz reformulated this diagrammatic representation
for the correlation functions $W_n^{(g)}$ of the local topological recursion
as a summation over decorated stable graphs in \cite[\S 3]{doss}.
Now let us recall their results  in the rest of this subsection.

For a fixed positive integer $N$,
define the analytic functions
\ben
x^i(z):=z^2+a_i,
\qquad
y^i(z):=\sum_{k=0}^\infty h_k^i z^k,
\een
and a bi-differential
\ben
B^{ij}(z,z')=\delta_{ij}\cdot\frac{dz\otimes dz'}{(z-z')^2}
+\sum_{k,l=0}^\infty B_{k,l}^{i,j}z^kz'^l dz\otimes dz'
\een
in a neighborhood of $0\in \bC$ for $i,j\in \{1,2,\cdots,N\}$,
where $\{h_k^i\}_{k\in \bN}$ are $N$ families of complex numbers called times,
and $\{B^{i,j}_{k,l}\}_{(k,l)\in \bN^2}$ are
$N\times N$ families of complex numbers called jumps.
The Eynard-Orantin invariants are defined by
\ben
\omega_{0,1}^i(z):=0,\qquad
\omega_{0,2}^{i,j}(z,z'):=B^{i,j}(z,z'),
\een
and the recursion relation
\begin{equation*}
\begin{split}
&\omega_{g,n+1}^{i_0,i_1,\cdots,i_n}(z_0,z_1,\cdots,z_n):=
\sum_{j=1}^N Res_{z\to 0}
\frac{\int_{-z}^z B^{i_0,j}(z_0,\cdot)}{2(y^j(z)-y^j(-z))dx^j(z)}\cdot\\
&\biggl(\omega_{g-1,n+2}^{j,j,i_1,\cdots,i_n}(z,-z,z_1,\cdots,z_n)+
\sum_{\substack{0\leq h\leq g\\A\sqcup B=\{1,\cdots,n\}}}
\omega_{h,|A|+1}^{j,i_A}(z,z_A)\omega_{g-h,|B|+1}^{j,i_B}(-z,z_B)
\biggr)
\end{split}
\end{equation*}
for every $2g-2+n>0$,
where $i_A:=\{i_{a_1},\cdots, i_{a_k}\}$ and $z_A:=\{z_{a_1},\cdots, z_{a_k}\}$
for a subset $A=\{a_1,\cdots,a_k\}$ of $\{1,2,\cdots,n\}$.

Following \cite{doss},
let $\Gamma_{g,n}$ be the set of connected decorated stable graphs of genus $g$,
with $n$ ordinary leaves marked by arguments $z_1,\cdots, z_n$.
Let $V(\Gamma)$, $E(\Gamma)$, $H(\Gamma)$ and $L(\Gamma)$
be the set of vertices, edges, half-edges, and leaves of $\Gamma$ respectively.
The set of leaves $L(\Gamma)=L^*(\Gamma)\sqcup L^\bullet(\Gamma)$,
where $L^*(\Gamma)$ is the set of ordinary leaves,
and $L^\bullet(\Gamma)$ is the set of dilaton leaves.
A half-edge is defined to be either a leaf,
or an internal edge together with a choice of one of the two vertices it is attached to.
Let $i:V(\Gamma)\to \{1,2,\cdots, N\}$ be the markings of vertices,
and this induces the markings on leaves $i:L(\Gamma)\to \{1,2,\cdots, N\}$
by assigning $i(l):=i(v)$ if $l$ is a leaf attached to the vertex $v$.
The marking on a dilaton leaf is required to be greater than one.
For an internal edge $e$,
denote by $v_1(e)$, $v_2(e)$ the two vertices $e$ is attached to,
and by $h_1(e)$, $h_2(e)$ the two corresponding half-edges.
For a vertex $v\in V(\Gamma)$,
denote by $H(v)$ the set of half-edges incident at $v$.
Let $k:H(\Gamma)\to\bZ_{\geq 0}$ be the heights of the half-edges.
Define
\ben
\omega_{g,n}(\vec{z}):=\sum_{\vec{i}}\omega_{g,n}^{\vec{i}}(\vec{z}).
\een

\begin{Theorem}(\cite{ey2}, \cite[Theorem 3.7]{doss})
For $2g-2+n>0$,
\begin{equation}\label{EO-graphsum}
\begin{split}
\omega_{g,n}(\vec{z})=&\sum_{\Gamma\in\Gamma_{g,n}}\frac{1}{|\Aut(\Gamma)|}
\prod_{v\in V(\Gamma)}\biggl(-2h_1^{i(v)}\biggr)^{\chi_{g(v),\val(v)}}
\langle\prod_{h\in H(v)}\tau_{k(h)}\rangle_{g(v),\val(v)}\\
&\cdot\prod_{e\in E(\Gamma)}\check{B}_{k(h_1(e)),k(h_2(e))}^{i(v_1(e)),i(v_2(e))}
\prod_{l\in L^*(\Gamma)}\sum_{j=1}^Nd\xi_{k(l)}^{i(l)}(z_l,j)
\prod_{\lambda\in L^{\bullet}(\Gamma)}\check{h}_{k(\lambda)}^{i(\lambda)},
\end{split}
\end{equation}
where
\ben
&&\chi_{g(v),\val(v)}:=2-2g(v)-\val(v);\\
&&\check{h}_k^i:=2(2k-1)!! h_{2k-1}^i;\\
&&d\xi_d^i(z_a,j):=Res_{z\to 0}\frac{(2d+1)!!dz}{z^{2d+2}}\int^z B^{i,j}(z,z_a);\\
&&\check{B}_{d_1,d_2}^{i,j}:=(2d_1-1)!!(2d_2-1)!!B_{2d_1,2d_2}^{i,j};\\
&&\langle\prod_{i=1}^n \tau_{k_i}\rangle_{g,n}:=
\int_{\Mbar_{g,n}} \psi_1^{k_1}\psi_2^{k_2}\cdots\psi_n^{k_n}.
\een
\end{Theorem}

\subsection{Graph sum for quantum spectral curves}

In \cite{gs}, Gukov and Su{\l}kowski proposed a method to
construct the quantum spectral curves using
Eynard-Orantin topological recursion.
In this subsection we first recall their construction,
and then reformulate it as a summation over
decorated stable graphs without markings $z_1,\cdots, z_n$ on ordinary leaves.

Let $\bC\times\bC$ be the complex plane with coordinates $(u,v)$,
equipped with the symplectic form $\omega=\frac{\sqrt{-1}}{\hbar}du\wedge dv$.
Let $A(u,v)$ be a polynomial in $u$ and $v$,
then the algebraic curve
\ben
\mathcal{C}:\quad A(u,v)=0
\een
is a Lagrangian submanifold in $(\bC\times\bC,\omega)$.
Quantization turns the coordinates $u$ and $v$ into operators $\hat u$ and $\hat v$
which satisfy the commutation relation $[\hat v,\hat u]=\hbar$,
and algebras of functions in $u,v$ into a noncommutative algebra of operators.
The quantization of the polynomial $A(u,v)$ is an operator
\be
\widehat A
=\widehat A_0 +\hbar \widehat A_1 +\hbar^2 \widehat A_2 +\cdots,
\ee
where $\widehat A_0$ is obtained from $A$ by replacing $u,v$ by $\hat u,\hat v$ respectively.

Inspired by the matrix model origin of the Eynard-Orantin topological recursion,
Gukov and Su{\l}kowski defined the following Baker-Akhiezer function
(see \cite[\S 2.1]{gs}):
\be
Z(z):=\exp\biggl(\sum_{n=0}^{\infty}\hbar^{n-1}S_n(z)\biggr),
\ee
where
\be  \label{eqn:GS}
\begin{split}
& S_0(z):=\int^z v(z)du(z),\\
& S_1(z):=-\half \log \frac{du}{dz},\\
& S_n(z):=\sum_{2g-1+k=n}\frac{1}{k!}\int^z\cdots\int^z
\omega_{g,k}(z_1,\cdots,z_k),
\quad n\geq 2,
\end{split}
\ee
and $\omega_{g,k}(z_1,\cdots,z_k)$ are the Eynard-Orantin invariants defined by
the spectral curve $u=u(z),v=v(z)$
together with a choice of Bergman kernel $B(p,q)$.
Then these authors conjectured that the quantum spectral curve $\widehat A$
can be obtained by solving the Schr\"odinger equation
\be\label{eq-Schrodinger}
\widehat A Z(z) =0.
\ee

Now let us combine the above construction of $Z(z)$ due to Gukov-Su{\l}kowski
and the graph sum \eqref{EO-graphsum} of Dunin-Barkowski {\em et al} \cite{doss}
to give a diagrammatic representation for the functions $S_n(z)$ ($n\geq 2$).

It follows from \eqref{EO-graphsum} that
\begin{equation}\label{QSC-graphsum}
\begin{split}
&\int^z\cdots\int^z \omega_{g,k}(z_1,\cdots,z_k)\\
=&\sum_{\Gamma\in\Gamma_{g,k}}\frac{1}{|\Aut(\Gamma)|}
\prod_{v\in V(\Gamma)}\biggl(-2h_1^{i(v)}\biggr)^{\chi_{g(v),\val(v)}}
\langle\prod_{h\in H(v)}\tau_{k(h)}\rangle_{g(v),\val(v)}\\
&\cdot\prod_{e\in E(\Gamma)}\check{B}_{k(h_1(e)),k(h_2(e))}^{i(v_1(e)),i(v_2(e))}
\prod_{l\in L^*(\Gamma)}\int^z\biggl(\sum_{j=1}^Nd\xi_{k(l)}^{i(l)}(z_l,j)\biggr)
\prod_{\lambda\in L^{\bullet}(\Gamma)}\check{h}_{k(\lambda)}^{i(\lambda)}.
\end{split}
\end{equation}
Therefore $S_n$ can be represented as a summation over connected decorated stable graphs
of type $(g,k)$ with $2g-1+k=n$.
Notice here the weight of every ordinary leaf $l\in L^*(\Gamma)$ is given by
\be\label{weight-leaf}
\int^z\biggl(\sum_{j=1}^Nd\xi_{k(l)}^{i(l)}(z_l,j)\biggr),
\ee
which is independent of $z_l$.
Thus we now obtain an enumeration problem of decorated stable graphs
where we do not mark $z_1,\cdots,z_n$ on the ordinary leaves,
i.e., we do not distinguish these external edges.

Let $\widetilde\Gamma_{g,n}$ be the set of decorated stable graphs
obtained from decorated stable graphs in $\Gamma_{g,n}$ by forgetting
the markings $z_1,\cdots,z_n$ on ordinary leaves,
then we have the following:
\begin{Theorem}
For every $2g-2+k>0$, we have:
\begin{equation}\label{QSC-graphsum2}
\begin{split}
&\frac{1}{k!}
\int^z\cdots\int^z \omega_{g,k}(z_1,\cdots,z_k)\\
=&\sum_{\widetilde\Gamma\in\widetilde\Gamma_{g,k}}\frac{1}{|\Aut(\widetilde\Gamma)|}
\prod_{v\in V(\widetilde\Gamma)}\biggl(-2h_1^{i(v)}\biggr)^{\chi_{g(v),\val(v)}}
\langle\prod_{h\in H(v)}\tau_{k(h)}\rangle_{g(v),\val(v)}\\
&\cdot\prod_{e\in E(\widetilde\Gamma)}\check{B}_{k(h_1(e)),k(h_2(e))}^{i(v_1(e)),i(v_2(e))}
\prod_{l\in L^*(\widetilde\Gamma)}\int^z\biggl(\sum_{j=1}^N d\xi_{k(l)}^{i(l)}(z,j)\biggr)
\prod_{\lambda\in L^{\bullet}(\widetilde\Gamma)}\check{h}_{k(\lambda)}^{i(\lambda)}.
\end{split}
\end{equation}
\end{Theorem}
\begin{proof}
Let $\widetilde\Gamma\in \widetilde\Gamma_{g,k}$ be a decorated stable graph of type $(g,k)$,
whose ordinary leaves are unmarked.
Denote by $\mathcal S_{\widetilde\Gamma}$ the set of decorated stable graphs $\Gamma\in\Gamma_{g,k}$,
such that $\widetilde\Gamma$ can be obtained from $\Gamma$
by forgetting the markings $z_1,\cdots, z_k$ on the ordinary leaves of $\Gamma$.

Since we already have \eqref{QSC-graphsum},
it suffices to show that
\be\label{eq-QSC-graph}
\frac{1}{k!}\cdot
\sum_{\Gamma\in\mathcal S_{\widetilde\Gamma}}\frac{1}{|\Aut(\Gamma)|}
=\frac{1}{|\Aut(\widetilde\Gamma)|}.
\ee

Now let us prove the above equality.
First observe that for every $\Gamma,\Gamma'\in \mathcal S_\Gamma$,
we have $|\Aut(\Gamma)|=|\Aut(\Gamma')|$.
Thus we only need to check
\be\label{eq2-QSC-graph}
\frac{1}{k!}\cdot
\frac{|\mathcal S_{\widetilde\Gamma}|}{|\Aut(\Gamma)|}=\frac{1}{|\Aut(\widetilde\Gamma)|}
\ee
for an arbitrary $\Gamma\in \mathcal S_{\widetilde\Gamma}$.

Let the permutation group $\mathcal S_k$ act on the set $\mathcal S_{\widetilde\Gamma}$
by permuting the markings $z_1,\cdots,z_n$ on ordinary leaves.
This action is clearly transitive, thus
\be\label{eq3-QSC-graph}
|\mathcal S_{\widetilde\Gamma}|=\frac{|\mathcal S_k|}{|H|}=\frac{k!}{|H|},
\ee
where $H$ consists of the permutations of the markings that preserves $\Gamma$.

On the other hand, the group $\Aut(\Gamma)$ acts transitively on $H$,
and the invariant subgroup is exactly isomorphic to $\Aut(\widetilde\Gamma)$,
thus
\be\label{eq4-QSC-graph}
\frac{|\Aut(\widetilde\Gamma)|}{|\Aut(\Gamma)|}=|H|.
\ee
Then \eqref{eq3-QSC-graph} and \eqref{eq4-QSC-graph} give us \eqref{eq2-QSC-graph}.
\end{proof}

As a corollary,
we obtain the following graph sum for $S_n(z)$ ($g\geq 2$):

\begin{Corollary}
For every $g\geq 2$, we have
\begin{equation}
\begin{split}
S_n(z)
=&\sum_{\substack{\widetilde\Gamma\in\sqcup\widetilde\Gamma_{g,k}\\2g-1+k=n}}
\frac{1}{|\Aut(\widetilde\Gamma)|}
\prod_{v\in V(\widetilde\Gamma)}\biggl(-2h_1^{i(v)}\biggr)^{\chi_{g(v),\val(v)}}
\langle\prod_{h\in H(v)}\tau_{k(h)}\rangle_{g(v),\val(v)}\\
&\cdot\prod_{e\in E(\widetilde\Gamma)}\check{B}_{k(h_1(e)),k(h_2(e))}^{i(v_1(e)),i(v_2(e))}
\prod_{l\in L^*(\widetilde\Gamma)}\int^z\biggl(\sum_{j=1}^N d\xi_{k(l)}^{i(l)}(z,j)\biggr)
\prod_{\lambda\in L^{\bullet}(\widetilde\Gamma)}\check{h}_{k(\lambda)}^{i(\lambda)}.
\end{split}
\end{equation}
\end{Corollary}

\subsection{Recursion relations for $S_n(z)$}

In last subsection we have represented the function $S_n(z)$ ($n\geq 2$)
as a summation over decorated stable graphs
(without markings $z_1,\cdots, z_n$ on ordinary leaves).
Therefore it is natural to understand this summation as a particular realization
of the abstract quantum field theory.

Here the stable graphs we need are those obtained from the connected stable graphs
with labels $\{1,2,\cdots, N\}$ on half-edges (in the sense of \S \ref{sec3})
by adding some dilaton leaves,
and assigning a height $k:H(\widetilde\Gamma)\to\bZ_{\geq 0}$.
From \eqref{QSC-graphsum2} we know that
the weight of such a stable graph $\Gamma$ is
\ben
\omega_{\widetilde\Gamma}&:=&\prod_{v\in V(\widetilde\Gamma)}
\biggl(-2h_1^{i(v)}\biggr)^{\chi_{g(v),\val(v)}}
\langle\prod_{h\in H(v)}\tau_{k(h)}\rangle_{g(v),\val(v)}\\
&&\cdot\prod_{e\in E(\widetilde\Gamma)}\check{B}_{k(h_1(e)),k(h_2(e))}^{i(v_1(e)),i(v_2(e))}
\prod_{l\in L^*(\widetilde\Gamma)}\int^z\biggl(\sum_{j=1}^N d\xi_{k(l)}^{i(l)}(z,j)\biggr)
\prod_{\lambda\in L^{\bullet}(\widetilde\Gamma)}\check{h}_{k(\lambda)}^{i(\lambda)},
\een
if for every vertex $v\in V(\widetilde\Gamma)$,
the half-edges incident at $v$ have the same index $i(v)$.
If not,
we assign $\omega_{\widetilde\Gamma}:=0$ (by the definition of decorated stable graphs).

Denote by $\widetilde\cG_{g,n}^c(N)$ all such stable graphs of type $(g,n)$,
and $\widetilde\cG_{g;l_1,\cdots,l_N}^c(N)\subset\widetilde\cG_{g,n}^c(N)$ the subset
consisting of all such graphs of genus $g$, with $l_j$ external edges
labelled by $j$ for every $j\in\{1,\cdots,N\}$
(we have $l_1+\cdots+\l_N=n$).
The $n$-point function of genus $g$ in this case is
\be
\widetilde F_{g;l_1,\cdots,l_N}(z):=
\frac{1}{n!}\int^z\cdots\int^z \omega_{g,n}^{\vec{i}}(z_1,\cdots,z_n)
=\sum_{\widetilde\Gamma\in\widetilde\cG_{g;l_1,\cdots,l_N}^c(N)}
\frac{1}{|\Aut(\widetilde\Gamma)|}\omega_{\widetilde\Gamma},
\ee
where $\vec{i}=(i_0,\cdots, i_{n-1})$ is a sequence such that
$j$ appears exactly $l_j$ times for every $j\in \{1,2,\cdots, N\}$.
We also define
\be
\widetilde F_{g,n}(z):=
\frac{1}{n!}\int^z\cdots\int^z \omega_{g,n}(z_1,\cdots,z_n)
=\sum_{\widetilde\Gamma\in\widetilde\cG_{g,n}^c}
\frac{1}{|\Aut(\widetilde\Gamma)|}\omega_{\widetilde\Gamma}
\ee
for every $2g-2+n$ and $n>0$.
Then clearly
\be
\widetilde F_{g,n}(z)=\sum_{(l_1,\cdots,l_N)}\widetilde F_{g;l_1,\cdots,l_N}(z),
\ee
and the functions $S_n(z)$ are given by
\be\label{eq-reln-S-F}
S_n(z)=\sum_{2g-1+k=n}\widetilde F_{g,k}(z).
\ee

Notice that the recursion relation for abstract $n$-point functions in Theorem \ref{thm-N-rec-npt}
can be naturally generalized to the case of stable graphs in this case,
therefore we obtain the following recursion relation for $\widetilde F_{g;l_1,\cdots,l_N}(z)$:

\begin{Theorem}
For $2g-2+\sum l_j>0$, we have
\begin{equation*}
\begin{split}
&\widetilde K_{ij}\widetilde F_{g;l_1,\cdots,l_N}
=\widetilde D_i\widetilde D_j\widetilde F_{g-1;l_1,\cdots,l_N}+
\sum_{\substack{g_1+g_2=g\\p_k+q_k=l_k}}
\widetilde D_i\widetilde F_{g_1;p_1,\cdots,p_N}
\widetilde D_j\widetilde F_{g_2;q_1,\cdots,q_N},
\quad i\not=j;\\
&\widetilde K_{ii}\widetilde F_{g;l_1,\cdots,l_N}=
\half\biggl(
\widetilde D_i\widetilde D_i\widetilde F_{g-1;l_1,\cdots,l_N}+
\sum_{\substack{g_1+g_2=g\\p_k+q_k=l_k}}
\widetilde D_i\widetilde F_{g_1;p_1,\cdots,p_N}
\widetilde D_i\widetilde F_{g_2;q_1,\cdots,q_N}\biggr),
\end{split}
\end{equation*}
where $\widetilde K_{ij}$ and $\widetilde D_{i}$
are realizations of the edge-cutting operators $K_{ij}$
and the edge-adding operators $\cD_i$ (see \S \ref{sec:N}) respectively.
\end{Theorem}

Now let us take summation over all sequences $(l_1,\cdots,l_N)$.
We get:

\begin{Corollary}
For evry $2g-2+n>0$, we have
\begin{equation*}
\begin{split}
&\widetilde K_{ij}\widetilde F_{g,n}
=\widetilde D_i\widetilde D_j\widetilde F_{g-1,n}+
\sum_{\substack{g_1+g_2=g\\n_1+n_2=n}}
\widetilde D_i\widetilde F_{g_1,n_1}
\widetilde D_j\widetilde F_{g_2,n_2},
\quad i\not=j;\\
&\widetilde K_{ii}\widetilde F_{g,n}=
\half\biggl(
\widetilde D_i\widetilde D_i\widetilde F_{g-1,n}+
\sum_{\substack{g_1+g_2=g\\n_1+n_2=n}}
\widetilde D_i\widetilde F_{g_1,n_1}
\widetilde D_i\widetilde F_{g_2,n_2}\biggr).
\end{split}
\end{equation*}
\end{Corollary}

The above recursion relations for $\widetilde F_{g,n}$ give us
a way to compute $S_n(z)$ recursively.
According to \cite{gs}, the functions $S_n(z)$ are supposed to satisfy some  recursion relations,
and one may expect that the Schr\"odinger equation \eqref{eq-Schrodinger}
for this quantum spectral curve
can be obtained from these recursion relations,
i.e., the operator
$\widehat A=\widehat A_0+\hbar\widehat A_1+\cdots$
can be solved from these relations.

In the next section,
we will present two examples where the quadratic recursion for $S_n(z)$
can be written down explicitly
(see \S \ref{sec-airy}, \S \ref{sec-catalan}).
In these two examples,
the partition functions are known to be some one-dimensional formal Gaussian integrals,
thus our formalism of the abstract QFT and its realization can be applied
directly to this formal Gaussian integral to derive these quadratic recursions.
The edge-cutting and edge-adding operators are realized by some differential operators.
We will see that these recursions are indeed equivalent to the Schr\"odinger equations.

However,
in general it is not easy to find explicit expressions
for these quadratic recursion relations.
For a given E-O topological recursion,
the realization of the edge-cutting and edge-adding operators operators
can be complicated,
and the recursions for $S_n(z)$ may be difficult to write down explicitly.
We hope to return to this problem in the future.

\section{Application to Topological 1D Gravity and Quantum Spectral Curves}\label{sec7}

In this section we present another class of applications
of our formalism.
We will apply it to topological 1D gravity \cite{zhou2}.
The propagator in this case is holomorphic.
By making particular choices of coupling constants,
this gives us examples concerning the quantum spectral curves \cite{gs}.

\subsection{Some results in topological 1D gravity}

In \cite{zhou2}, the second author defined the partition function of the
topological 1D gravity as a formal Gaussian integral:
\be\label{1d-partition-t}
Z^{1D}=\frac{1}{(2\pi\lambda^2)^\frac{1}{2}}\int dx
\exp\big[\frac{1}{\lambda^2}S(x)\big],
\ee
where the action $S(x)$ is
\be\label{1d-action-t}
S(x)=-\frac{1}{2}x^2+\sum_{n\geq 1}t_{n-1}\frac{x^n}{n!}.
\ee
The following explicit formula is given in \cite[Prop. 4.5]{zhou2}:
\be \label{eqn:Formula-for-Z}
Z = \sum_{n \geq 0} \sum_{\sum_{j=1}^k m_j j =2n}
\frac{(2n-1)!!}{\prod_{j=1}^k (j!)^{m_j} m_j!}
\lambda^{2n-2\sum_{j=1}^k m_j} \cdot \prod_{j=1}^k t_{j-1}^{m_j}.
\ee
Using Wick's theorem,
it is also straightforward to write down the Feynman graphs and the Feynman rules£º
\be \label{eqn:Feynman-Z}
Z = \sum_{\Gamma \in \cG} \frac{1}{|\Aut(\Gamma)|} \prod_{v\in V(\Gamma)} \lambda^{\val(v)-2} t_{\val(v)-1},
\ee
where the sum is taken over the set $\cG$ all possible graphs,
with the following Feynman rules:
\bea
&& w(v) = \lambda^{\val(v)-2} t_{\val(v)-1}, \\
&& w(e) = 1.
\eea

One of the results obtained in \cite{zhou2} is that after a natural change of variables,
the free energy functions $F_g$ have nice simple expressions.
This coordinate change can be obtained by the idea of renormalization.
Alternatively,
one can formally apply the stationary phase method and
find the critical point of $S(x)$ and consider the Taylor expansion of $S(x)$
at its critical point.
This gives us the following change of coordinates:
\be\label{1d-coord1}
\begin{split}
&I_0=\sum_{k=1}^{\infty}\frac{1}{k}\sum_{p_1+\cdots+p_k=k-1}\frac{t_{p_1}}{p_1!}\cdots\frac{t_{p_k}}{p_k!},\\
&I_k=\sum_{n\geq 0}t_{n+k}\frac{I_0^n}{n!},\quad k\geq 1,
\end{split}
\ee
and
\be\label{1d-coord12}
t_k=\sum_{n=0}^{\infty}\frac{(-1)^n I_0^n}{n!}I_{n+k},
\ee
where $I_0$ is the critical point of $S(x)$, and $I_n-\delta_{n,1}$ ($n\geq 1$) are
the Taylor coefficients at $x=I_0$.
Using the new coordinate $\{I_k\}_{k\geq 0}$,
the action \eqref{1d-action-t} can be rewritten as
\be\label{1d-action-I}
S(x)=\sum_{k=0}^\infty \frac{(-1)^k}{(k+1)!}(I_k+\delta_{k,1})I_0^{k+1}+
\sum_{n=2}^{\infty}\frac{(I_{n-1}-\delta_{n,2})}{n!}(x-I_0)^n.
\ee
Then the partition function \eqref{1d-partition-t} can be rewritten in terms
of another formal Gaussian integral:
\be
\begin{split}
Z^{1D}=&\exp\biggl[\frac{1}{\lambda^2}\sum_{k=0}^{\infty}\frac{(-1)^k}{(k+1)!}
(I_k+\delta_{k,1})I_0^{k+1}+\frac{1}{2}\log\frac{1}{1-I_1}\biggr]\\
&\cdot\frac{1}{(\pi\lambda^2)^\frac{1}{2}}\int dx
\exp\biggl[\frac{1}{\lambda^2}(-\frac{1}{2}x^2+
\sum_{n\geq 3}\frac{I_{n-1}}{(1-I_1)^\frac{n}{2}}\cdot\frac{x^n}{n!})\biggr].
\end{split}
\ee
This is interpreted in \cite{zhou2} as giving a mean field theory of the topological 1D gravity.
Using this new equivalent theory,
the free energy $F^{1D}=\sum_{g\geq 0}\lambda^{2g-2}F_g^{1D}=\log(Z^{1D})$ of this theory
can be written down as follows:
\be
\begin{split}
&F_0^{1D}=\sum_{k=0}^{\infty}\frac{(-1)^k}{(k+1)!}(I_k+\delta_{k,1})I_0^{k+1},\\
&F_1^{1D}=\frac{1}{2}\log\frac{1}{1-I_1},\\
&F_g^{1D}=\sum_{\sum_{j=2}^{2g-1}\frac{j-1}{2}l_j=g-1}
\langle \tau_2^{l_2}\cdots\tau_{2g-1}^{l_{2g-1}}\rangle_g
\prod_{j=2}^{2g-1}\frac{1}{l_j!}
\biggl(\frac{I_j}{(1-I_1)^{(j+1)/2}}\biggr)^{l_j},g\geq2.
\end{split}
\ee

By rewriting the theory in the new coordinates $\{I_k\}$,
the free energy $F^{1D}$ has a different Feynman diagram expansion \cite[(130)]{zhou2}:
\be \label{eqn:F-1D}
F_g^{1D}=\sum_{\Gamma\in\cG_{g,0}^{c,0}}\frac{1}{|\Aut(\Gamma)|}\prod_{v\in V(\Gamma)}I_{\val(v)-1}\cdot\prod_{e\in E(\Gamma)}\frac{1}{1-I_1},
\ee
for $g\geq 2$,
where $\cG_{g,0}^{c,0}$ is the subset of $\cG_{g,0}^c$ consists of
{\em stable} graphs whose vertices are all of genus $0$.

\subsection{Topological 1D gravity as realization of the diagrammatics of stable graphs}

\label{sec:1D-Rec}

Now we see that the Feynman rules \eqref{eqn:F-1D}
is indeed a special case of the realization of diagrammatics of stable graphs
described in Section \ref{sec4}.
I.e.,
$F_g^{1D}= \wF_g$ ($g \geq 2$) for suitable $\kappa$ and a suitable sequence
$F_g(X)$ for $g \geq 0$.
By comparing \eqref{eqn:F-1D} with the Feynman rules \eqref{omega},
it is clear that for topological 1D gravity,
the propagator is taken to be
\be \label{kappa-I1}
\kappa=\frac{1}{1-I_1}.
\ee
Since in \eqref{eqn:F-1D},
the summation is taken over stable graphs with only genus zero vertices,
so we have $F_g(X) = 0$ for $g \geq 1$,
and
\be \label{I-n-1}
F_0^{(\val(v))}(X)= I_{\val(v)-1},
\ee
where $\val(v)\geq 3$.
In other words,
we need a function $F_0(X)$ such that
\be \label{F-I-n-1}
\frac{d^n}{dX^n} F_0(X) = I_{n-1}
\ee
for $n \geq 3$.
Since $I_{n-1}$ involves an infinite set of indeterminates $t_0, t_1, \dots$,
it is not clear how to choose one of them or a combination of them to be $X$.
We will take a different approach  as follows.
Regard $\{I_k\}_{k\geq 0}$ as some parameters, and $X$ as a new independent variable.
Now we define a function $F_0(X)$ by
\be \label{eqn:F0T}
F_0(X)=\sum_{k=0}^\infty \frac{(-1)^k}{(k+1)!}(I_k+\delta_{k,1})I_0^{k+1}+
\sum_{n=2}^{\infty}\frac{(I_{n-1}-\delta_{n,2})}{n!}(X-I_0)^n.
\ee
Write $\tilde X=X-I_0$, then
\be
F_0(I_0+\tilde X)=\sum_{k=0}^\infty \frac{(-1)^k}{(k+1)!}(I_k+\delta_{k,1})I_0^{k+1}+
\sum_{n=2}^{\infty}\frac{(I_{n-1}-\delta_{n,2})}{n!}\tilde X^n.
\ee
This is equivalent to say
\be\label{1d-evaluation}
\begin{split}
&F_0(I_0)=\sum_{k=0}^\infty \frac{(-1)^k}{(k+1)!}(I_k+\delta_{k,1})I_0^{k+1},\\
&\frac{d F_0(X)}{dX}\biggl|_{X=I_0}=0,\\
&\frac{d^n F_0(X)}{dX^n}\biggl|_{X=I_0}=I_{n-1}-\delta_{n,2},\quad n\geq 2.
\end{split}
\ee

Now using the formalism in Section \ref{sec4}, we can construct a sequence
$\wF_g(X)$ from $F(X)=\lambda^{-2}\cdot F_0(X)$,
and the propagator
\be\label{1d-kappa}
\kappa=-\frac{1}{F_0''(X)}
\ee
which clearly satisfies the condition \eqref{constraint}.
Hence \eqref{kappa-I1}, \eqref{I-n-1}and \eqref{F-I-n-1} hold when $X = I_0$.
For $g\geq 2$ we have
\be
\wF_g(X)=\sum_{\Gamma\in\cG_{g,0}^{c,0}}\frac{1}{|\Aut(\Gamma)|}
\biggl(\prod_{v\in V(\Gamma)}F_0^{(\val(v))}(X)\biggr)\cdot\kappa^{|E(\Gamma)|},
\ee
Then evaluate this expression at $X=I_0$ using \eqref{1d-evaluation} and \eqref{1d-kappa}, we get
\be
\begin{split}
\wF_g(I_0)=&\sum_{\Gamma\in\cG_{g,0}^{c,0}}\frac{1}{|\Aut(\Gamma)|}
\biggl(\prod_{v\in V(\Gamma)}F_0^{(\val(v))}(X)\biggl|_{X=I_0}\biggr)
\cdot(-\frac{1}{F_0''(X)|_{X=I_0}})^{|E(\Gamma)|}\\
=&\sum_{\Gamma\in\cG_{g,0}^{c,0}}\frac{1}{|\Aut(\Gamma)|}
\biggl(\prod_{v\in V(\Gamma)}I_{\val(v)-1}\biggr)\cdot\biggl(\frac{1}{1-I_1}\biggr)^{|E(\Gamma)|},
\end{split}
\ee
this is exactly the Feynman rule \eqref{eqn:F-1D}, i.e.,
we have proved the following

\begin{Theorem} \label{thm:F-1D}
The free energy of the topological 1D gravity
can be realized by the abstract free energy:
\be
F_g^{1D}(I_0,I_1,I_2,\cdots)=\wF_g(X)|_{X=I_0}
\ee
for $g\geq 2$.
Here $F=\lambda^{-2} F_0$ and $\kappa$ are given by \eqref{eqn:F0T} and \eqref{1d-kappa}
respectively.
\end{Theorem}

\begin{Example}
$F_2^{1D}$ and $F_3^{1D}$ are explicitly given by:
\begin{flalign*}
\begin{tikzpicture}
\node [align=center,align=center] at (0,0) {$F_2^{1D}=\frac{1}{12}$};
\draw (1,0) circle [radius=0.05];
\draw (2,0) circle [radius=0.05];
\draw (1.05,0)--(1.95,0);
\draw (1.03,0.03) .. controls (1.35,0.3) and (1.65,0.3) ..  (1.97,0.03);
\draw (1.03,-0.03) .. controls (1.35,-0.3) and (1.65,-0.3) ..  (1.97,-0.03);
\node [align=center,align=center] at (2.6,0) {$+\frac{1}{8}$};
\draw (3.5,0) circle [radius=0.05];
\draw (4,0) circle [radius=0.05];
\draw (3.55,0)--(3.95,0);
\draw (3.47,0.03) .. controls (2.9,0.5) and (2.9,-0.5) ..  (3.47,-0.03);
\draw (4.03,0.03) .. controls (4.6,0.5) and (4.6,-0.5) ..  (4.03,-0.03);
\node [align=center,align=center] at (4.9,0) {$+\frac{1}{8}$};
\draw (5.8,0) circle [radius=0.05];
\draw (5.77,0.03) .. controls (5.2,0.5) and (5.2,-0.5) ..  (5.77,-0.03);
\draw (5.83,0.03) .. controls (6.4,0.5) and (6.4,-0.5) ..  (5.83,-0.03);
\end{tikzpicture},&&
\end{flalign*}

\begin{flalign*}
\begin{split}
&\begin{tikzpicture}
\node [align=center,align=center] at (-0.4,0) {$F_3^{1D}=\frac{1}{48}$};
\draw (1,0) circle [radius=0.05];
\draw (0.96,0.02) .. controls (0.4,0.3) and (0.6,-0.6) ..  (0.96,-0.02);
\draw (1.04,0.02) .. controls (1.6,0.3) and (1.4,-0.6) ..  (1.04,-0.02);
\draw (0.97,0.03) .. controls (0.5,0.6) and (1.5,0.6) ..  (1.03,0.03);
\node [align=center,align=center] at (2,0) {$+\frac{1}{48}$};
\draw (2.6,0) circle [radius=0.05];
\draw (3.8,0) circle [radius=0.05];
\draw (2.64,0.02) .. controls (3,0.1) and (3.4,0.1) ..  (3.76,0.02);
\draw (2.64,-0.02) .. controls (3,-0.1) and (3.4,-0.1) ..  (3.76,-0.02);
\draw (2.63,0.03) .. controls (3,0.35) and (3.4,0.35) ..  (3.77,0.03);
\draw (2.63,-0.03) .. controls (3,-0.35) and (3.4,-0.35) ..  (3.77,-0.03);
\node [align=center,align=center] at (4.4,0) {$+\frac{1}{16}$};
\draw (5.3,0) circle [radius=0.05];
\draw (6,0) circle [radius=0.05];
\draw (5.27,0.03) .. controls (4.7,0.5) and (4.7,-0.5) ..  (5.27,-0.03);
\draw (6.03,0.03) .. controls (6.6,0.5) and (6.6,-0.5) ..  (6.03,-0.03);
\draw (5.32,0.04)--(5.98,0.04);
\draw (5.32,-0.04)--(5.98,-0.04);
\node [align=center,align=center] at (7,0) {$+\frac{1}{16}$};
\draw (7.9,0) circle [radius=0.05];
\draw (8.5,0) circle [radius=0.05];
\draw (7.95,0)--(8.45,0);
\draw (8.53,0.03) .. controls (9.1,0.5) and (9.1,-0.5) ..  (8.53,-0.03);
\draw (7.87,0.03) .. controls (7.55,0.8) and (7.2,0) ..  (7.87,0);
\draw (7.87,-0.03) .. controls (7.55,-0.8) and (7.2,0) ..  (7.87,0);
\end{tikzpicture}
\\
&\qquad\quad\begin{tikzpicture}
\node [align=center,align=center] at (0,0) {$+\frac{1}{12}$};
\draw (1,0) circle [radius=0.05];
\draw (1.8,0) circle [radius=0.05];
\draw (0.97,0.03) .. controls (0.4,0.5) and (0.4,-0.5) ..  (0.97,-0.03);
\draw (1.05,0)--(1.75,0);
\draw (1.03,0.03) .. controls (1.3,0.2) and (1.5,0.2) ..  (1.77,0.03);
\draw (1.03,-0.03) .. controls (1.3,-0.2) and (1.5,-0.2) ..  (1.77,-0.03);
\node [align=center,align=center] at (2.3,0) {$+\frac{1}{8}$};
\draw (3.2,0) circle [radius=0.05];
\draw (3.9,0) circle [radius=0.05];
\draw (4.6,0) circle [radius=0.05];
\draw (3.25,0)--(3.85,0);
\draw (3.92,0.04)--(4.58,0.04);
\draw (3.92,-0.04)--(4.58,-0.04);
\draw (3.17,0.03) .. controls (2.6,0.5) and (2.6,-0.5) ..  (3.17,-0.03);
\draw (4.63,0.03) .. controls (5.2,0.5) and (5.2,-0.5) ..  (4.63,-0.03);
\node [align=center,align=center] at (5.6,0) {$+\frac{1}{16}$};
\draw (6.5,0) circle [radius=0.05];
\draw (7.1,0) circle [radius=0.05];
\draw (7.7,0) circle [radius=0.05];
\draw (6.47,0.03) .. controls (5.9,0.5) and (5.9,-0.5) ..  (6.47,-0.03);
\draw (7.73,0.03) .. controls (8.3,0.5) and (8.3,-0.5) ..  (7.73,-0.03);
\draw (7.07,0.03) .. controls (6.6,0.6) and (7.6,0.6) ..  (7.13,0.03);
\draw (6.55,0)--(7.05,0);
\draw (7.15,0)--(7.65,0);
\end{tikzpicture}
\\
&\qquad\quad\begin{tikzpicture}
\node [align=center,align=center] at (0,0) {$+\frac{1}{12}$};
\draw (1,0) circle [radius=0.05];
\draw (1.5,0) circle [radius=0.05];
\draw (2.2,0) circle [radius=0.05];
\draw (1.05,0)--(1.45,0);
\draw (0.97,0.03) .. controls (0.4,0.5) and (0.4,-0.5) ..  (0.97,-0.03);
\draw (1.55,0)--(2.15,0);
\draw (1.53,0.03) .. controls (1.7,0.2) and (2,0.2) ..  (2.17,0.03);
\draw (1.53,-0.03) .. controls (1.7,-0.2) and (2,-0.2) ..  (2.17,-0.03);
\node [align=center,align=center] at (2.7,0) {$+\frac{1}{8}$};
\draw (3.2,-0.2) circle [radius=0.05];
\draw (4.2,-0.2) circle [radius=0.05];
\draw (3.7,0.3) circle [radius=0.05];
\draw (3.22,-0.16)--(4.18,-0.16);
\draw (3.22,-0.24)--(4.18,-0.24);
\draw (3.22,-0.16)--(3.67,0.27);
\draw (4.18,-0.16)--(3.73,0.27);
\draw (3.67,0.33) .. controls (3.2,0.9) and (4.2,0.9) ..  (3.73,0.33);
\node [align=center,align=center] at (4.7,0) {$+\frac{1}{8}$};
\draw (5.2,-0.2) circle [radius=0.05];
\draw (6.2,-0.2) circle [radius=0.05];
\draw (5.7,0.3) circle [radius=0.05];
\draw (5.25,-0.2)--(6.15,-0.2);
\draw (5.25,-0.2)--(5.7,0.25);
\draw (6.15,-0.2)--(5.7,0.25);
\draw (5.2,-0.15)--(5.65,0.3);
\draw (6.2,-0.15)--(5.75,0.3);
\node [align=center,align=center] at (6.8,0) {$+\frac{1}{48}$};
\draw (8.3,-0.1) circle [radius=0.05];
\draw (7.8,-0.1) circle [radius=0.05];
\draw (8.8,-0.1) circle [radius=0.05];
\draw (8.3,0.3) circle [radius=0.05];
\draw (8.27,0.33) .. controls (7.8,0.9) and (8.8,0.9) ..  (8.33,0.33);
\draw (7.85,-0.1)--(8.25,-0.1);
\draw (8.35,-0.1)--(8.75,-0.1);
\draw (8.3,-0.05)--(8.3,0.25);
\draw (7.77,-0.07) .. controls (7.2,0.4) and (7.2,-0.6) ..  (7.77,-0.13);
\draw (8.83,-0.07) .. controls (9.4,0.4) and (9.4,-0.6) ..  (8.83,-0.13);
\end{tikzpicture}
\\
&\qquad\quad\begin{tikzpicture}
\node [align=center,align=center] at (0,0) {$+\frac{1}{24}$};
\draw (1,0) circle [radius=0.05];
\draw (0.5,-0.3) circle [radius=0.05];
\draw (1.5,-0.3) circle [radius=0.05];
\draw (1,0.5) circle [radius=0.05];
\draw (0.55,-0.3)--(1.45,-0.3);
\draw (0.53,-0.27)--(0.97,-0.03);
\draw (1.47,-0.27)--(1.03,-0.03);
\draw (0.5,-0.25)--(0.97,0.47);
\draw (1.5,-0.25)--(1.03,0.47);
\draw (1,0.05)--(1,0.45);
\node [align=center,align=center] at (2.1,0) {$+\frac{1}{16}$};
\draw (2.6,-0.3) circle [radius=0.05];
\draw (2.6,0.3) circle [radius=0.05];
\draw (3.4,-0.3) circle [radius=0.05];
\draw (3.4,0.3) circle [radius=0.05];
\draw (2.6,-0.25)--(2.6,0.25);
\draw (3.4,-0.25)--(3.4,0.25);
\draw (2.62,0.34)--(3.38,0.34);
\draw (2.62,0.26)--(3.38,0.26);
\draw (2.62,-0.34)--(3.38,-0.34);
\draw (2.62,-0.26)--(3.38,-0.26);
\node [align=center,align=center] at (4,0) {$+\frac{1}{16}$};
\draw (4.9,0) circle [radius=0.05];
\draw (5.4,0) circle [radius=0.05];
\draw (5.9,0) circle [radius=0.05];
\draw (6.4,0) circle [radius=0.05];
\draw (4.87,0.03) .. controls (4.3,0.5) and (4.3,-0.5) ..  (4.87,-0.03);
\draw (6.43,0.03) .. controls (7,0.5) and (7,-0.5) ..  (6.43,-0.03);
\draw (4.95,0)--(5.35,0);
\draw (5.95,0)--(6.35,0);
\draw (5.42,0.04)--(5.88,0.04);
\draw (5.42,-0.04)--(5.88,-0.04);
\node [align=center,align=center] at (7.3,0) {$+\frac{1}{8}$};
\draw (8.2,0) circle [radius=0.05];
\draw (8.8,0) circle [radius=0.05];
\draw (9.3,0.3) circle [radius=0.05];
\draw (9.3,-0.3) circle [radius=0.05];
\draw (8.17,0.03) .. controls (7.6,0.5) and (7.6,-0.5) ..  (8.17,-0.03);
\draw (8.25,0)--(8.75,0);
\draw (8.83,0.03)--(9.25,0.3);
\draw (8.83,-0.03)--(9.25,-0.3);
\draw (9.26,0.28)--(9.26,-0.28);
\draw (9.34,0.28)--(9.34,-0.28);
\end{tikzpicture},
\end{split}&&
\end{flalign*}
or,
\be\label{1d-f2f3}
\begin{split}
&F_2^{1D}=\frac{5}{24}I_2^2\kappa^3+\frac{1}{8}I_3\kappa^2,\\
&F_3^{1D}=\frac{1}{48}I_5\kappa^3+(\frac{1}{12}I_3^2+
\frac{7}{48}I_2I_4)\kappa^4+\frac{25}{48}I_2^2I_3\kappa^5+
\frac{5}{16}I_2^4\kappa^6,
\end{split}
\ee
where $\kappa=\frac{1}{1-I_1}$.
\end{Example}

As a corollary to Theorem \ref{1-recursion} and  Theorem \ref{thm:F-1D},
we get a quadratic recursion relation for $F^{1D}$ as follows.
First we have
\be
\pd_\kappa \wF_g = \frac{1}{2}((\pd_{X}+\kappa F_0'''(X))\pd_{X}\wF_{g-1}+
\sum^{g-1}_{r=1} \pd_{X}\wF_r\pd_{X}\wF_{g-r}),
\ee
where $\wF_1 = F_1^{1D} = \half \log \kappa$,
next we restrict  to $X=I_0$ to get:

\begin{Theorem}For every $g\geq 2$,
\be \label{eqn:Q-Rec}
\pd_\kappa F_g^{1D} = \frac{1}{2}\big((d_{X}+\kappa I_2)d_{X}F_{g-1}^{1D}+
\sum^{g-1}_{r=1} d_{X}F_r^{1D}\cdot d_{X}F_{g-r}^{1D}\big),
\ee
where $d_X$ is a differential operator such that
\begin{align} \label{eqn:d-X}
d_X I_k & =I_{k+1},\;\; k \geq 1,  & d_X \kappa & =\kappa^2\cdot I_2,
\end{align}
i.e., $d_X:=\sum_{k\geq 1}I_{k+1}\frac{\pd}{\pd I_k}$.
\end{Theorem}

This theorem provides a way to compute $F_g^{1D}$ recursively
without listing all possible Feynman graphs.

\begin{Example}
For $g=4$, we have
\ben
&& 2\frac{\partial F_4^{1D}}{\partial\kappa}
=(d_{X}+\kappa I_2)\partial_{I_0} F_3^{1D}+2d_X F_1^{1D}\cdot d_X F_3^{1D}
+(d_X F_2^{1D})^2\\
&& \quad =\frac{1}{48}I_7\kappa^3+(\frac{5}{16}I_2I_6+\frac{25}{48}I_3I_5+\frac{21}{64}I_4^2)\kappa^4+(\frac{113}{48}I_2^2I_5+\frac{33}{24}I_3^3+7I_2I_3I_4)\kappa^5\\
&& \quad +(\frac{3115}{144}I_2^2I_3^2+\frac{1127}{96}I_2^3I_4)\kappa^6+\frac{985}{24}I_2^4I_3\kappa^7+\frac{1105}{64}I_2^6\kappa^8,
\een
Then $F_4^{1D}$ can be obtained by integrate with respext with $\kappa$:
\be\label{1d-f4}
\begin{split}
F_4^{1D}=&\frac{1}{384}I_7\kappa^4+(\frac{1}{32}I_2I_6+\frac{5}{96}I_3I_5+\frac{21}{640}I_4^2)\kappa^5\\
&+(\frac{113}{576}I_2^2I_5+\frac{11}{96}I_3^3+\frac{7}{12}I_2I_3I_4)\kappa^6+(\frac{445}{288}I_2^2I_3^2+\frac{161}{192}I_2^3I_4)\kappa^7\\
&+\frac{985}{384}I_2^4I_3\kappa^8+\frac{1105}{1152}I_2^6\kappa^9.
\end{split}
\ee
\end{Example}

\subsection{Example: The Airy function and quantum Airy curve}
\label{sec-airy}

Our first example is the Airy function
\be\label{airy-fn}
Z^A(t)=C\int d\eta\exp\big[ - \lambda^{-2} (-\eta t+\frac{\eta^3}{3})\big],
\ee
where $C$ is a suitable normalization constant.
Let $f(\eta)=-\eta t+\frac{\eta^3}{3}$, then the critical points of $f(\eta)$
is $\eta=\pm t^{\frac{1}{2}}$. Denote by $c$ one of the critical points
and $\tilde\eta=\eta-c$, then
\be
f(\eta)=-(\tilde\eta+c)c^2+\frac{1}{3}(\tilde\eta+c)^3
=-\frac{2}{3}c^3+c\tilde\eta^2+\frac{1}{3}\tilde\eta^3.
\ee
Therefore in this case
\be\label{a1-I}
\begin{split}
&I_0=c,\quad I_1=1-2c,\quad I_2=-2; \\
&I_k=0,\quad k\geq 3.
\end{split}
\ee
And the propagator is
\be\label{airy-propagator}
\kappa=\frac{1}{1-I_1}=\frac{1}{2c}.
\ee

Using the coordinate change \eqref{1d-coord12}, this particular point \eqref{a1-I}
in the space of coupling constants is
\be
\begin{split}
&t_0=c^2,\quad t_1=1, \quad t_2=-2,\\
&t_k=0, \quad k\geq 3.
\end{split}
\ee

In what follows we will use $\frac{\pd}{\pd c}$ to replace the operator $d_T$ since
we have $\frac{\pd I_k}{\pd c}=I_{k+1}$ for $k\geq 1$
and $\frac{\pd \kappa}{\pd c}=\kappa^2 I_2$.
Now using our formalism we can compute the free energies $F_g^{A}$
associated with $Z^A$.
The Feynman graphs of this integral only involves cubic graphs
without external edges, whose vertices are all of genus zero.

\begin{Example}
The free energies $F_0^{A}$ and $F_1^{A}$ are given by
\be
F_0^{A}=-\frac{2}{3}c^3,\quad F_1^{A}=\frac{1}{2}\log\frac{1}{2c}.
\ee
And $F^{A}_2$ and $F^{A}_3$ are given by
\begin{flalign*}
\begin{tikzpicture}
\node [align=center,align=center] at (0.1,0) {$F_2^{A}=\frac{1}{12}$};
\draw (1,0) circle [radius=0.05];
\draw (2,0) circle [radius=0.05];
\draw (1.05,0)--(1.95,0);
\draw (1.03,0.03) .. controls (1.35,0.3) and (1.65,0.3) ..  (1.97,0.03);
\draw (1.03,-0.03) .. controls (1.35,-0.3) and (1.65,-0.3) ..  (1.97,-0.03);
\node [align=center,align=center] at (2.7,0) {$+\frac{1}{8}$};
\draw (3.5+0.2,0) circle [radius=0.05];
\draw (4+0.2,0) circle [radius=0.05];
\draw (3.55+0.2,0)--(3.95+0.2,0);
\draw (3.47+0.2,0.03) .. controls (2.9+0.2,0.5) and (2.9+0.2,-0.5) ..  (3.47+0.2,-0.03);
\draw (4.03+0.2,0.03) .. controls (4.6+0.2,0.5) and (4.6+0.2,-0.5) ..  (4.03+0.2,-0.03);
\end{tikzpicture},&&
\end{flalign*}

\begin{flalign*}
\begin{split}
&\begin{tikzpicture}
\node [align=center,align=center] at (-0.3,0) {$F_3^{A}=\frac{1}{8}$};
\draw (1,0) circle [radius=0.05];
\draw (1.5,0) circle [radius=0.05];
\draw (1.05,0)--(1.45,0);
\draw (2,0.3) circle [radius=0.05];
\draw (2,-0.3) circle [radius=0.05];
\draw (1.53,0.03)--(1.97,0.27);
\draw (1.53,-0.03)--(1.97,-0.27);
\draw (1.96,0.28) -- (1.96,-0.28);
\draw (2.04,0.28) -- (2.04,-0.28);
\draw (0.97,0.03) .. controls (0.4,0.5) and (0.4,-0.5) ..  (0.97,-0.03);
\node [align=center,align=center] at (2.7,0) {$+\frac{1}{16}$};
\draw (3.5+0.2,0) circle [radius=0.05];
\draw (4+0.2,0) circle [radius=0.05];
\draw (4.5+0.2,0) circle [radius=0.05];
\draw (5+0.2,0) circle [radius=0.05];
\draw (3.47+0.2,0.03) .. controls (2.9+0.2,0.5) and (2.9+0.2,-0.5) ..  (3.47+0.2,-0.03);
\draw (5.03+0.2,0.03) .. controls (5.6+0.2,0.5) and (5.6+0.2,-0.5) ..  (5.03+0.2,-0.03);
\draw (3.55+0.2,0) -- (3.95+0.2,0);
\draw (4.55+0.2,0) -- (4.95+0.2,0);
\draw (4.02+0.2,0.04) -- (4.48+0.2,0.04);
\draw (4.02+0.2,-0.04) -- (4.48+0.2,-0.04);
\node [align=center,align=center] at (6.3,0) {$+\frac{1}{16}$};
\draw (6.6+0.4,0.3) circle [radius=0.05];
\draw (6.6+0.4,-0.3) circle [radius=0.05];
\draw (7.3+0.4,0.3) circle [radius=0.05];
\draw (7.3+0.4,-0.3) circle [radius=0.05];
\draw (6.6+0.4,0.25) -- (6.6+0.4,-0.25);
\draw (7.3+0.4,0.25) -- (7.3+0.4,-0.25);
\draw (6.62+0.4,0.34) -- (7.28+0.4,0.34);
\draw (6.62+0.4,0.26) -- (7.28+0.4,0.26);
\draw (6.62+0.4,-0.34) -- (7.28+0.4,-0.34);
\draw (6.62+0.4,-0.26) -- (7.28+0.4,-0.26);
\end{tikzpicture}
\\
&\qquad\quad\begin{tikzpicture}
\node [align=center,align=center] at (-0.1,0) {$+\frac{1}{24}$};
\draw (0.8,0) circle [radius=0.05];
\draw (1.5,0) circle [radius=0.05];
\draw (0.85,0)--(1.45,0);
\draw (2,0.4) circle [radius=0.05];
\draw (2,-0.4) circle [radius=0.05];
\draw (1.53,0.03)--(1.97,0.37);
\draw (1.53,-0.03)--(1.97,-0.37);
\draw (2,0.35) -- (2,-0.35);
\draw (0.83,0.03)--(1.97,0.37);
\draw (0.83,-0.03)--(1.97,-0.37);
\node [align=center,align=center] at (2.9,0) {$+\frac{1}{48}$};
\draw (4+0.2,0) circle [radius=0.05];
\draw (4.7+0.2,0) circle [radius=0.05];
\draw (5.3+0.2,0.4) circle [radius=0.05];
\draw (5.3+0.2,-0.4) circle [radius=0.05];
\draw (4.05+0.2,0)--(4.65+0.2,0);
\draw (4.73+0.2,0.03)--(5.27+0.2,0.37);
\draw (4.73+0.2,-0.03)--(5.27+0.2,-0.37);
\draw (3.97+0.2,0.03) .. controls (3.4+0.2,0.5) and (3.4+0.2,-0.5) ..  (3.97+0.2,-0.03);
\draw (5.33+0.2,0.43) .. controls (5.9+0.2,0.9) and (5.9+0.2,-0.1) ..  (5.33+0.2,0.37);
\draw (5.33+0.2,-0.43) .. controls (5.9+0.2,-0.9) and (5.9+0.2,0.1) ..  (5.33+0.2,-0.37);
\end{tikzpicture},
\end{split}&&
\end{flalign*}
or,
\be\label{a1-f2f3}
F_2^{A}=\frac{5}{48c^3},\quad F_3^{A}=\frac{5}{64c^6}.
\ee
\end{Example}

Using \eqref{airy-propagator}, for $g\geq 1$ we have
\ben
\partial_\kappa F^{A}_g=(-2c^2)\cdot\frac{\partial F^{A}_g}{\partial c},
\een
thus Theorem \ref{1-recursion} give us the following recursion relation.
\begin{Theorem}
For $g\geq 2$ we have
\be\label{a1-recursion}
\frac{\pd F_g^{A}}{\pd c} =-\frac{1}{4c^2}\big[(\frac{\pd}{\pd c}
-\frac{1}{c})\frac{\pd F_{g-1}^{A}}{\pd c}+
\sum_{r=1}^{g-1}\frac{\pd F_{r}^{A}}{\pd c}\frac{\pd F_{g-r}^{A}}{\pd c} \big].
\ee
\end{Theorem}

Note that in this case $F^{A}_{g}$ must be a monomial in $\kappa=\frac{1}{2c}$
of degree $3g-3$, thus $ F^{A}_{g}$ can be solved recursively using this relation. For example,
\be\label{a1-f4f5}
F_4^{A}=\frac{1105}{9216c^9},\quad F_5^{A}=\frac{565}{2048c^{12}},\quad\cdots
\ee
The recursion relation \eqref{a1-recursion} is equivalent to the equation
\be
\big((\hbar\frac{\pd}{\pd t})^2-t\big)Z^A(t)=0
\ee
if we set $\hbar=\lambda^2$.
This is the Schr\"odinger equation of the quantum Airy curve.

Now let us recall some results about the quantum Airy curve \cite{gs, zhou1}
and compare with the above construction.
First let $u(z)=\frac{1}{2}z^2$, $v(z)=z$ be a parametrization of the Airy curve
\ben
A(u,v)=\frac{1}{2}v^2-u=0,
\een
then using the Eynard-Orantin topological recursion \cite{eo} one may get a family of differentials
\ben
W_{g,n}(p_1,\cdots,p_n)=\mathcal{W}_{g,n}(z_1,\cdots,z_n)dz_1\cdots dz_n.
\een
These invariants associated to the Airy curve encode the information of the Witten-Kontsevich tau-function \cite{kon, wit2, zhou3}. Following \cite{gs}, define
\be
Z_{Airy}=\exp(\sum_{n=0}^{\infty}\hbar^{n-1}S_n),
\ee
where  $S_n$ are defined in \eqref{eqn:GS}.
Then $S_k$ can be used to construct the quantization of the Airy curve.
The following result was proposed by Gukov-Su{\l}kowski in \cite{gs}, and proved in \cite{zhou1}.

\begin{Theorem}(\cite{zhou1})\label{thm-airy}
The function $Z_{Airy}$ satisfies
\be\label{airy-eqn}
\hat{A}Z_{Airy}=0,
\ee
where
\be
\hat{A}=\frac{1}{2}\hat{v}^2-\hat{u}=\frac{1}{2}(\hbar\partial_u)^2-u\cdot
\ee
is the quantization of the Airy curve.
\end{Theorem}

One computes
\be\label{airy-S}
\begin{split}
&S_0=\frac{1}{3}z^3,\quad S_1=-\frac{1}{2}\log(z),\\
&S_2=\frac{5}{24z^3},\quad
S_3=\frac{5}{16z^6}, \quad S_4=\frac{1105}{1152z^9}, \quad\cdots
\end{split}
\ee
We make a comparison between this and \eqref{a1-f2f3}, \eqref{a1-f4f5},
then we might simply expect $F_g^{A}=S_g$ for $g\geq 2$ once we set $z^3=2c^3$.
In fact, this can be proved by the following recursion relation.

\begin{Lemma}(\cite{zhou1})
For $n>2$, we have
\be
\frac{1}{2}\partial_u^2 S_{n-1}+\partial_u S_0\cdot\partial_u S_n+\partial_u S_1\cdot\partial_u S_{n-1}+\frac{1}{2}\sum_{\substack{i+j=n\\i,j\geq 2}}\partial_u S_i\cdot\partial_u S_j=0.
\ee
\end{Lemma}

This lemma can be obtained from Theorem \ref{thm-airy} by
expanding \eqref{airy-eqn} directly as a series in $\hbar$ and
comparing the coefficients. Using $z^3=2c^3$ and $u=\frac{1}{2}z^2$,
this recursion is equivalent to \eqref{a1-recursion}.

In conclusion, for the particular choices \eqref{a1-I} of coupling constants,
the procedure to produce $ F^{A}_g$ is equivalent to the solution of
quantization of Airy curve obtained by the Eynard-Orantin topological recursion.

\subsection{Example: The Kontsevich-Penner matrix model for $N=1$}
\label{sec-catalan}

Our next example is the Kontsevich-Penner matrix model for $N=1$ \cite{kon,pen},
which is related to the quantization of the Catalan curve
\ben
x=z+\frac{1}{z},
\een
see \cite{gs, ms}. The Schr\"odinger equation in this case is
\be\label{kon-pen schr}
[\hbar^2\frac{\partial^2}{\partial t^2}+\hbar t\frac{\partial}{\partial t}+(1-\hbar)]Z(t)=0,
\ee
where the partition function is the formal integral
\be
Z^C(t) = \int d\eta \exp\big[-\lambda^{-2}(\frac{1}{2}\eta^2+\eta t+log\eta)\big],
\ee
up to a normalization constant,
and we set $\hbar=\lambda^2$.

Now let us understand this integral using the model of topological 1D gravity.
Write $f(\eta)=\frac{1}{2}\eta^2+\eta t+\log(\eta)$,
we expand $f(\eta)$ at its critical points $c=\frac{-t\pm\sqrt{t^2-4}}{2}$.
Let $\tilde{\eta}=\eta-c$, then
\be
\begin{split}
f(\eta)=&f(\tilde\eta+c)\\
=&(-\frac{c^2}{2}-1+\log(c))+\frac{1}{2}(1-c^{-2})\tilde\eta^2+\frac{\tilde\eta^3}{3c^3}-\frac{\tilde\eta^4}{4c^4}+\frac{\tilde\eta^5}{5c^5}-\cdots
\end{split}
\ee
By taking
\be
\kappa=\frac{1}{1-c^{-2}}
\ee
and
\be
I_k=\frac{(-1)^{k+1}\cdot k!}{c^{k+1}},\quad k\geq 1,
\ee
we see $I_k'=I_{k+1}$ and $\kappa'=\kappa^2\cdot I_2$ where the prime means
taking derivative with respect to $c$.
Thus we get a special case of the theory in Section \ref{sec:1D-Rec},
thus we obtain a sequence of free energies $ F^{C}_g$.
Using the formula for $ F^{C}_g$ for $g=2,3,4$ given by \eqref{1d-f2f3} and \eqref{1d-f4},
we get:
\ben
 F^{C}_0&=&-\frac{c^2}{2}-1+\log(c),\quad F^{C}_1=\frac{1}{2}\log\frac{1}{1-c^{-2}},\\
 F^{C}_2&=&\frac{3}{4(c^2-1)^2}+\frac{5}{6(c^2-1)^3},\\
 F^{C}_3&=&\frac{5}{2(c^2-1)^3}+\frac{10}{(c^2-1)^4}+\frac{25}{2(c^2-1)^5}+\frac{5}{(c^2-1)^6},\\
 F^{C}_4&=&\frac{105}{8(c^2-1)^4}+\frac{507}{5(c^2-1)^5}+\frac{6391}{24(c^2-1)^6}+\frac{767}{2(c^2-1)^7}\\
&+&\frac{985}{4(c^2-1)^8}+\frac{1105}{18(c^2-1)^9}.
\een

The quadratic recursion relation in Theorem \ref{1-recursion} reads
\be
\partial_\kappa F^{C}_g=\frac{1}{2}[(\frac{\partial}{\partial c}+\kappa\cdot I_2)
\frac{\partial}{\partial c} F^{C}_{g-1}
+\sum_{r=1}^{g-1}\frac{\partial}{\partial c} F^{C}_r\cdot
\frac{\partial}{\partial c} F^{C}_{g-r}],\quad g\geq 2.
\ee
To write this relation completely in terms of $\frac{\partial}{\partial c}$,
we need the following lemma.

\begin{Lemma}
For $g\geq 2$, we have
\be
\frac{\partial}{\partial c} F^{C}_g
=(\kappa^2\cdot I_2-\frac{2\kappa}{c})\partial_\kappa F^{C}_g.
\ee
\end{Lemma}
\begin{proof}
Using the homogeneity we could see that every monomial in $ F^{1D}_g$ is of the form $a_{g,i_1\cdots i_k}I_{i_1}\cdots I_{i_k}\cdot\kappa^l$ with $i_1+\cdots i_k+k=2l$. Thus
\ben
\frac{\frac{\partial}{\partial c}
(I_{i_1}\cdots I_{i_k}\kappa^l)}{\partial_\kappa (I_{i_1}\cdots I_{i_k}\kappa^l)}
&=&\frac{\sum_{j=1}^{k}I_{i_1}\cdots I_{i_j+1}\cdots
I_{i_k}\kappa^l+I_{i_1}\cdots I_{i_k} l\kappa^{l-1}\cdot \kappa^2 I_2}{I_{i_1}
\cdots I_{i_k} l\kappa^{l-1}}\\
&=&\frac{1}{l}\sum_{j=1}^{k}\frac{I_{i_{j}+1}}{I_j}\cdot\kappa+\kappa^2\cdot I_2\\
&=&\frac{1}{l}\sum_{j=1}^{k}\frac{(-1)(i_j+1)}{c}\cdot\kappa+\kappa^2\cdot I_2\\
&=&-\frac{2}{c}\cdot\kappa+\kappa^2\cdot I_2,
\een
which proves the lemma.
\end{proof}

Now using this lemma and
\ben
\kappa^2\cdot I_2-\frac{2\kappa}{c}=-\frac{2c^3}{(c^2-1)^2},
\een
we rewrite the recursion relation as follows.

\begin{Theorem}
For the Kontsevich-Penner model, for $g\geq 2$, we have
\be
\frac{\partial}{\partial c} F^{C}_g=-\frac{c^3}{(c^2-1)^2}
\big[\big(\frac{\partial}{\partial c}-\frac{2}{c(c^2-1)}\big)
\frac{\partial}{\partial c} F^{C}_{g-1}+\sum_{r=1}^{g-1}
\frac{\partial}{\partial c} F^{C}_r\cdot
\frac{\partial}{\partial c} F^{C}_{g-r}\big].
\ee
Or in terms of $t=-(c+\frac{1}{c})$,
\be
(c-\frac{1}{c})\cdot\frac{\partial}{\partial t} F^{C}_g
=\frac{\partial^2}{\partial t^2} F^{C}_{g-1}
+\sum_{r=1}^{g-1}\frac{\partial}{\partial t} F^{C}_r\cdot
\frac{\partial}{\partial t} F^{C}_{g-r}.
\ee
\end{Theorem}

The second equation in the above theorem
is exactly the same as the recursive formula for
the quantization of the Catalan curve \cite[(4.13)]{ms}.
Similar to the case of quantum Airy curve, the free energies $ F^{C}_g$ here
can also be constructed from the Catalan curve using the Eynard-Orantin topological recursion.

\subsection{Example: Enumeration of stable graphs with genus zero vertices}

In this subsection we apply topological 1D gravity and our formalism to study
the problem of enumeration of stable graphs with only genus zero vertices.
In this model,
all vertices of the Feynman graphs have valence $\geq 3$.
The partition function  is given by
\be \label{eqn:Z-st}
Z^{st}(T, \kappa)= \frac{\kappa^{1/2}}{(2\pi \lambda^2)^{1/2}}
\int dx \exp[\lambda^{-2}(T\sum_{n \geq 3} \frac{x^n}{n!} -\frac{x^2}{2\kappa})].
\ee
Write the free  energy $F^{st} = \log Z^{st}$ in the following expansion:
\be
F^{st}(T, \kappa)= \sum_{g\geq 2} \sum_{d=1}^{2g-2}a_d^g T^d\kappa^{g-1+d}.
\ee
By \eqref{eqn:Formula-for-Z}
\be\label{generating-a}
\begin{split}
F^{st}(T, \kappa)
=\log\biggl(\sum_{k\geq 0}\sum_{l_i\geq 3}\sum_l
\frac{\lambda^{2l-2k}}{k!\cdot l_1!\cdots l_k!}\cdot
(2l-1)!!\cdot T^k \cdot\kappa^l\cdot\delta_{l_1+\cdots+l_k,2l}\biggr).
\end{split}
\ee
The coefficients $a^g_d$ counts the number of connected stable graphs of genus $g$
with $d$ vertices of marked genus $0$,
weighted by the inverse of the orders of their automorphism groups.
Such graphs has $g-1+d$ edges.
Using this interpretation one gets the following explicit formulas for $a^g_d$, $d=1,2$:
\be\label{eg3-agk}
\begin{split}
a^g_1&=\frac{1}{(2g)!!} = \frac{1}{2^gg!},\\
a^g_2&=\frac{1}{2}\sum_{k=1}^{g+1}\sum_{
\substack{k+2l\geq 3\\2g+2-k-2l\geq 3}}
\frac{1}{k!\cdot(2l)!!\cdot(2g+2-2k-2l)!!}.
\end{split}
\ee
By an elementary calculation, one can simplify the second equality to get
\be
a^g_2=\frac{2^{2g}-2^{g-1}-(g+1)^2}{2^g\cdot (g+1)!}.
\ee
Since $a^g_{2g-2}$ counts trivalent graphs of genus $g$,
it equals to $b_{g-1}$ in \cite[(65), (66)]{dyz}, and so
\be
\sum_{g \geq 2} a^g_{2g-2} x^{g-1}
= \log \biggl(\sum_{m=0}^\infty \frac{(6m)!}{(3m)! (2m)!} \big(\frac{x}{288}\big)^m\biggr).
\ee
It is interesting to compute $a^g_d$ for $3 \leq d < 2g-2$.
Formula \eqref{generating-a} and the direct counting of graphs
are not very effective for this purpose.
We will apply the quadratic recursion relations developed in Section \ref{sec:1D-Rec}
to derive some recursion relations for $a^g_d$.

In \eqref{eqn:Q-Rec} and \eqref{eqn:d-X},
we take $I_k = T$ for all $k \geq 1$, and then $d_X$ changes
to an operator $D$ such that
\ben
D T=T,\qquad D\kappa=\kappa^2 T
\een
then we get from \eqref{eqn:Q-Rec} the following recursion relation:
\be \label{eqn:st-Rec}
\pd_\kappa F^{st}_g = \frac{1}{2} \big((D+\kappa T) DF^{st}_{g-1}+
\sum^{g-1}_{r=1} D F^{st}_r\cdot DF^{st}_{g-r} \big),
\ee
where we have to add the following convention:
\be
F_1^{st} = \half \log \kappa.
\ee
From this we recursively:
\ben
F_2^{st}&=&\frac{T\kappa^2}{8}+\frac{5T^2\kappa^3}{24},\\
F_3^{st}&=&\frac{T\kappa^3}{48}+\frac{11T^2\kappa^4}{48}+\frac{25T^3\kappa^5}{48}
+\frac{5T^4\kappa^6}{16},\\
F_4^{st}&=&\frac{T\kappa^4}{384}+\frac{223T^2\kappa^5}{1920}+\frac{515T^3\kappa^6}{576}
+\frac{1373T^4\kappa^7}{576}
+\frac{985T^5\kappa^8}{384}+\frac{1105T^6\kappa^9}{1152}.
\een
We have
\ben
&&DF_g^{st}=\sum_d  a_d^g d T^d \kappa^{g-1+d}
+\sum_d  a_d^g (g-1+d) T^{d+1} \kappa^{g+d},\\
&&D^2 F_g^{st}=\sum_d  a_d^g d^2 T^d \kappa^{g-1+d}
+\sum_d  a_d^g (g-1+d)d T^{d+1} \kappa^{g+d}\\
&&\;\;+\sum_d  a_d^g (g-1+d)(d+1) T^{d+1} \kappa^{g+d}
+\sum_d  a_d^g (g-1+d)(g+d) T^{d+2} \kappa^{g+d+1}.
\een
Then \eqref{eqn:st-Rec} gives us
\ben
&&2 \sum_{d}a_d^g (g-1+d) T^d \kappa^{g-2+d}\\
&=&(D+\kappa T)DF_{g-1}^{st}+\sum_{r=1}^{g-1}DF_r^{st}DF_{g-r}^{st}\\
&=&\sum_d  a_d^{g-1} d^2 T^d \kappa^{g-2+d}
+\sum_d  a_d^{g-1} (g-2+d)d T^{d+1} \kappa^{g-1+d}\\
&&+\sum_d  a_d^{g-1} (g-2+d)(d+1) T^{d+1} \kappa^{g-1+d}\\
&&+\sum_d  a_d^{g-1} (g-2+d)(g-1+d) T^{d+2} \kappa^{g+d}\\
&&+2\biggl[\sum_d  a_d^{g-1} d T^{d+1} \kappa^{g-1+d}+
\sum_d  a_d^{g-1} (g-2+d) T^{d+2} \kappa^{g+d}\biggr]\\
&&+\sum_{r=2}^{g-2}
\biggl(\sum_d  a_d^r d T^d \kappa^{r-1+d}
+\sum_d  a_d^r (r-1+d) T^{d+1} \kappa^{r+d}\biggr)\cdot\\
&&\qquad\quad\biggl(\sum_d  a_d^{g-r} d T^d \kappa^{g-r-1+d}
+\sum_d  a_d^{g-r} (g-r-1+d) T^{d+1} \kappa^{g-r+d}\biggr)
\een
for $g\geq 3$. Therefore the following quadratic recursion relations hold for $a^g_d$:
\ben
a_d^g& =&\frac{1}{2(g-1+d)}\biggl\{
\sum_d  a_d^{g-1} d^2 +\sum_d  a_{d-2}^{g-1} (g-4+d)(g-1+d) \\
&+&\sum_d  a_{d-1}^{g-1} \big[(g-3+d)(2d-1)+2(d-1)\big] \\
&+&\sum_{r=2}^{g-2}\biggl(
\sum_{d_1+d_2=d}a_{d_1}^r a_{d_2}^{g-r}d_1 d_2
+\sum_{d_1+d_2=d-1}a_{d_1}^r a_{d_2}^{g-r}(r-1+d_1) d_2\\
&+&\sum_{d_1+d_2=d-1}a_{d_1}^r a_{d_2}^{g-r}d_1(g-r-1+ d_2)\\
&+&\sum_{d_1+d_2=d-2}a_{d_1}^r a_{d_2}^{g-r}(r-1+d_1) (g-r-1+ d_2) \biggr)\biggr\}.
\een

\subsection{Example: Enumeration of graphs}

In this subsection we study another example:
the model describing enumeration of graphs,
not necessarily stable,
introduced and studied in \cite{dyz}.
The partition function for this case is
\be \label{eqn:gr}
Z(t)=  \frac{1}{\sqrt{2\pi \lambda^2}}
\int d\eta \exp\big[\lambda^{-2}(t\cdot e^\eta-\frac{\eta^2}{2})\big].
\ee
The critical point $\eta=T$ of the function $f(\eta)=t\cdot e^\eta-\frac{\eta^2}{2}$ is given by
\be \label{eqn:Lambert}
T=t\cdot e^T,
\ee
whose solution may be explicitly given by the Lambert series \cite{dyz}:
\be
T=\sum_{d=1}^{\infty}\frac{d^{d-1}}{d!}t^d.
\ee
Write $x=\eta-T$, then the expansion $f(\eta)$ at $T$ is
\ben
f(\eta) & = & (te^T-\frac{T^2}{2})
+\frac{te^T-1}{2}x^2+\sum_{n=3}^{\infty}\frac{te^T}{n!}x^n \\
& = & (T-\frac{T^2}{2}) +\frac{T-1}{2}x^2+\sum_{n=3}^{\infty}\frac{T}{n!}x^n.
\een
Therefore in this case, the propagator is
\be\label{eg3-kappa}
\kappa=\frac{1}{1-T},
\ee
and the point on the space of coupling constants is specified by
\be\label{eg3-I}
I_0=T; \quad I_k=T, \quad k\geq 1,
\ee
or in the coordinates $\{t_k\}$,
\be
t_k=t, \quad k\geq 0.
\ee
By \eqref{eqn:F-1D},
the free energy is
\be
 F^{gr}_g=\sum_{\Gamma\in\cG_{g,0}^{c,0}}\frac{1}{|\Aut(\Gamma)|}\cdot
\frac{T^{|V(\Gamma)|}}{(1-T)^{|E(\Gamma)|}},
\ee
For example,
\ben
F^{gr}_0&=&T-\frac{T^2}{2},\quad F^{1D}_1=\frac{1}{2}\log\frac{1}{1-T},\\
F^{gr}_2&=&\frac{T}{8(1-T)^2}+\frac{5T^2}{24(1-T)^3},\\
F^{gr}_3&=&\frac{T}{48(1-T)^3}+\frac{11T^2}{48(1-T)^4}+\frac{25T^3}{48(1-T)^5}
+\frac{5T^4}{16(1-T)^6},\\
F^{gr}_4&=&\frac{T}{384(1-T)^4}+\frac{223T^2}{1920(1-T)^5}
+\frac{515T^3}{576(1-T)^6}+\frac{1373T^4}{576(1-T)^7}\\
&+&\frac{985T^5}{384(1-T)^8}+\frac{1105T^6}{1152(1-T)^9}.
\een
These match with $\cG_k$ in \cite[(63)]{dyz}.

By comparing \eqref{eqn:gr} with \eqref{eqn:Z-st},
it is clear that for $g \geq 2$,
\be
F_g^{gr} = F^{st}_g(T, \frac{1}{1-T}).
\ee
Therefore,
for $g\geq 2$ the free energy $F_g^{gr}$ can be written in the form
\ben
F_g^{gr}=\sum_{d=1}^{2g-2}\frac{a^g_d\cdot T^d}{(1-T)^{d+g-1}}.
\een
Moreover, note that for $g \geq 2$, $F_g^{gr}$ is of the form
\be
F_g^{gr}=\frac{f_g(T)}{(1-T)^{3g-3}},
\ee
where $f_g(T)$ is a polynomial of degree $\leq 2g-2$ without constant term.
For example,
\ben
F^{gr}_2&=&\frac{T}{24(1-T)^3}(3+2T),\\
F^{gr}_3&=&\frac{T}{48(1-T)^6}(1+8T+6T^2),\\
F^{gr}_4&=&\frac{T}{5760(1-T)^9}(15+594T+2624T^2+2144T^3+164T^4-16T^5),\\
F^{gr}_5&=&\frac{T}{11520(1-T)^{12}}(3+465T+6730T^2+21940T^3+18940T^4\\
&+& 3012T^5-240T^6),\\
F^{gr}_6&=&\frac{T}{2903040(1-T)^{15}}(63+35568T+1349298T^2+11582816T^3+31178616T^4\\
&&+27897072T^5+6526912T^6-266448T^7-36576T^8+2304T^9).
\een
Let
\ben
f_g(T)=\sum_{n=1}^{2g-2} b_n^g\cdot T^n,
\een
then the two type of coefficients
$\{a_k^g\}$ and $\{b_k^g\}$ are related by
\be\label{relation-ab}
b_k^g=\sum_{l=1}^{k}(-1)^{k+l}\binom{2g-2-l}{k-l}a_l^g.
\ee
In particular,
\be
\begin{split}
b_1^g=&a_1^g=\frac{1}{(2g)!!},\\
b_2^g=&a_2^g-(2g-3)a_1^g\\
=&\frac{2^{2g}-2^{g-1}-(g+1)^2}{2^g\cdot (g+1)!}-\frac{2g-3}{(2g)!!}.
\end{split}
\ee

As pointed out in \cite[(64), (65)]{dyz}, $F_g^{gr}$ is of the form
$F_g^{gr}=\sum_{n=g-1}^{3g-3}\lambda_{g,n}\kappa^n$, and
\ben
\lambda_{g,g-1}=\frac{B_g}{g(g-1)},\quad \lambda_{g,3g-3}=b_{g-1},
\een
where $B_g$ are the Bernoulli numbers, and $b_g$ are given by
\ben
\sum_{n=1}^{\infty}b_r x^r=\log\biggl(
\sum_{m=0}^{\infty}\frac{(6m)!}{(3m)!(2m)!}
\big(\frac{x}{288}\big)^m\biggr).
\een
The numbers $\{b_g\}$ are exactly the coefficients appearing in the expansion of
the Airy function \eqref{airy-S}.
Here are some examples of $F_g^{gr}$ written as polynomials in $\kappa$:
\ben
F^{gr}_2&=&\frac {5} {24}\kappa^3 - \frac {7} {24}\kappa^2 + \frac {1} {12}\kappa,\\
F^{gr}_3&=&\frac{5}{16} \kappa^6 - \frac{35}{48} \kappa^5 + \frac{13}{24} \kappa^4 - \frac{1}{8} \kappa^3,\\
F^{gr}_4&=&\frac{1105}{1152} \kappa^9 - \frac{1225}{384} \kappa^8 + \frac{2273}{576 }\kappa^7 - \frac{313}{
 144} \kappa^6 + \frac{227}{480} \kappa^5 - \frac{17}{1440} \kappa^4 - \frac{1}{360} \kappa^3,\\
F^{gr}_5&=&\frac{565}{128} \kappa^{12} -\frac{ 14665}{768} \kappa^{11} +\frac{ 76367}{2304} \kappa^{10} -\frac{ 11191}{
 384} \kappa^9 + \frac{2557}{192} \kappa^8 - \frac{7993}{2880} \kappa^7 \\
 &+ &\frac{37}{320} \kappa^6 +\frac{ 1}{
 48} \kappa^5,\\
F^{gr}_6&=&\frac{82825}{3072} \kappa^{15} - \frac{441245}{3072} \kappa^{14} +\frac{ 493235}{
 1536} \kappa^{13} - \frac{16116187}{41472} \kappa^{12} +\frac{ 2827135}{10368} \kappa^{11} \\
 &- &\frac{1884983}{
 17280} \kappa^{10} +\frac{ 567289}{25920} \kappa^9 - \frac{7489}{6480} \kappa^8 - \frac{10249}{
 60480} \kappa^7 + \frac{47}{10080} \kappa^6 + \frac{1}{1260} \kappa^5,\\
F^{gr}_7&=&\frac{19675}{96} \kappa^{18} - \frac{7969325}{6144} \kappa^{17} +\frac{ 65405005}{
 18432} \kappa^{16} -\frac{ 453853985}{82944} \kappa^{15} \\
 &+& \frac{ 215237149}{41472} \kappa^{14}
  -\frac{ 64035527}{20736} \kappa^{13} +
 \frac{23126555}{20736} \kappa^{12} - \frac{11204309}{51840} \kappa^{11} \\
& + & \frac{1352989}{103680} \kappa^{10} +\frac{ 6481}{4032} \kappa^9
 - \frac{1927}{20160} \kappa^8 - \frac{1}{72} \kappa^7.
\een
Plug $T=1-\frac{1}{\kappa}$ into \eqref{generating-a},
we get the generating series for $\{\lambda_{g,k}\}$:
\be
\begin{split}
&\sum_{g=2}^\infty \lambda^{2g-2}
\sum_{d=g-1}^{3g-3}\lambda_{g,d} \kappa^d\\
=&\log\biggl(\sum_{k\geq 0}\sum_{l_i\geq 3}\sum_l
\frac{\lambda^{2l-2k}}{k!\cdot l_1!\cdots l_k!}\cdot
(2l-1)!!\cdot (\kappa-1)^k \kappa^{l-k}\cdot\delta_{l_1+\cdots+l_k-2l,0}\biggr).
\end{split}
\ee

The relations among $\{a_k^g\}$, $\{b_k^g\}$ and $\{\lambda_{g,k}\}$
are given by \eqref{relation-ab} and
\ben
&& \lambda_{g,3g-3-k}=(-1)^k\sum_{l=k}^{2g-2}\binom{l}{k}b_l^g, \\
&& b^g_n  = (-1)^n \sum_{j=g-1}^{3g-3}\binom{3g-3-j}{n} \lambda_{g,j}, \\
&& \lambda_{g,3g-3-k}=\sum_{l=0}^{k}(-1)^{k+l}\binom{2g-2-l}{k-l}a_{2g-2-l}.
\een

Now we apply the quadratic recursion relations for topological 1D
developed in Section \ref{sec:1D-Rec} to this problem.
For this purpose,
we need to find an operator that satisfies the conditions in \eqref{eqn:d-X}.
We can take $d_X$ to be $D=\frac{\partial}{\partial(\log T)}=T\frac{\partial}{\partial T}$,
because we have
\begin{align*}
I_k &= T, &\kappa &= \frac{1}{1-T},
\end{align*}
the conditions in \eqref{eqn:d-X} are satisfied by $D$.
So we get the following equation from \eqref{eqn:Q-Rec}:
\be \label{eqn:gr-Rec}
\pd_\kappa F^{gr} = \frac{1}{2}((D+\kappa T)DF_{g-1}^{gr}+
\sum^{g-1}_{r=1} DF_r^{gr}DF_{g-r}^{gr}),
\ee

\begin{Lemma}\label{lemma-eg3}
For $g\geq 2$, we have
\be
(\kappa+\kappa^2 T)\partial_\kappa F^{gr}_g-D F^{gr}_g=(g-1) F^{gr}_g.
\ee
\end{Lemma}
\begin{proof}
Every monomial in $ F^{gr}_g$ for $g\geq 2$ is of the form $a_{g,k} T^k\kappa^{k+g-1}$.
Then the lemma follows from
\ben
\kappa\partial_\kappa(T^k\kappa^{k+g-1})=(k+g-1)T^k\kappa^{k+g-1}
\een
and
\ben
D(T^k\kappa^{k+g-1})=kT^k\kappa^{k+g-1}+(k+g-1)T^{k}\kappa^{k+g-2}\cdot\kappa^2 T.
\een
\end{proof}

Since $\kappa+\kappa^2 T=\frac{1}{(1-T)^2}$,
\eqref{eqn:gr-Rec} and Lemma \ref{lemma-eg3} give us:

\begin{Theorem}
For $g\geq 2$, we have
\be\label{eg3-eqthm}
(D+g-1) F^{gr}_g=\frac{1}{2(1-T)^2}[(D+\frac{T}{1-T})D F^{gr}_{g-1}
+\sum_{r=1}^{g-1}(D F^{gr}_r)(D F^{gr}_{g-r})].
\ee
\end{Theorem}

This theorem provides a way to compute $F^{gr}_g$ recursively.
In fact, the left-hand side of \eqref{eg3-eqthm} can be rewritten as
\be
\begin{split}
(D+g-1) F^{gr}_g=&T\frac{d}{dT}F^{gr} _g  (T)+(g-1)F^{gr}_g (T)\\
=&T^{2-g}\frac{d}{dT}(T^{g-1}\cdot F^{gr}_g (T)),
\end{split}
\ee
thus $F^{gr}_g$ is determined by the formula
\be\label{rec-integral}
F_g^{gr}(T)=T^{1-g}\cdot\int_0^T \frac{T^{g-2}}{2(1-T)^2}\biggl[(D+\frac{T}{1-T})D F^{gr}_{g-1}
+\sum_{r=1}^{g-1}(D F^{gr}_r)(D F^{gr}_{g-r})\biggr]dT.
\ee

Moreover, for $g \geq 2$, when $F_g^{gr}(T)$ is written as a polynomial
in the propagator $\kappa=(1-T)^{-1}$ using $T=1-\frac{1}{\kappa}$,
the recursion relation \eqref{eg3-eqthm} can be written as follows.

\begin{Theorem}\label{eg3-kappa-rec}
For $g\geq 2$, we have
\be
\begin{split}
\big[(\kappa^2-\kappa)\frac{d}{d\kappa}+g-1\big]F_g^{gr}=&\frac{\kappa^2}{2}\biggl[
(\kappa^2-\kappa)^2\frac{d^2F_{g-1}^{gr}}{d\kappa^2}
+(3\kappa-2)(\kappa^2-\kappa)\frac{dF_{g-1}^{gr}}{d\kappa}\\
&+(\kappa^2-\kappa)^2\sum_{r=1}^{g-1}\frac{dF_{r}^{gr}}{d\kappa}
\cdot\frac{dF_{g-r}^{gr}}{d\kappa}\biggr].
\end{split}
\ee
In particular, for $g\geq 3$,
\be\label{kappa-rec}
\begin{split}
\big[(\kappa^2-\kappa)\frac{d}{d\kappa}+g-1\big]F_g^{gr}=&\frac{\kappa^2}{2}\biggl[
(\kappa^2-\kappa)^2\frac{d^2F_{g-1}^{gr}}{d\kappa^2}
+(4\kappa-3)(\kappa^2-\kappa)\frac{dF_{g-1}^{gr}}{d\kappa}\\
&+(\kappa^2-\kappa)^2\sum_{r=2}^{g-2}\frac{dF_{r}^{gr}}{d\kappa}\cdot
\frac{dF_{g-r}^{gr}}{d\kappa}\biggr].
\end{split}
\ee
\end{Theorem}

The equation \eqref{kappa-rec} can be rewritten as
\be
\begin{split}
&\frac{d}{d\kappa} \biggl( (1-\kappa)^{g-1}\frac{F_g^{gr}}{\kappa^{g-1}}\biggr)\\
=&-\frac{1}{\kappa^2} \big( \frac{1}{\kappa} -1 \big)^{g-2}
\biggl[ (g-1)F_g^{gr}+(\kappa^2-\kappa)\frac{dF_g^{gr}}{d \kappa} \biggr]\\
=&-\frac{1}{2}\frac{(1-\kappa)^{g-2}}{\kappa^{g-2}}\biggl[
(\kappa^2-\kappa)^2\frac{d^2F_{g-1}^{gr}}{d\kappa^2}+
(4\kappa-3)(\kappa^2-\kappa)\frac{dF_{g-1}^{gr}}{d\kappa}\\
&\qquad\qquad +(\kappa^2-\kappa)^2\sum_{r=2}^{g-2}\frac{dF_{r}^{gr}}{d\kappa}
\cdot\frac{dF_{g-r}^{gr}}{d\kappa}\biggr].
\end{split}
\ee
Now if we want to integrate this equality with respect to $\kappa$,
we need to determine the constant of integration.
Note that $(1-\kappa)^{g-1}\frac{F_g^{gr}}{\kappa^{g-1}}$ is a polynomial in $\kappa$
whose constant term is $\lambda_{g,g-1}=\frac{B_g}{g(g-1)}$, therefore
\be
\begin{split}
F_g^{gr}=&\frac{\kappa^{g-1}}{(1-\kappa)^{g-1}}\biggl\{ \frac{B_g}{g(g-1)}+ \int_0^\kappa\biggl[-\frac{1}{2}\frac{(1-\kappa)^{g-2}}{\kappa^{g-2}}\biggl(
(\kappa^2-\kappa)^2\frac{d^2F_{g-1}^{gr}}{d\kappa^2}+\\
&(4\kappa-3)(\kappa^2-\kappa)\frac{dF_{g-1}^{gr}}{d\kappa}
+(\kappa^2-\kappa)^2\sum_{r=2}^{g-2}\frac{dF_{r}^{gr}}{d\kappa}
\cdot\frac{dF_{g-r}^{gr}}{d\kappa}\biggr)\biggr]d\kappa\biggr\}.
\end{split}
\ee

Now let us derive some recursion relations for the coefficients $\{\lambda_{g,k}\}$
using Theorem \ref{eg3-kappa-rec}.
Expand the two sides of \eqref{kappa-rec}, we get
\ben
&&\sum_l \big[l(\kappa-1)+(g-1)\big]\lambda_{g,l}\kappa^l\\
&=&\frac{\kappa^2}{2}\biggr[\sum_l(\kappa-1)^2 l(l-1)\lambda_{g-1,l}\kappa^l
+\sum_l (4\kappa-3)(\kappa-1)l\lambda_{g-1,l}\kappa^l\\
&+&\sum_{r=2}^{g-2}\sum_{l_1,l_2}(\kappa-1)^2
l_1 l_2\lambda_{r,l_1}\lambda_{g-r,l_2}\kappa^{l_1+l_2}\biggr],
\een
which gives us
\be\label{lambda-rec}
\begin{split}
&(l-1)\lambda_{g,l-1}+(g-l-1)\lambda_{g,l}\\
=&\frac{1}{2}(l-4)(l-1)\lambda_{g-1,l-4}-\frac{1}{2}(l-3)(2l-1)\lambda_{g-1,l-3}
+\frac{1}{2}(l-2)l\lambda_{g-1,l-2}\\
&+\frac{1}{2}\sum_{r=2}^{g-2}\biggl(\sum_{l_1+l_2=l-4}l_1l_2\lambda_{r,l_1}\lambda_{g-r,l_2}
-2\sum_{l_1+l_2=l-3}l_1l_2\lambda_{r,l_1}\lambda_{g-r,l_2}\\
&+\sum_{l_1+l_2=l-2}l_1l_2\lambda_{r,l_1}\lambda_{g-r,l_2}\biggr).
\end{split}
\ee
Note that $\lambda_{g,n}=0$ unless $g-1 \leq n \leq 3g-3$,
thus the above relation indeed determines $\{\lambda_{g,n}\}$ uniquely.
In fact, first we set $l=3g-2$, then \eqref{lambda-rec} gives
\ben
\lambda_{g,3g-3}
&=&\frac{1}{2}(3g-6)\lambda_{g-1,3g-6}-\frac{(3g-5)(6g-5)}{2(3g-3)}\lambda_{g-1,3g-5}
\\
&+&\frac{(3g-4)(3g-2)}{2(3g-3)}\lambda_{g-1,3g-4}+\frac{1}{2(3g-3)}\sum_{r=2}^{g-2}\biggl(\sum_{l_1+l_2=3g-6}l_1l_2\lambda_{r,l_1}\lambda_{g-r,l_2}
\\
&-&2\sum_{l_1+l_2=3g-5}l_1l_2\lambda_{r,l_1}\lambda_{g-r,l_2}+\sum_{l_1+l_2=3g-4}l_1l_2\lambda_{r,l_1}\lambda_{g-r,l_2}\biggr).
\een
Then set $l=3g-3$, \eqref{lambda-rec} gives
\ben
&&(3g-4)\lambda_{g,3g-4}+(2-2g)\lambda_{g,3g-3}\\
&=&\frac{1}{2}(3g-7)(3g-4)\lambda_{g-1,3g-7}-\frac{1}{2}(3g-6)(6g-7)\lambda_{g-1,3g-6}
\\
&+&\frac{1}{2}(3g-5)l\lambda_{g-1,3g-5}+\frac{1}{2}\sum_{r=2}^{g-2}\biggl(\sum_{l_1+l_2=3g-7}l_1l_2\lambda_{r,l_1}\lambda_{g-r,l_2}
\\
&-&2\sum_{l_1+l_2=3g-6}l_1l_2\lambda_{r,l_1}\lambda_{g-r,l_2}+\sum_{l_1+l_2=3g-5}l_1l_2\lambda_{r,l_1}\lambda_{g-r,l_2}\biggr),
\een
from which we can solve $\lambda_{g,3g-4}$ using $\lambda_{g,3g-3}$ and lower genus $\lambda_{r,n}$. Similarly $\lambda_{g,n}$
can be solved recursively from $n=3g-3$ to $n=g-1$ using \eqref{lambda-rec}.

\vspace{.2in}

{\em Acknowledgements}.
We thank an anonymous referee for helpful suggestions that improve the presentation of this paper. 
The second author is partly supported by NSFC grant 11661131005.

\newpage
\begin{appendices}

\section{Some Explicit Expressions for $\widehat{F}_{g,n}$}\label{app1}

In this appendix, we give more examples of the abstract free energies.

For the case $N=1$, here are the explicit expressions for $\wF_{2,2}$ and $\wF_3$ in terms of Feynman graphs.

\begin{flalign*}
\begin{split}
&
.
\end{split}&&
\end{flalign*}

And in terms of the derivatives of
the holomorphic free energy $F(t)$ together with the propagator $\kappa$,
the explicit expressions for $\wF_3$ and $\wF_4$ are as follows.
\begin{flalign*}
\begin{split}
\wF_3=F_3&+(\frac{1}{2}F_2''+F_1'F_2')\kappa\\
+&[\frac{1}{8}F_1^{(4)}+\frac{1}{4}(F_1'')^2+\frac{1}{2}F_0'''F_2'+\frac{1}{2}F_1'F_1'''+\frac{1}{2}(F_1')^2F_1'']\kappa^2\\
+&[\frac{1}{48}F_0^{(6)}+\frac{1}{4}F_0^{(4)}F_1''+\frac{5}{12}F_0'''F_1'''+\frac{1}{8}F_0^{(5)}F_1'+F_0'''F_1'F_1''\\
&+\frac{1}{4}F_0^{(4)}(F_1')^2+\frac{1}{6}F_0'''(F_1')^3]\kappa^3\\
+&[\frac{1}{12}(F_0^{(4)})^2+\frac{7}{48}F_0'''F_0^{(5)}+\frac{5}{8}(F_0''')^2F_1''+\frac{2}{3}F_0'''F_0^{(4)}F_1'+\frac{1}{2}(F_0''')^2(F_1')^2]\kappa^4\\
+&[\frac{25}{48}(F_0''')^2F_0^{(4)}+\frac{5}{8}(F_0''')^3F_1']\kappa^5+\frac{5}{16}(F_0''')^4\kappa^6.
\end{split}&&
\end{flalign*}

\begin{flalign*}
\begin{split}
\wF_4=F_4&+[\frac{1}{2}(F_2')^2+F_1'F_3'+\frac{1}{2}F_3'']\kappa\\
+&[F_1'F_1''F_2'+\frac{1}{2}(F_1')^2F_2''+\frac{1}{2}F_1''F_2''+\frac{1}{2}F_0'''F_3'+\frac{1}{2}F_1'''F_2'+\frac{1}{2}F_1'F_2'''\\
&+\frac{1}{8}F_2^{(4)}]\kappa^2\\
+&[\frac{1}{2}(F_1')^2(F_1'')^2+\frac{1}{6}(F_1'')^3+\frac{1}{2}F_0'''(F_1')^2F_2'+F_0'''F_1''F_2'+F_0'''F_1'F_2''\\
&+\frac{1}{6}(F_1')^3F_1'''+F_1'F_1''F_1'''+\frac{5}{24}(F_1''')^2+\frac{5}{12}F_0'''F_2'''+\frac{1}{2}F_0^{(4)}F_1'F_2'\\
&+\frac{1}{4}F_0^{(4)}F_2''+\frac{1}{4}(F_1')^2F_1^{(4)}+\frac{1}{4}F_1''F_1^{(4)}+\frac{1}{8}F_0^{(5)}F_2'+\frac{1}{8}F_1^{(5)}F_1'+\frac{1}{48}F_1^{(6)}]\kappa^3\\
+&[\frac{1}{2}F_0'''(F_1')^3F_1''+\frac{3}{2}F_0'''F_1'(F_1'')^2+(F_0''')^2F_1'F_2'+\frac{5}{8}(F_0''')^2F_2''+F_0'''(F_1')^2F_1'''\\
&+\frac{5}{4}F_0'''F_1''F_1'''+\frac{1}{24}F_0^{(4)}(F_1')^4+\frac{3}{4}F_0^{(4)}(F_1')^2F_1''+\frac{3}{8}F_0^{(4)}(F_1'')^2\\
&+\frac{2}{3}F_0'''F_0^{(4)}F_2'+\frac{2}{3}F_0^{(4)}F_1'F_1'''+\frac{2}{3}F_0'''F_1'F_1^{(4)}+\frac{1}{6}F_0^{(4)}F_1^{(4)}+\frac{1}{12}F_0^{(5)}(F_1')^3\\
&+\frac{3}{8}F_0^{(5)}F_1'F_1''+\frac{7}{48}F_0^{(5)}F_1'''+\frac{7}{48}F_0'''F_1^{(5)}+\frac{1}{16}F_0^{(6)}(F_1')^2+\frac{1}{16}F_0^{(6)}F_1''\\
&+\frac{1}{48}F_0^{(7)}F_1'+\frac{1}{384}F_0^{(8)}]\kappa^4\\
+&[\frac{1}{8}(F_0''')^2(F_1)^4+2(F_0''')^2(F_1')^2F_1''+\frac{5}{4}(F_0''')^2(F_1'')^2+\frac{5}{8}(F_0''')^3F_2'\\
&+\frac{15}{8}(F_0''')^2F_1'F_1'''+\frac{7}{12}F_0'''F_0^{(4)}(F_1')^3+\frac{8}{3}F_0'''F_0^{(4)}F_1'F_1''+\frac{25}{24}F_0'''F_0^{(4)}F_1'''\\
&+\frac{1}{3}(F_0^{(4)})^2(F_1')^2+\frac{1}{3}(F_0^{(4)})^2F_1''+\frac{25}{48}(F_0''')^2F_1^{(4)}+\frac{25}{48}F_0'''F_0^{(5)}(F_1')^2\\
&+\frac{7}{12}F_0'''F_0^{(5)}F_1''+\frac{5}{16}F_0^{(4)}F_0^{(5)}F_1'+\frac{21}{640}(F_0^{(5)})^2+\frac{5}{24}F_0'''F_0^{(6)}F_1'\\
&+\frac{5}{96}F_0^{(4)}F_0^{(6)}+\frac{1}{32}F_0'''F_0^{(7)}]\kappa^5\\
+&[\frac{2}{3}(F_0''')^3(F_1')^3+\frac{25}{8}(F_0''')^3F_1'F_1''+\frac{5}{4}(F_0''')^3F_1'''+\frac{109}{48}(F_0''')^2F_0^{(4)}(F_1')^2\\
&+\frac{125}{48}(F_0''')^2F_0^{(4)}F_1''+\frac{11}{8}F_0'''(F_0^{(4)})^2F_1'+\frac{11}{96}(F_0^{(4)})^3+\frac{53}{48}(F_0''')^2F_0^{(5)}F_1'\\
&+\frac{7}{12}F_0'''F_0^{(4)}F_0^{(5)}+\frac{113}{576}(F_0''')^2F_0^{(6)}]\kappa^6\\
+&[\frac{25}{16}(F_0''')^4(F_1')^2+\frac{15}{8}(F_0''')^4F_1''+\frac{185}{48}(F_0''')^3F_0^{(4)}F_1'+\frac{445}{288}(F_0''')^2(F_0^{(4)})^2\\
&+\frac{161}{192}(F_0''')^3F_0^{(5)}]\kappa^7\\
+&[\frac{15}{8}(F_0''')^5F_1'+\frac{985}{384}(F_0''')^4F_0^{(4)}]\kappa^8+\frac{1105}{1152}(F_0''')^6\kappa^9.
\end{split}&&
\end{flalign*}

In particular, by setting $F_g=0$ for all $g\geq 1$ and $F_0^{(k+1)}=I_k$,
we recover the expressions \eqref{1d-f2f3} and \eqref{1d-f4} for $F_3^{1D}$ and $F_4^{1D}$.

\end{appendices}

\end{document}